\documentclass{article}
\usepackage{amssymb,latexsym}
\usepackage{color,graphicx,itemlist}
\usepackage{epic,eepic}
\usepackage{showlabels}

\newcommand{\dwedge}[2]%
{\overset{#1}{\underset{#2}{\mbox{$\bigwedge\hspace{-2ex}\bigwedge$}}}}
\newcommand{\dvee}[2]%
{\overset{#1}{\underset{#2}{\mbox{$\bigvee\hspace{-2ex}\bigvee$}}}}

\newcommand{\inMath}[1]{\relax\ifmmode{#1}%
\else{\mbox{$#1$}}\fi}

\newcommand{\vdot}%
{\unitlength0.4mm
\begin{picture}(8,10)
\put(4,0){.}
\put(0,10){.}
\put(8,10){.}
\put(2,5){.}
\put(6,5){.}
\end{picture}
}

\newcommand{\ignore}[1]{}

\newcommand{\romanref}[1]{%
\if \ref{#1}\empty {\setcounter{romanrefcounter}0} \else 
{\setcounter{romanrefcounter}{\ref{#1}}}\fi%
{\it \roman{romanrefcounter}}}

\newcommand{\CA}{{\cal A}}

\newcommand{\CB}{\mathcal{B}}
\newcommand{\BC}{{\bf C}}

\newcommand{\CE}{\mbox{$\mathcal{E}$}}
\newcommand{\BE}{{\bf E}}
\newcommand{\Be}{{\bf e}}
\newcommand{\BF}{{\bf F}}

\newcommand{\Bf}{{\bf f}}

\newcommand{\BL}{{{\bf L}}}
\newcommand{\CL}{{\mathcal{L}}}
\newcommand{\BBL}{\mbox{\(\Bbb L\)}}

\newcommand{\Bm}{{\bf m}}
\newcommand{\BN}{{\bf N}}
\newcommand{\Bn}{{\bf n}}

\newcommand{\BBP}{\mbox{\(\Bbb P\)}}

\newcommand{\BBR}{\mbox{\(\Bbb R\)}}
\newcommand{\CS}{{\mathcal{S}}}
\newcommand{\BS}{{\bf S}}
\newcommand{\Bs}{{\bf s}}

\newcommand{\BT}{{\bf T}}
\newcommand{\CT}{\mathcal{T}}

\newcommand{\Bt}{{\bf t}}

\newcommand{\Bu}{{\bf u}}

\newcommand{\tO}{\twoheadrightarrow}

\newlength{\circlen}
\newlength{\symblen}

\ifx\text\undefined
  \newcommand{\text}[1]{\relax
    \ifmmode\mathchoice
      {\hbox{\the\textfont0\relax#1}}%
      {\hbox{\the\textfont0\relax#1}}%
      {\hbox{\the\scriptfont0\relax#1}}%
      {\hbox{\the\scriptscriptfont0\relax#1}}%
    \else{\relax#1}\fi}
  \fi

\newcommand{\defsub}[1]{}

\newcommand{\To}{\Rightarrow}

\newcommand{\nin}{{\not\in}}

\def\twoheaddownarrow{\rlap{$\downarrow$}\raise-.5ex\hbox{$\downarrow$}}
\def\twoheaduparrow{\rlap{$\uparrow$}\raise.5ex\hbox{$\uparrow$}}

\def\texturespicture #1 by #2 (#3){
\vbox to #2 {\hrule width #1 height 0pt depth 0pt
}}

\def\scaledpicture #1 by #2 (#3 scaled #4){{\dimen0=#1 \dimen1=#2
\divide\dimen0 by 1000 \multiply \dimen0 by #4
\divide\dimen1 by 1000 \multiply \dimen1 by #4
\texturespicture\dimen0 by \dimen1 (#3 scaled #4)}}

 {\begin{tabular}[t]{llr} && \fbox{#1}\\}%
 {\end{tabular}}

{\end{list}}

{\end{tabular}%
\end{trivlist}}

\def\introrule{{\cal I}}\def\elimrule{{\cal E}}

\newdimen\proofrulebreadth \proofrulebreadth=.05em
\newdimen\proofdotseparation \proofdotseparation=1.25ex
\newdimen\proofrulebaseline \proofrulebaseline=2ex
\newcount\proofdotnumber \proofdotnumber=3
\let\then\relax
\def\hfi{\hskip0pt plus.0001fil}
\mathchardef\squigto="3A3B
\newif\ifinsideprooftree\insideprooftreefalse
\newif\ifonleftofproofrule\onleftofproofrulefalse
\newif\ifproofdots\proofdotsfalse
\newif\ifdoubleproof\doubleprooffalse
\let\wereinproofbit\relax
\newdimen\shortenproofleft
\newdimen\shortenproofright
\newdimen\proofbelowshift
\newbox\proofabove
\newbox\proofbelow
\newbox\proofrulename
\def\shiftproofbelow{\let\next\relax\afterassignment\setshiftproofbelow\dimen0 }
\def\shiftproofbelowneg{\def\next{\multiply\dimen0 by-1 }%
\afterassignment\setshiftproofbelow\dimen0 }
\def\setshiftproofbelow{\next\proofbelowshift=\dimen0 }
\def\setproofrulebreadth{\proofrulebreadth}

\def\prooftree{
\ifnum  \lastpenalty=1
\then   \unpenalty
\else   \onleftofproofrulefalse
\fi
\ifonleftofproofrule
\else   \ifinsideprooftree
        \then   \hskip.5em plus1fil
        \fi
\fi
\bgroup
\setbox\proofbelow=\hbox{}\setbox\proofrulename=\hbox{}%
\let\justifies\proofover\let\leadsto\proofoverdots\let\Justifies\proofoverdbl
\let\using\proofusing\let\[\prooftree
\ifinsideprooftree\let\]\endprooftree\fi
\proofdotsfalse\doubleprooffalse
\let\thickness\setproofrulebreadth
\let\shiftright\shiftproofbelow \let\shift\shiftproofbelow
\let\shiftleft\shiftproofbelowneg
\let\ifwasinsideprooftree\ifinsideprooftree
\insideprooftreetrue
\setbox\proofabove=\hbox\bgroup$\displaystyle 
\let\wereinproofbit\prooftree
\shortenproofleft=0pt \shortenproofright=0pt \proofbelowshift=0pt
\onleftofproofruletrue\penalty1
}

\def\eproofbit{
\ifx    \wereinproofbit\prooftree
\then   \ifcase \lastpenalty
        \then   \shortenproofright=0pt  
        \or     \unpenalty\hfil         
        \or     \unpenalty\unskip       
        \else   \shortenproofright=0pt  
        \fi
\fi
\global\dimen0=\shortenproofleft
\global\dimen1=\shortenproofright
\global\dimen2=\proofrulebreadth
\global\dimen3=\proofbelowshift
\global\dimen4=\proofdotseparation
\global\count255=\proofdotnumber
$\egroup  
\shortenproofleft=\dimen0
\shortenproofright=\dimen1
\proofrulebreadth=\dimen2
\proofbelowshift=\dimen3
\proofdotseparation=\dimen4
\proofdotnumber=\count255
}

\def\proofover{
\eproofbit 
\setbox\proofbelow=\hbox\bgroup 
\let\wereinproofbit\proofover
$\displaystyle
}%
\def\proofoverdbl{
\eproofbit 
\doubleprooftrue
\setbox\proofbelow=\hbox\bgroup 
\let\wereinproofbit\proofoverdbl
$\displaystyle
}%
\def\proofoverdots{
\eproofbit 
\proofdotstrue
\setbox\proofbelow=\hbox\bgroup 
\let\wereinproofbit\proofoverdots
$\displaystyle
}%
\def\proofusing{
\eproofbit 
\setbox\proofrulename=\hbox\bgroup 
\let\wereinproofbit\proofusing
\kern0.3em$
}

\def\endprooftree{
\eproofbit 
  \dimen5 =0pt
\dimen0=\wd\proofabove \advance\dimen0-\shortenproofleft
\advance\dimen0-\shortenproofright
\dimen1=.5\dimen0 \advance\dimen1-.5\wd\proofbelow
\dimen4=\dimen1
\advance\dimen1\proofbelowshift \advance\dimen4-\proofbelowshift
\ifdim  \dimen1<0pt
\then   \advance\shortenproofleft\dimen1
        \advance\dimen0-\dimen1
        \dimen1=0pt
        \ifdim  \shortenproofleft<0pt
        \then   \setbox\proofabove=\hbox{%
                        \kern-\shortenproofleft\unhbox\proofabove}%
                \shortenproofleft=0pt
        \fi
\fi
\ifdim  \dimen4<0pt
\then   \advance\shortenproofright\dimen4
        \advance\dimen0-\dimen4
        \dimen4=0pt
\fi
\ifdim  \shortenproofright<\wd\proofrulename
\then   \shortenproofright=\wd\proofrulename
\fi
\dimen2=\shortenproofleft \advance\dimen2 by\dimen1
\dimen3=\shortenproofright\advance\dimen3 by\dimen4
\ifproofdots
\then
        \dimen6=\shortenproofleft \advance\dimen6 .5\dimen0
        \setbox1=\vbox to\proofdotseparation{\vss\hbox{$\cdot$}\vss}%
        \setbox0=\hbox{%
                \advance\dimen6-.5\wd1
                \kern\dimen6
                $\vcenter to\proofdotnumber\proofdotseparation
                        {\leaders\box1\vfill}$%
                \unhbox\proofrulename}%
\else   \dimen6=\fontdimen22\the\textfont2 
        \dimen7=\dimen6
        \advance\dimen6by.5\proofrulebreadth
        \advance\dimen7by-.5\proofrulebreadth
        \setbox0=\hbox{%
                \kern\shortenproofleft
                \ifdoubleproof
                \then   \hbox to\dimen0{%
                        $\mathsurround0pt\mathord=\mkern-6mu%
                        \cleaders\hbox{$\mkern-2mu=\mkern-2mu$}\hfill
                        \mkern-6mu\mathord=$}%
                \else   \vrule height\dimen6 depth-\dimen7 width\dimen0
                \fi
                \unhbox\proofrulename}%
        \ht0=\dimen6 \dp0=-\dimen7
\fi
\let\doll\relax
\ifwasinsideprooftree
\then   \let\VBOX\vbox
\else   \ifmmode\else$\let\doll=$\fi
        \let\VBOX\vcenter
\fi
\VBOX   {\baselineskip\proofrulebaseline \lineskip.2ex
        \expandafter\lineskiplimit\ifproofdots0ex\else-0.6ex\fi
        \hbox   spread\dimen5   {\hfi\unhbox\proofabove\hfi}%
        \hbox{\box0}%
        \hbox   {\kern\dimen2 \box\proofbelow}}\doll%
\global\dimen2=\dimen2
\global\dimen3=\dimen3
\egroup 
\ifonleftofproofrule
\then   \shortenproofleft=\dimen2
\fi
\shortenproofright=\dimen3
\onleftofproofrulefalse
\ifinsideprooftree
\then   \hskip.5em plus 1fil \penalty2
\fi
}

\newcommand\strikethrough[1]{{\setbox0=\hbox{$#1$}
\hrule height.75ex depth-.65ex width\wd0 \kern-\wd0\box0}}

\mathchardef\gt="313E \mathchardef\lt="313C

\newcommand{\comma}{,\ldots ,}

\def\undern#1{\vtop{\ialign{##\crcr
 $\hfil\displaystyle{#1}\hfil$\crcr
 \noalign{\kern-.1pt\nointerlineskip}%
 ~\,\raisebox{-.5ex}{$n \to \infty$} \crcr}}}

\def\undern#1{\vtop{\ialign{##\crcr
 $\hfil\displaystyle{#1}\hfil$\crcr
 \noalign{\kern-.1pt\nointerlineskip}%
 ~\,\raisebox{-.5ex}{$i$} \crcr}}}

\usepackage{proof,amsmath,enumitem}

\usepackage{hyperref}

 \newtheorem{theorem}{Theorem}[section]					
 \newtheorem{definition}[theorem]{Definition}				
 \newtheorem{lemma}[theorem]{Lemma}	

\newtheorem{corollary}[theorem]{Corollary}				
\newtheorem{remark}[theorem]{Remark}					
\newtheorem{example}[theorem]{Example}					

\newcommand{\Nneg}{ \protect{\raisebox{-.3ex}{$\neg$}\hspace{-1ex} {\raisebox{.3ex}{$\neg$} }}}

\newenvironment{proof}{\begin{trivlist}\item[]{\bf Proof.}}{\hspace*{\fill} $\blacksquare$ \end{trivlist}}

\newcommand\DSN[1][]{{\bf CN}^{#1}}

\newcommand\SN[1][]{{\bf CNN}^{#1}}
\newcommand\pa[1]{\langle #1\rangle}
\newcommand\upa[1]{\{ #1\}}
\newcommand\istrue[3][M]{#1\vDash_{#2} #3}
\newcommand\isntrue[3][M]{#1\nvDash_{#2} #3}
\newcommand\Num{\upa{1,2}}
\renewcommand\nin{\mathbin{\not\in}}

\newcommand\one{\text{t}}
\newcommand\K{K}
\newcommand\F{F}
\newcommand\C{C}
\newcommand\A{A}
\newcommand\T{T}
\newcommand\MP{MP}
\newcommand\N{N}

\begin{document}
\title{The Attack as Strong Negation, Part I}

\date{Compiled on \today} 

\author{D. Gabbay\\
King's College London,\\ Department of Informatics,\\
Strand, London, WC2R 2LS, UK;\\
Bar Ilan University, Ramat Gan, Israel\\
and  \\
University of Luxembourg, Luxembourg.\\
{\tt dov.gabbay@kcl.ac.uk}\\[2.5ex]
M. Gabbay\\
Cambridge University, UK.\\
{\tt mg639@cam.ac.uk}\\
{\small Submitted to {\em Logic Journal of the IGPL} 1.4.15. Revised 1.6.15. Paper 540}
}
\maketitle

\begin{abstract}
We add strong negation $N$ to classical logic and interpret the attack relation of ``$x$ attacks $y$"  in argumentation as $(x\to Ny)$.
We write a corresponding object level (using $N$ only) classical theory for each argumentation network and show that the classical models of this theory correspond exactly to the complete extensions of the argumentation network.
We show by example how this approach simplifies the study of abstract argumentation networks.
We compare with other translations of abstract argumentation networks into logic, such as classical  predicate logic or modal logics, or logic programming, and we also compare with Abstract Dialectical Frameworks.
\end{abstract}

\section{Background from classical logic: non-logical axiomatic theories}

This paper introduces a particularly intuitive and simple  representation/translation of abstract argumentation networks into classical propositional logic. All we need is a simple version of strong negation. So our starting point must be  to introduce this negation.

Classical propositional logic can be properly axiomatised in many ways. For simplicity, let us take the set $\BT$ of all tautologies as axioms and the rule of modus ponens as the deduction rule. Let us assume that the connectives used are the usual ones $\{\neg, \wedge, \vee, \to, \bot, \top\}$ and the atomic sentences are the set $P=\{p_1, p_2, p_3,\ldots\}$.  Classical logic is strongly complete for the classical semantics. Models are assignments $h$ giving values in $\{0,1\}$ to the atoms of $P$.

We have satisfaction defined as follows for wffs $A$ and theories $\Delta$.
\begin{itemize}
\item $h\vDash p$ iff $h ( p )=1$, for $p\in P$
\item $h\vDash \top$
\item $h\not\vDash \bot$
\item $h\vDash \neg A$ iff $h\not\vDash A$
\item $h\vDash A\wedge B$ iff $h\vDash A$ and $h\vDash B$
\item $h\vDash A\vee B$ iff $h\vDash A$ or $h\vDash B$
\item $h\vDash A\to B$ iff $h\vDash A$ implies $h\vDash B$.
\item $h\vDash \Delta$ iff $h\vDash A$ for all $A\in \Delta$.
\end{itemize}

The notion of proof $A_1\comma A_n\vdash B$ can be defined in classical logic in many ways and completeness holds for any set of wffs $\Delta$ and any $A$:

\begin{itemize}
\item $\Delta\vdash A$ iff for all models $h$ we have $h\vDash \Delta$ implies $h\vDash A$.
\end{itemize}

We now introduce the notion of a set of sentences being a  specific non-logical axiomatic theory $\Theta_\Bn$.

Consider again the set of atomic wff
\[
P=\{p_1, p_2,p_3,\ldots\}.
\]
Suppose we insist, for our own reasons, that we want to consider only those models $h$ satisfying the restriction $(\Bn)$ below:
\begin{itemlist}{(N)}
\item [(\Bn )] For all even index atoms $p_2, p_4, p_6\comma p_{2i},\ldots$ we have $h(p_{2i})=1$ implies $h(p_{2i-1})=0, i=1, 2, 3,\ldots$
\end{itemlist}

 There is a theory $\Theta_{\Bn}$ of classical propositional logic whose models are exactly all the models
satisfying $(\Bn)$. This theory is
\[
\Theta_{\Bn} =\{p_2\to \neg p_1, p_4\to \neg p_3, p_6\to \neg p_5,\ldots\}
\]

Let us now for the sake of clarity, rename the atoms of $P$ with the help of a new symbol $N$.  We rename as follows:

\[
\begin{array}{l}
q_1 = \mbox{def. } p_1\\
Nq_1=\mbox{def. } p_2\\
q_2=\mbox{def. } p_3\\
Nq_3=\mbox{def. } p_4\\
~~~~\vdots\\
q_i=\mbox{def. } p_{2i-1}\\
Nq_i=\mbox{def. } p_{2i}\\
~~~\vdots
\end{array}
\]

We can thus write the set of atoms as $P$ as the set $Q$
\[
Q = \{ q_1, Nq_1, q_2, Nq_2,\ldots\}.
\]

The theory $\Theta_{\Bn}$ becomes the theory
\[
\Theta_{\Bn}=\{Nq\to \neg q | q\in \{q_1, q_2,\ldots\}\}
\]

$\Theta_{\Bn}$ is considered a non-logical set of axioms on the symbol $N$.

Note that $N$ cannot be iterated and can be applied only to atoms taken from $\{q_1, q_2, \ldots\}$.

We said that we regard $\Theta_{\Bn}$ as a non-logical axiomatic theory on the symbol $N$. This is common practice in logic and model theory. Consider, for example, the classical theory of Abelian groups formalised in classical logic for the multiplication symbol $*$ and the constant ${\bf 1}$.  We add to the logical axioms of predicate logic the non-logical group axioms below

\begin{itemize}
\item $\forall xyz(x * (y*z)= (x*y)*z)$
\item $\forall xy (x*y=y*x)$
\item $\forall x (x* {\bf 1} =x)$
\item $\forall x \exists y (x*y={\bf }1)$
\end{itemize}

In our case our non-logical axioms are $\Theta_{\Bn}=\{Nq\to \neg q\}$.

Thus our logic for $N$ is a theory (of $N$) within classical propositional logic, much in the same way as the theory of Abelian groups is a theory within the classical predicate calculus.
The next section defines  the logic \BC\BN\ formally.

\section{The logic {\bf CN}}

We add to classical propositional logic the strong negation symbol $N$.  We can thus form atomic sentences like $\{q_1\comma q_n,\ldots \}$ as well as atoms of the form $\{Nq_1, Nq_2,\ldots\}$.  We do this as explained in Section 1.

At this point we do not allow iterations of $N$. We have the usual connectives $\neg$ (negation), $\wedge, \vee, \to, \top, \bot$.  We shall discuss iterated use of $N$ in a later section.  We require the only additional non-logical axioms $\Theta_{\Bn} =\{Nq\to \neg q\}$.
 We can thus view $Nq$ as strong negation. For example $q$ might be true or $q$ might be false (i.e.\ $\neg q$ is true) or $q$ might be strongly false ($Nq$ true) or  $q$ might be false but not strongly false ($\neg q$ and $\neg Nq$ true).

 Let us define our logic directly.

 \begin{definition}[The {\bf CN} classical propositional logic with atomic strong negation]\label{540-D1}{\ }
 \begin{enumerate}
 \item Let $\BBL_N$ be a language with a set of atoms $Q_0=\{q_1, q_2\comma q_n,\ldots\}$ and the connectives $\{\neg, \wedge, \vee, \to, N\}$.  $\{\neg, \wedge, \vee, \to\}$ are the usual classical connectives and $N$ is a unary connective applied once to atoms only.
 \item An atomic formula has the form $q$ or $Nq$ where $q\in Q_0$.
 \item A wff has the form $A=$ atomic formula or $A\wedge B, \neg A, A\vee B, A\to B$ where $A$ and $B$ are wffs.
 \item We regard as axioms for our logic {\bf CN} the nonlogical set of axioms $\Theta_{\Bn}$ as discussed in Section 1.     The proof theory {\bf CN} is relying on the proof theory of classical logic \BC, as follows:
 \begin{itemize}
 \item $\Delta\vdash_{{\bf CN}} A$ iff (def) $\Delta\cup \Theta_{\Bn}\vdash_{\bf C} A$.
 \end{itemize}
 \item A model $h$ for the logic {\bf CN} is an assignment $h$ giving each atomic wff of the form $q$ or $Nq$ a value in $\{0,1\}$ such that if $h(Nq)=1$ then $h(q) =0$.
 \item Satisfaction is defined in the usual way.
 \end{enumerate}
 \end{definition}

\begin{theorem}\label{540-T2}
{\bf CN} is complete for the proposed semantics.
\end{theorem}

\begin{proof}
Our discussion in Section 1 presented {\bf CN} as a non-logical theory $\Theta_{\bf n}$ of the classical propositional calculus of item 4 of Definition \ref{540-D1} defined the consequence for {\bf CN}, via the classical consequence. Since we have strong completeness for \BC\ we also have it for {\bf CN}.
\end{proof}

\begin{remark}\label{540-R2}
Note that what we  are calling  our  `logic'  {\bf CN} is actually a theory in a two sorted classical propositional logic with two sorts of atoms  of the form $q$ and $Nq$. In Appendix~\ref{mike.appendix.2} we will turn our logic  into a proper modal logic, which we will  call  $\SN$.
\end{remark}
We now have enough machinery to faithfully represent abstract argumentation networks in the logic \BC\BN. This is the job of the next section.

\section{Expressing argumentation networks in {\bf CN}}\label{540-S3}
This section will show a simple way of translating formal argumentation networks into {\bf CN}. We assume we are dealing with finite argumentation networks.\footnote{Actually, we do not need the requirement that the network  is finite. All we need is that it is finitary, namely that each point is attacked by a finite number of attackers. This will allow us to write classical wffs describing the attacks.}

We shall use the Caminada labelling characterisation of complete  extensions.   See \cite{540-1} for a survey.  Given an argumentation network $(S, R)$, with $S\neq \varnothing$ and $R\subseteq S\times S$, a legitimate Caminada labelling $\lambda$ on $S$ is a function $\lambda: S\mapsto \{\mbox{in, out, und}\}$ giving values ``in'', ``out'', ``undecided'' to each $x\in S$, satisfying the following properties.

\begin{itemlist}{(CC)}
\item [(C1)] $\lambda (x)=$ in iff either $\neg \exists y (yRx)$ or $\forall y (yRx\to \lambda (y)=$ out).
\item [(C2)] $\lambda (x)=$ out iff $\exists y (yRx\wedge\lambda (y)=$ in)
\item [(C3)] $\lambda (x) =$ und iff $\forall y (yRx\to \lambda (y)\neq$ in) and $\exists y(yRx\wedge\lambda (y)=$ und).
\end{itemlist}

Each legitimate Caminada labelling gives rise to a unique complete extension and all complete extensions can be obtained in this way.  See \cite{540-1}.

\begin{definition}\label{540-D3}
Let $\CA =(S, R)$ be a formal argumentation network which is finitary, i.e.\ each point has a finite number of attackers  This means that $S\neq \varnothing$  and $R\subseteq S\times S$ is the attack relation. Define a theory $\Delta_{\CA}$ of the logic {\bf CN} as follows
\begin{enumerate}
\item We can assume that $S\subseteq Q$ (i.e.\ the arguments of $S$ are identified as atoms of the logic).
\item $\Delta_{\CA}=\{x|\neg \exists y(yRx)\} \cup \{y\leftrightarrow \bigwedge_{z\in \mbox{ Attack } (y)}
Nz | y\in S\} \cup
\{z\to Ny | zRy\} \cup \{(\bigwedge_{z\in\mbox{  Attack}(y)}  \neg z) \wedge (\bigvee_{z\in\mbox{ Attack}(y)} \neg Nz))\to \neg y \wedge\neg Ny |y\in S\}$, where  Attack$(y) =\{z|zRy\}$.
\end{enumerate}
\end{definition}

\begin{theorem}[First Correspondence Theorem]\label{540-T4}
Let $\CA = (S, R)$. Then the models of $\Delta_{\CA}$ correspond exactly to the complete extensions of $\CA$.
\end{theorem}

\begin{proof}
\begin{enumerate}
\item Let $h$ be a model of $\Delta_{\CA}$. We show it defines a complete extension on $\CA$   (note that  two different models can yield the same complete extension): We use the Caminada labelling function. Let $x\in S$. Define
\[
\lambda_h(x) =
\left\{
\begin{array}{l}
\mbox{in, if } h(x) =1\\
\mbox{out, if } h(Nx) =1\\
\mbox{und, if } h(x) =h(Nx) =0
\end{array}
\right.
\]
we prove the following
\begin{enumerate}
\item $\lambda_h$ is well defined. For each $x$ we can have exactly one of the three cases mentioned in the definition of $\lambda_h$. The reason for that is that we have the axiom $Nq \to \neg q$ and so the case $h(x) =1$ and $h(Nx) =1$ does not arise.
\item The crucial points to show are
\begin{enumerate}
\item If $x$ is not attacked then $\lambda_h(x) =$ in. This holds because $x\in \Delta_{\CA}$.
\item If $zRy$ and $\lambda_h(z)=$  in then $h(z)=1$. Therefore $h(Ny)=1$ (because $z\to Ny$ is in the theory) and so $\lambda_h(y) =$ out.
\item Suppose for all $z$ such that $zRy$ we have $\lambda_h(z)=$ out.  We want to show the $\lambda_h(y)=$ in.  We have for all $z$ such that $zRy$ that $h(Nz)=1$ hence $h(\bigwedge_{zRy}Nz)=1$ and therefore $h(y)=1$  (because of the theory), and so $\lambda_h(y)=$ in.
\item Suppose $\lambda_h(y)=$ und.  Then $h(y)=h(Ny) =0$.  Let $zRy$. We            cannot have $\lambda_h(z)=$ in because then $h(z) =1$ and this implies $h(Ny)=1$. Thus none of the     attackers $z$ of $y$ are in. Thus for all $z$ such that $zRy$ we have $h(z)=0$.  Can we have that for all of such $z,h(Nz)=1$?  If this were the case, since $y\leftrightarrow \bigwedge_{zRy}Nz$ is in $\Delta_{\CA}$ we get $h(y) =1$ i.e.\ $\lambda_h(y)=$ in, contradicting our assumption that $\lambda_h(y)=$ und.

Therefore for some $z$ such that $zRy$ we have $h(Nz)=0$. But this means that $\lambda_h(z)=$ und. We thus got that if $\lambda_h(y)=$ und then none of the attackers $z$ of $y$ are in and at least one of them is undecided.
\item If $\lambda_h (y) =$out, (i.e.\ $h(Ny)=1$),  show that for some $z$ such that $zRy$ , we have $\lambda_h (z) =$in, (i.e.\  $h(z)=1$).  Otherwise for all $z$ such that $zRy$, we have $h(z)=0$. We ask about any such $z$, is $h(Nz)=1$? If the answer is yes for all of them then we must have $h(y)=1$, by item (iii) above. So for some $z$ we have $h(Nz)=0$. Then by the axiom $\bigwedge_{z\in\mbox{ Attack}(y)}  \neg z) \wedge (\bigvee_{z\in\mbox{ Attack}(y)} \neg Nz))\to \neg y \wedge\neg Ny$ we get that $h(Ny)=0$.

Either way, the assumption that for all $z$ such that $zRy$ we have $h(z)=0$, leads to a contradiction.
\end{enumerate}
Thus $\lambda_h$ is a legitimate labelling, giving rise to a complete extension.
\end{enumerate}
\item Let $\lambda$ be a legitimate labelling giving rise to a complete extension. We show that it gives a model for $\Delta_{\CA}$.

Define $h_\lambda$ as follows:
\[
\begin{array}{l}
h_\lambda (x)=1\mbox{ if }\lambda (x)=\mbox{in}\\
 h_\lambda (Nx) =1\mbox{ if } \lambda(x) =\mbox{out}.
\end{array}
\]
We show that all axioms of $\Delta_{\CA}$ hold
\begin{enumerate}
\item Show that $h_\lambda \vDash Nx\to \neg x$ otherwise $h_\lambda (Nx) =h_\lambda(x)=1$.  This means that $\lambda(x)=$ out and $\lambda(x)=$ in which is not possible.
\item Second we show that if $zRy$ then $h_\lambda \vDash z \to Ny$.  Otherwise we have $h_\lambda(z) =1$ and $h_\lambda (Ny)=0$.  The former implies $\lambda (z)=$ in. Therefore $\lambda (y)-$ out and so by definition $h_\lambda (Ny) =1$, a contradiction.
\item We now show that
\[
h_\lambda \vDash y \leftrightarrow \bigwedge_{zRy} Nz.
\]
Assume $h_\lambda (y)=1$. Then $\lambda (y)=$ in. So for all $z$ such that $zRy, \lambda (z)=$ out, and so by definition of $h_\lambda$, for all such $z, h_\lambda (Nz) =1$ and hence $h_\lambda (\bigwedge_{zRy}Nz)=1$.

Assume $h_\lambda (y) =0$. Therefore $\lambda (y)\neq $ in. Then either $\lambda(y)=$ out or $\lambda (y)=$ und. If $\lambda(y)=$ out then for some $z$ such that $zRy$ we have $\lambda (z) =$ in. Therefore $\lambda (z)\neq$ out and so $h_\lambda (Nz) =0$.

If $\lambda (y)=$ und, then for some $z, zRy$ an $\lambda (z)=$ und again $\lambda(z)\neq$ out and so $h_\lambda (Nz) =0$.

Thus for sure if $h_\lambda(y)=0$ then for some $z, zRy$ and $h_\lambda (Nz) =0$ and so $h_\lambda (\bigwedge_{zRy} Nz)=0$.
\item  We show that
$h_\lambda \vDash \bigwedge_{z\in\mbox{ Attack}(y)}  \neg z) \wedge(\bigvee_{z\in\mbox{ Attack}(y)} \neg Nz))\to \neg y \wedge \neg Ny$.
If the antecedent holds then all attackers of $y$ are not in and one of them is undecided. Therefore $y$ is undecided and so the consequent holds.
\end{enumerate}
\end{enumerate}
\end{proof}
\begin{corollary}\label{540-C5}
Let $\CA$ be an argumentation network. Consider $\Delta_{\CA}$. Then $\Delta_{\CA}$ is {\bf CN} consistent.
\end{corollary}

\begin{proof}
Since $\CA$ has the complete ground extension, by the previous Theorem \ref{540-T4} this yields a model for $\Delta_{\CA}$.
\end{proof}

\begin{example}\label{540-E5}
\begin{enumerate}
\item Consider the argumentation network of Figure \ref{540-F6}.

\begin{figure}
\centering
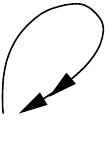
\caption{}\label{540-F6}
\end{figure}

Its theory in {\bf CN} is $\Delta = \{ x\to Nx, Nx\leftrightarrow x, \neg x \wedge\neg Nx \to \neg x \wedge\neg Nx\}$.  Since we have the $N$ axiom $Nx \to \neg x$ we get that $x\to \neg x$ is provable and so $\neg x$ is provable and so $\neg Nx$ is also provable.  This means we have only one model $h$ with $h(x) =h(Nx)=0$ and therefore the only extension is $\lambda (x) =$ und.
\item Consider the additional axiom added to the theory $\Theta_\Bn$ namely
\[
{\bf Stable}: \{x\vee Nx |\mbox{for all } x\}.
\]
The theory $\Theta_\Bn +$ {\bf stable} does not have models $h$ in which $h(x) =h(Nx)=0$ for any $x$. Thus this axiom characterises all stable extensions.
\end{enumerate}
\end{example}

\begin{theorem}\label{540-T6}
Let $\CA= (S, R)$ be an argumentation network. Consider $\Delta_{\CA}$ and let $E =\{x|\Delta_{\CA}\vdash_{\bf CN} x\}$.  Then $E$ is the ground extension of $\CA$.
\end{theorem}

\begin{proof}
We know that $\Delta_{\CA}$ is consistent. Hence $E$ is consistent. We show that $E$ is a complete extension. Since other complete extension $E'$ corresponds to a model of $\Delta_{\CA}$, $E'$ contains $E$. Thus $E$ would be the smallest complete extension --- namely $E$ is the ground extension.  We now show that $E$ is a complete extension.
\begin{enumerate}
\item $E$ is conflict free.  Let $x,y \in E$.  If $xRy$ holds, then $x\to Ny \in \Delta_{\CA}$ and hence $\Delta_{\CA}\vdash_{\bf CN} x\to \neg y$, which contradicts the consistency of $\Delta_{\CA}$ since $\Delta_{\CA}\vdash_{\bf CN} y$.
\item Assume $E$ protects $x$. We show $x\in E$. Let $zRx$. then for some $y\in E, yRz$ holds. Hence $y\to Nz$ is in $\Delta_{\CA}$ and so $\Delta_{\CA}\vdash_{\bf CN} Nz$. Thus we have
\[
\Delta_{\CA}\vdash_{\bf CN}\bigwedge_{zRx} Nz
\]
and hence $\Delta_{\CA}\vdash_{\bf CN} x$ and so $x\in E$.
\end{enumerate}
\end{proof}

\begin{remark}\label{RMarch22}
This is a clarifying remark about the correspondence between extensions of an argumentation network $\CA =(S, R)$ and {\bf CN} models of the theory $\Delta_{\CA}$.

Consider the three argumentation networks $\CA_1, \CA_2, \CA_3$ in Figure \ref{540-FMarch22}.   They all have the same extension, $x=$ in, $y=$ out, $z=$ in.

\begin{figure}
\centering
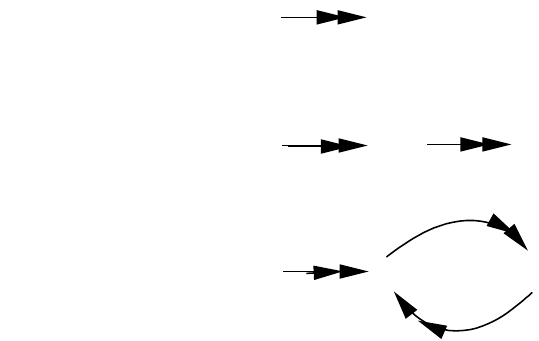
\caption{}\label{540-FMarch22}
\end{figure}

Their theories $\Delta_{\CA_i}$ are different, but they have the same models.
\[\begin{array}{lcl}
\Delta_{\CA_1} &=&\{ x, x\to Ny, y\leftrightarrow Nx, z, \neg x\wedge\neg Nx\to \neg y\wedge\neg Ny\}\\
\Delta_{\CA_2}&=&\{x, x\to Ny, y\leftrightarrow Nx, y\to Nz, z\leftrightarrow Ny,\\
&& \neg x\wedge\neg Nx\to \neg y\wedge \neg Ny, \neg y\wedge\neg Ny\to \neg z\wedge\neg Nz\}\\
\Delta_{\CA_3}&=& \{x, x\to Ny, y\leftrightarrow Nx\wedge Nz, x\vee z\to Ny, \neg x\wedge\neg z\wedge \\
&&(\neg Nx\vee \neg Nz)\to\neg y\wedge\neg Ny, y\to Nz,  z\to Ny, \\
&&\neg y\wedge \neg Ny\to \neg z\wedge\neg Nz, \neg z\wedge\neg Nz\to \neg y\wedge\neg Ny,\\
&& y\leftrightarrow Nz, z\leftrightarrow Ny\}
\end{array}
\]
All three theories have only one model $h$
\[
x =1, Nx =0, y=0, Ny=1, z=1, Nz=0.
\]
The moral of the example is that the theories describe the extensions and not the geometry of the networks. In fact, $\Delta_{\CA_i}$ are all  equivalent to $\Delta_{\CA}=\{x, \neg Nx, \neg y, Ny, z, \neg Nz\}$.

Compare the {\bf CN} approach with the truly meta-level approach of \cite[Section 5]{540-1}, discussed below in the beginning of Section 5.

We describe $(S, R)$ in predicate logic with unary predicates $Q_1(x) =x$ is in, $Q_0(x)=x$ is out, and $Q_?(x)=x$ is und, and a relation $xRy$ for attack.  Thus using $R$ the networks of Figure \ref{540-FMarch22} can be distinguished in the semantics.  Note, however, that we can read the geometry of the networks from the syntax of the theories $\Delta_{\CA_i}, i=1,2,3$, but not from their models!  This highlights the importance of proof theory.
\end{remark}

\begin{remark}\label{540-R5}
The previous correspondence theorem reduces the idea of attack in argumentation networks to the idea of  strong negation in classical propositional logic. This reduction simplifies every move we make in the argumentation area and gives us a tremendous advantage in making available to us all the machinery of classical logic.  This includes
\begin{enumerate}
\item extensions for formal argumentations such as joint attacks, support, higher level attacks, probabilistic argumentation, predicate/modal logic argumentation and more, all can be done more simply and easily using strong negation, see the following sections;
\item applications of argumentation become application of classical logic;
\item new ideas can be more readily imported from classical logic into argumentation;
\item the Caminada labels, $x$ is  in, $x$ is out and $x$ is undecided are available in the object language as $x$ is true, $Nx$ is  true and $\neg x \wedge\neg Nx$ is true, respectively.
\item we now need to ask ourselves: what is the added value of argumentation over classical logic?  We need a clear and detailed answer to this question.
\end{enumerate}
\end{remark}

\section{Intermediate critical evaluation}
This section pauses our formal development to evaluate what we have so far and to explain to the reader where we are going. We shall do this by a list of critical comments.

\subsection*{CC1. Basing argumentation on the unary notion of ``being attacked''}
We read $Nx$ as ``$x$ is being attacked'', we are not saying how and from where this attack comes. This makes $Nx$ a kind of strong negation (with axiom $Nx\to \neg x$), see \cite{540-13}.  This single simple idea allows us to have Theorem \ref{540-T4}.  It also allows us to go in the direction of turning the system {\bf CN} into a paraconsisent logic of negation (see Wikipedia) by adding axioms on iterating $N$ (e.g.\ the axiom $NN x \leftrightarrow x$). We can do this safely for as long as Theorem \ref{540-T4} is retained.  We shall address this direction in later sections. We believe we can achieve similar results for logic programming by reading $Nx$ as ``$x$ fails''.  This direction, and the connection with answer set programming, is left for a subsequent paper.

There is another direction we can go in.  We can use another meaningful logic such as intuitionistic logic, linear logic or relevance logic or fuzzy logic to replace classical logic and thus get argumentation theory in those contexts.  Again we shall look at this in a subsequent paper.

\subsection*{CC2.  Simple way of defining joint attacks}
The theory $\Delta_{\CA}$ written for an argumentation network $\CA =(S, R)$ is comprised of several components.
\begin{enumerate}
\item The logic {\bf CN} (the use of $N$)
\item Formulas defining when an argument $x$ is ``in''.

This is the part
$$
x\leftrightarrow \bigwedge_{zRx} Nz
\leqno (\BF_{\rm in})
$$

$\BF_{\rm in}$ also includes the part relating to ``$x$ is not attacked'', since the empty conjunction is $\top$.
\item Formulas defining when an argument $x$ is ``out''.  This is the part $z\to Nx$, for all $zRx$, or if we write it as a single formula, it is $\BF_{\rm out}$.

$$
\bigvee_{zRx} z\to Nx.
\leqno (\BF_{\rm out})$$

\item Formulas defining when $x$ is undecided
$$
( \bigwedge_{z\in \mbox{ Attack}(x)}  \neg z) \wedge(\bigvee_{z\in \mbox{ Attack}(x)} \neg Nz)\to \neg x \wedge\neg Nx
\leqno (\BF _{\rm und})
$$

Once we put the above in $\Delta_{\CA}$ we get the correspondence theorem, Theorem \ref{540-T4}, generating complete extensions by models of $\Delta_{\CA}$.

Now we can see how easy it is to generalise argumentation to joint attacks.  Joint attacks, introduced in \cite{540-3,540-4} can be described by the configuration in Figure \ref{540-F8} (using Gabbay's notation from \cite{540-3}).

\begin{figure}
\centering
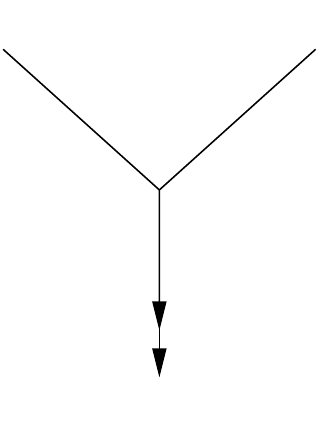
\caption{}\label{540-F8}
\end{figure}

The meaning is that we have $x$ is out if all of $z_1\comma z_n$ are in. This we can write in our logic {\bf CN} simply as $\bigwedge_{(z_1\comma z_n) R_0 x} z_i\to Nx$ where $R_0$ is the joint attack relation of the form $R_0\subseteq (2^S-\varnothing) \times S$.

So we can write the formula $\BF^{\rm Joint}_{\rm out}$ for joint attacks as
$$
(\bigvee_{GR_0 x} \bigwedge_{z\in G} z)\to Nx
\leqno (\BF^{\rm Joint}_{\rm out})
$$
where $G\subseteq S, x \in S$.

The formula for $\BF^{\rm Joint}_{\rm out}$ for a point $x$ will be
$$
x\leftrightarrow (\bigwedge_{GR_0 x}\bigvee_{z\in G} Nz)
\leqno (\BF^{\rm Joint}_{\rm in})
$$

$$
\bigwedge_{GR_o x}(\bigvee_{z\in G}\neg z )\wedge(\bigvee_{GR_0 x}\bigwedge_{z\in G}(z\vee \neg Nz))\to \neg x\wedge\neg Nx
\leqno (\Bf^{\rm Joint}_{\rm und})
$$

If we use these three wffs $\BF^{\rm Joint}_{\rm in}$ and $\BF^{\rm Joint}_{\rm out}$  and $\BF^{\rm Joint}_{\rm und}$ to define $\Delta_{\CA}$, we can study the semantics for joint attacks, in the object level, in classical logic as models of $\Delta_{\CA}$, provided  we prove a correspondence theorem similar to Theorem \ref{540-T4}.

The reader can compare the simplicity of this approach to what is done in the papers \cite{540-3,540-4}.
\end{enumerate}

\subsection*{CC3.  Simple way of defining higher level attacks}
Higher level attacks on attacks. These were studied in many papers \cite{540-2,540-5,540-6,540-7,540-9}.  They were first introduced in general in Gabbay's paper \cite{540-8}.  Figure \ref{540-F9} illustrates the basic configuration for higher level attacks.

\begin{figure}
\centering
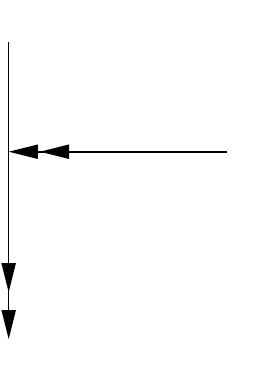
\caption{}\label{540-F9}
\end{figure}

$z$ attacks $x$ and $y$ attacks the attack $z\tO x$. We write it as $y\tO (z\tO x)$.

We need to write the $\BF^{\rm higher}_{\rm in}$ and $\BF^{\rm higher}_{\rm out}$  and $\BF^{\rm higher}_{\rm und}$ of this type of attack.  In our set up with $N$ it is easy to write this!  The $\BF_{\rm out}$ is $z\wedge Ny \to Nx$ and the $\BF_{\rm in}$ involves $Nz\vee y$. Let us write the wff's in detail for a network with one level of higher attacks.  So our networks have the form $(S, R, R_1)$, where $R\subseteq S\times S$ are the attacks and $R_1\subseteq S\times R$ are the attacks on the attacks. We write:

$$
(\bigvee_{zRx}\bigwedge_{yR_1(z,x)} z\wedge Ny) \to Nx
\leqno (\BF^{\rm higher}_{\rm out})
$$

$$\bigwedge_{zRx}\bigvee_{yR_1(z,x)} (Nz\vee y))\leftrightarrow x
\leqno (\BF^{\rm higher}_{\rm in})
$$

$$
(\neg\bigwedge_{zRx}\bigvee_{yR_1(z,x)} (Nz\vee y) \wedge\neg\bigvee_{zRx} z\wedge\bigwedge_{yR_1(z,x)}Ny)\to \neg x \wedge\neg Nx
\leqno (\BF^{\rm higher}_{\rm und})
$$

Again we use $\BF^{\rm high}_{\rm in}$,   $\BF^{\rm high}_{\rm out}$  and $\BF^{\rm high}_{\rm und}$ to define $\Delta_{\CA}$, and let the semantics be all {\bf CN} models of $\Delta_{\CA}$. We need to prove a correspondence theorem similar to Theorem \ref{540-T4}.

\subsection*{CC4.  A word of caution}
Although we are showing how other types of networks can be translated into {\bf CN}, we are not just saying ``look, our paper is introducing you to another master generalisation of argumentation''. We reserve judgement about $N$ until the end of this paper and we might say at the end to the reader ``in view of our paper, do not think any more in the old conceptual framework of argumentation of $(S, R)$ but think in terms of strong negation $N$ of just being attacked, and do your argumentation from now on in classical logic with $N$''.  The reason we reserve judgement is because we want to figure out first how the use of $N$ affects  related systems such as ABA (Assumption Based Argumentation, see~\cite{540-20}) and  ASPIC (see~\cite{540-21}).  In ASPIC and ABA, the game is to start with a logic theory $\Delta$, define proofs from $\Delta$ as the argument set $S$, define a respective suitable  attack $R$ between proofs, then define $(S, R)$ as an instantiated network, take extensions $E\subseteq S$ and then make sure that $E$ is consistent in the logic of $\Delta$. This we perceive as a possibly unnecessary external detour.  Our method might say to ABA and ASPIC, ``why do you need all this roundabout way involving a multitude of problems?  Why not add strong negation $N$ directly to the logic of $\Delta$ and model your argumentation there and you are done?  If {\bf CN} can swing this and we succeed in working out the details, then ABA and ASPIC would immensely benefit from our conceptual view.  See {\bf CC7} below.

\subsection*{CC5.  Comparing with abstract dialectical framework (ADF)}
ADF were introduced in \cite{540-11} by Professor G. Brewka as a generalisation of argumentation frameworks and immediately caused both excitement and criticism.  In this paper we shall use an example from Brewka's slides \cite{540-10} to do our comparison.

An ADF has the form $(S, R, C)$ where $S$ is a set of arguments and $R\subseteq S\times S$ is the link relation (Brewka calls them ``links'' because he does not view them as attacks). $C$ is a family of acceptance relations. For each $x\in S$, there is a formula of classical propositional logic $\varphi_x$, specifying the acceptance (``in'') condition for $x$, based on the acceptance values of $\{y|yRx\}$.

In our terms, as well as in classical logic terms, what ADF is saying is condition $\BF^{C}_{\rm in}$:
$$
x\leftrightarrow \varphi_x, \mbox{ for all } x\in S
\leqno (\BF^C_{\rm in})
$$

If we regard $\BF^C_{\rm in}$ as a condition in our logic {\bf CN}, then we can define all complete extensions in our sense as all extensions obtained from models \Bm\ of $\Delta_C =(\Theta_\Bn +\BF^C_{\rm in})$, where the values \{in, out, und\} are as in the proof of part 1 of Theorem \ref{540-T4}, namely:
\begin{itemize}
\item $x=$ in  if $\Bm(x) =1$
\item $x=$ out if $\Bm (Nx) =1$
\item $x=$ und if
$\Bm(x) =\Bm (Nx) =0$.
\end{itemize}

Brewka, however works only in classical proposition logic with the three valued semantics according to Kleene.  So his models give 3 values $\{1, 0, u\}$ with
\[\begin{array}{l}
1\wedge 1 = 1\\
0\wedge 0 = 0\\
1\wedge u = 0 \wedge u = u \wedge u = u
\end{array}
\]
Brewka derives his extensions for the 3 valued models using some process. See \cite[slides 13--18]{540-10}.  We now reproduce in Figure \ref{540-F20} the original slide 18 of Brewka \cite{540-10} in order to compare ADF with our {\bf CN} approach.

The Brewka extensions are given in the figure.

\begin{figure}
\centering
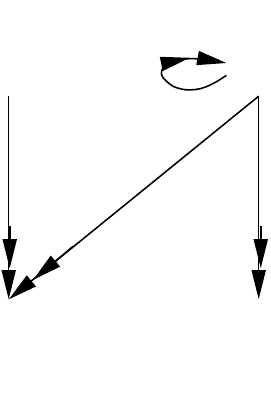
\begin{itemize}
\item models:
\begin{itemize}
\item $v_1=\{a\mapsto \Bt, b\mapsto\Bt, c\mapsto\Bt, d\mapsto f\}\mbox{ corresponds to } \{a,b,c,\neg d\}$
\item $v_2 =\{a\mapsto\Bt, b\mapsto f, c\mapsto f,d\mapsto\Bt\}\mbox{ corresponds to } \{a,\neg b,\neg c,d\}$
\end{itemize}
\item grounded model: $v_1=\{a\mapsto \Bt, b\mapsto \Bu, \mapsto \Bu, d\mapsto \Bu\}\mbox{ corresponds to }\{a\}$
\end{itemize}
\caption{}\label{540-F20}
\end{figure}

Figure \ref{540-F21} lists the models in {\bf CN} obtained for the ADF theory.

We note that our model $\Bm_1$ is he same as Brewka's $v_1$ and our model $\Bm_4$ is the same as Brewka's $v_2$.  However, model $\Bm_3$ gives $a=$ in, $b=$ in, $c=$ in and $d=$ und, and model $\Bm_2$ gives $b=c=$ und and $a=d=$ in. We do not get the Brewka's grounded model $G =\{a=\mbox{ in}, b=c=d=\mbox{ und}\}$.

This model is not one of the model of the theory
\[
\Delta_C=\{x\leftrightarrow \varphi_x|x \mbox{ a node in Figure \ref{540-F20}}\} \cup \Theta_\Bn.
\]

So how does Brewka get his grounded extension?  Look at $v_1$ and $v_2$, $a$ gets 1 in both, but $b,c$ and $d$ get value 1 in one of them and 0 in the other. So if we look at $G$ with the undecided value $u$ given to $b, c$ and $d$ for each of these arguments there are extensions which make its value 1 and there are extensions which make its value 0.  Therefore their value, according to ADF is undecided.

From our point of view, this way of looking at undecided is just external combinatorics, devoid of  conceptual content.  According to our view only $\neg x\wedge \neg Nx$ makes $x$ undecided, namely $x$ is false but not strongly false.

$G$ is not an extension. It is not a model of $\Delta_C$ because of the axiom $d=\neg b$.

\begin{figure}
\centering
\begin{tabular}{c|c|c|c|c|c|c|c|c}
model & $a=\top$& $Na$& $b=b$&$Nb$&$c=a\wedge b$&$Nc$&$d=\neg b$&$Nd$\\
\hline
$m_1$ &1&0&1&0&1&0&0&1\\
\hline
$m_2$ & 1&0&0&0&0&0&1&0\\
\hline
$m_3$&1&0&1&0&1&0&0&0\\
\hline
$m_4$&1&0&0&1&0&1&1&0
\end{tabular}
\caption{}\label{540-F21}
\end{figure}

Obviously we could try and add restrictions on the models to get the same results as Brewka (i.e.\ implement ADF in {\bf CN}) but why should we do that at all?  Our methodology is sound and stronger.

Given that we had $N$ in the object language, we can do more than ADF.  ADF writes the acceptance conditions in 2-valued classical logic and brings in the undecided value only through the semantic interpretation. We have the undecided value $(\neg x \wedge\neg Nx)$ in the language itself and therefore we can put the undecided property into the acceptance conditions. Consider the joint attack described in Figure \ref{540-F23}.

\begin{figure}
\centering
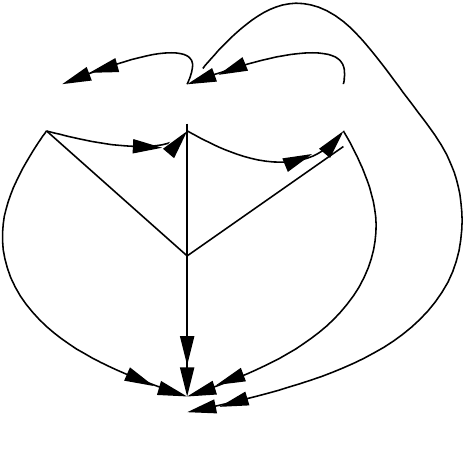
\caption{}\label{540-F23}
\end{figure}

Suppose all three attackers $a,b,c$ are undecided.  In this case we traditionally say that $x$ is also undecided because we do not know, maybe all three $a=b=c=$ in.  This could happen.  If, however, we adopt the new view that the chance that all three attackers are ``in'' can be disregarded, then we want to say it in our acceptance conditions.

So we want to say

$$ x\leftrightarrow (Na\wedge Nb\wedge Nc)\bigvee(\neg a\wedge \neg Na \wedge\neg b\wedge\neg Nb\wedge\neg c\wedge\neg Nc)
\leqno (\BF_{\rm in}(x))
$$

$$
a\vee b\vee c\to Nx
\leqno (\BF_{\rm out}(x))
$$

$$ \begin{array}{l}
a\leftrightarrow Nb\\
b\leftrightarrow Na\wedge Nc \\
c\leftrightarrow Nb
\end{array}
\leqno (\BF_{\rm in} (a,b,c))
$$

$$\begin{array}{l}
a\vee c \to Nb\\
b\to Na\wedge Nc
\end{array}
\leqno(\BF_{\rm out}(a,b,c))
$$

Can ADF express exactly the same extensions as what {\bf CN} gets for the above, especially the one $a=b=c=$ und, $x=$ in?

\subsection*{CC6.  Using {\bf CN} for the probabilistic approach}
There are many papers on probabilistic argumentation. Our paper \cite{540-18} offers a comprehensive approach based on probability models of classical logic. We shall therefore compare the probabilistic use of {\bf CN} with the approach in \cite{540-18}.  We shall use the network of Figure \ref{540-100F} as an example.

\begin{figure}
\centering
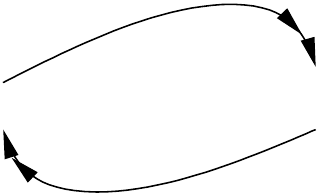
\caption{}\label{540-100F}
\end{figure}

The approach of \cite{540-18} regards $\{a,b\}$ as atoms of classical propositional logic. The logic with these atoms has the following models

\[\begin{array}{ll}
\Bm_2: & a\wedge b; x\\
\Bm_2: & a\wedge\neg b; y\\
\Bm_3:& \neg a\wedge b; z\\
\Bm_4:& \neg a \wedge \neg b; 1-x-y-z
\end{array}
\]

We can give a probability distribution $P$ on the models
\[\begin{array}{l}
P(\Bm_1)=x\\
P(\Bm_2)=y\\
P(\Bm_3)=z\\
P(\Bm_4)=1-x-y-z.
\end{array}
\]

The fact that these nodes $\{a,b\}$ are part of the network $(S, R)$ of Figure \ref{540-100F} is reflected in the probability having to satisfy the equational approach equation called E3 in \cite{540-18} (see \cite{540-39} for the Equational Approach to argumentation).

$$ P(x) =P(\bigwedge_{yRx}\neg y).
\leqno (E3)
$$

In our case this means
\begin{itemize}
\item $P(a) =P(\neg b)$
\item $P(b) =P(\neg a)$
\end{itemize}
where
\begin{itemize}
\item $P(x) =\sum_{\Bm\vDash x} P(\Bm)$
\end{itemize}

Therefore the equations we get are
\begin{itemize}
\item $x+y=1-x-z$
\item $x+z= 1-x-y$
\end{itemize}

These two are the same equation, yielding
\[
1-x-y-z=P(\Bm_4) =x.
\]
Therefore any probability distribution $P$ of the form
\[\begin{array}{l}
P(a\wedge b) =x\\
P(a\wedge\neg b) =y\\
P(\neg a\wedge b) =z\\
P(\neg a\wedge\neg b) =x
\end{array}
\]
with $2x+y+z=1$ is a good one, respecting the attack relation.

We  have $P(a) =x+y, P(b) =x+z, P(\neg a \wedge\neg b) =x$.

Let us now check how the models of {\bf CN} relate to probability. For the language with atoms $\{a, Na, b, Nb\}$  satisfying the {\bf CN} axiom
\[
Nx\to \neg x
\]
we can have the following models $h_{i,j} =\alpha_i\wedge\beta_j, i=1,2,3, j=1,2,3$ where
\[\begin{array}{lcl}
\alpha_1 &=& a\wedge \neg Na\\
\alpha_2 &=& \neg a \wedge\neg Na\\
\alpha_3 &=& \neg a\wedge Na\\
\beta_1&=& b\wedge\neg Nb\\
\beta_2 &=& \neg b\wedge\neg Nb\\
\beta_3 &=& \neg b\wedge Nb.
\end{array}
\]
Here we have 9 models as opposed to 6 models of the previous approach.  Let the probabilities on these models be
\[ \pi(\alpha_i\wedge\beta_j)=P_{i,j}
\]
with
\[
\sum_{i,j} P_{i,j} =1.
\]
The models must satisfy the theory $\Delta_{\CA}$ for the network of Figure \ref{540-100F}.  These are
\[\begin{array}{l}
Na\to \neg a\\
Nb\to \neg b\\
a\leftrightarrow Nb\\
b\leftrightarrow Na\\
\neg a\wedge\neg Na\to \neg b\wedge\neg Nb\\
\neg b\wedge\neg Nb\to\neg a\wedge\neg Na.
\end{array}
\]
So only 3 models are models of $\Delta_{\CA}$. These are
\[\begin{array}{lcl}
h_1 &=& a\wedge\neg Na\wedge Nb\wedge\neg b\\
h_2 &=& \neg a\wedge Na\wedge b\wedge\neg Nb\\
h_3 &=& \neg a\wedge\neg Na\wedge\neg b\wedge\neg Nb.
\end{array}
\]

Let the relative probabilities be $p_1, p_2, p_3$ respectively with $\sum p_i=1$.

Note that $h_1$ is the extension $a=$ in, $b=$ out $h_2$ is the extension $b=$ in, $a=$ out and $h_3$ is the extension $a=b=$ und.

The first difference between the approaches is that while the approach in \cite{540-18} gives probability to arguments (called in \cite{540-18} ``the internal approach''), the {\bf CN} approach ends up giving probabilities to extensions (called in \cite{540-18} ``the external approach'').

Giving probability to extensions is not new.  It is supported by many authors (see references/discussion in \cite{540-18}). Let us calculate the probabilities $\pi(a)$ and $\pi(b)$.  We get
\[
\pi(a) =p_1, \pi(b)=p_2, \pi(\neg a\wedge\neg b) =p_3.
\]

Comparing the two approaches, it makes sense to equate
\[
p_3 =x, p_2=x+z, p_1=x+y.
\]

The main difference is that we are giving probabilities to 2-valued models in \cite{540-18} and using {\bf CN} we are using 3-valued models.

In both cases we give probabilities to models. However, in case of {\bf CN}, the models are extensions and so we are giving probabilities to extensions also.

\subsection*{CC7.   {\bf CN} and bipolar networks}
Note that in {\bf CN} we get support and contrary arguments for free. Since we have implication ``$\to$'' in the logic, we can write ``$x\to y$'' to mean ``$x$ supports $y$'' and since we have negation ``$\neg$'' in the logic, we can view $\neg x$ as the contrary of $x$. We need not introduce additional atoms into the argumentation network for contraries, nor do we have to introduce an additional arrow for support into the network. Furthermore, since we have negation, we have the additional option to represent ``$x$ supports $y$'' as ``$x\tO \neg y$'' namely as $x\to N\neg y$.

Let us do this in a systematic manner.
\begin{definition}\label{540-DB1}{\ }
\begin{enumerate}
\item A bipolar network $\CB$ has the form $\CB =(S, R_a, R_s)$, where $R_a\subseteq S\times S$ is the attack relation (also denoted by $\tO$) and $R_s\subseteq S\times S$ is the support relation (also denoted by $\To$).
\item Given a bipolar network $\CB$, we offer two possible translations into {\bf CN}.
\begin{itemize}
\item $\tau_1(x\tO y) =x\to Ny$
\item $\tau_1(x\To y) =x\to y$
\item $\tau_2(x\tO y) =x\to Ny$
\item $\tau_2(x\To y) =x\to N\neg y$
\end{itemize}
\item Note that the complete extensions of the bipolar networks will be obtained from all the models of {\bf CN}.
\end{enumerate}
\end{definition}

\begin{remark}\label{540-RB2}{\ }
\begin{enumerate}
\item We note that $x\to y$ implies $x \to N\neg y$. This holds because, as we shall see in Section 5, $\neg Ny\to N\neg y$ and because $y\to N\neg y$ both hold.  Thus $\tau_1$ is stronger than $\tau_2$.
\item When we have a translation $\tau$ from one system (e.g.\ $\CB$) into another  (e.g.\ {\bf CN}), we need to examine the properties of {\em soundness} and {\em completeness}.

Soundness in our case means that whatever we consider as a bipolar complete extension for $\CB$ will turn out to be a complete extension according to {\bf CN}.  Completeness means that {\bf CN} does not give any additional complete extensions.

The problem in this case is that there is disagreement in the community about how to define the complete extensions for $\CB$.  The main papers of C. Cayrol and M. C. Lagasquie-Schiex are \cite{540-28,540-29}.  These have been criticised by G. Brewka and S. Woltran \cite{540-11} and a solution was proposed in our paper \cite{540-26}.  Thus a detailed analysis of soundness and completeness for our translations must be postponed to a continuation paper \cite{540-24}.  However, we can point out in this paper some key properties involved in \cite{540-26}.

The possible properties are as follows:
\item [(P1)] $\displaystyle\frac{x\To y, z\tO y}{z\tO x}$
\item [(P2)] $\displaystyle\frac{x\To y, y\tO u}{x\tO u}$
\item [(P3)] $\displaystyle\frac{w\tO x, x\To y}{w\tO y}$

We need to check these properties for both
\[
\tau(x\tO y)=x\to y
\]
and
\[
\tau(x\To y) =x\to N\neg y.
\]
What we need is consistency. Can we add the translation of these rules to {\bf CN} and remain consistent?

The answer is positive. So we can hope to have something like the following theorem:

\end{enumerate}
\end{remark}

\begin{theorem}\label{540-TB3}
Let $\BBP$ be a set of properties for a bipolar network $\CB$ (e.g.\ (P1)--(P3) of Remark \ref{540-RB2}).  Let $\tau$ be a translation of $\CB$ into {\bf CN}, and let $\tau(\BBP)$ be the translation of $\BBP$ into {\bf CN}.  Then the translation $\tau$ of $\CB$ into ${\bf CN}+\tau(\BBP)$ is sound and complete.
\end{theorem}

As we said, we shall address this theorem in Part 2 of this paper.

\subsection*{CC8.  Limitations of the {\bf CN} approach}
The {\bf CN} approach transforms the geometrical representation of an argumentation network $\CA =(S, R)$ into a theory $\Delta_{\CA}$ of the logic {\bf CN}.  Theorem \ref{540-T4} ensures that the correspondence between complete extensions of $\CA$ and models of $\Delta_{\CA}$ is sound and complete. We are trading off here, however, geometry for model theory.

The previous CC1--CC4 discussed the advantages of the {\bf CN} approach.  The limitations come from the fact that in the {\bf CN} approach we obscure/lose the geometrical aspects of $\CA$.  Therefore any moves in argumentation theory which use the geometry (e.g.\ the strongly connected components, SCC of Baroni {\em  et al.} \cite{540-14}) will become less transparent. We can mathematically do them in {\bf CN}, but we would have to extract the geometry of $\CA =(S, R)$ back out of $\Delta_{\CA}$!

The following Figure \ref{540-F24}  can be used to illustrate our point using the CF2 semantics \cite{540-14,540-15}.

\begin{figure}
\centering
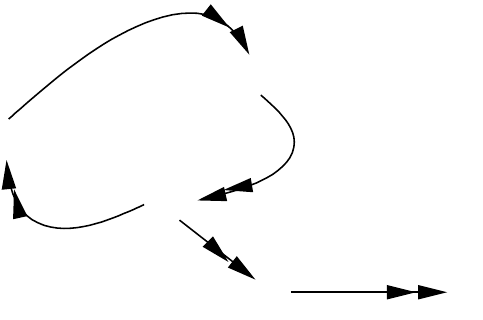
\caption{}\label{540-F24}
\end{figure}

The {\bf CN} semantics gives the complete extension of all undecided to the network $\CA$ of Figure \ref{540-F24} in agreement, of course, with the traditional Dung approach. The CF2 semantics takes maximal conflict free subsets of the top SCC and therefore yields the extensions $\{a,d\}, \{b,d\}, \{c,e\}$.  CF2 relies on identifying the SCCs.  It relies on the geometry of $(S, R)$.

The {\bf CN} theory for this network $\CA$ is
\[\begin{array}{lcl}
\Delta_{\CA}&=&\Theta_\Bn\cup \{a\to Nb, b\to Nc,c\to Na, c\to Nd,d\to Ne,e\leftrightarrow Nd,\\
&&d\leftrightarrow Nc,c\leftrightarrow Nb,b\leftrightarrow Na,a\leftrightarrow Nc\}
\end{array}
\]

Looking at $\Delta_{\CA}$, we have to define/extract the cycles in order to define the CF2 semantics for $\CA$.

Gabbay's approach using annihilators \cite{540-14} also requires the use of geometry but to a lesser extent.  We can identify syntactically a cycle, say

\[
x_1\to Nx_2, x_2\to Nx_3\comma x_n\to Nx_1
\]
and break the cycle by applying an annihilator, say to annihilate the point $x_i$, i.e.\ add a new point $z(x_i)$ to $\Delta_{\CA}$ with $z\to Nx_i$, and look instead at
\[
\Delta^{z(x_i)}_{\CA} = \Delta_{\CA}\cup \{z(x_i), z(x_i)\to Nx_i\}.
\]
We still need some geometric intuition in doing this.

\section{
Introducing the logic {\bf CNN}}
\subsection{A meta-level object level short discussion}
The perceptive reader may be aware of Section 5 of my 2009 paper with Caminada \cite{540-1}.  In this paper we describe several options of looking at an abstract argumentation network $(S, R)$. Since 2009, there were   many other papers, which put forward  different  representations of  Abstract Argumentation Networks in terms of well known logics, see for example \cite{540-32}--\cite{540-37}. We shall compare and discuss these papers in our Comparison with the Literature Section~\ref{comparison}. Meanwhile in  this section,  we want to make a critical point about the difference between Object Level Vs Meta-Level representation of Abstract Argumentation Networks, and so we consider now one of the possible options of Section 5 of \cite{540-1}. This  option is to describe $(S, R)$ completely in classical predicate logic. We consider $S$ as the domain of the logic and we consider the attack relation $R$ as a binary relation on $S$.  In addition to that, we need 3 unary predicates, $Q_1, Q_0$ an $Q_2$, representing the 3 Caminada labels for the elements of $S$, namely $x$ is ``in'', $x$ is ``out'' and $x$ is ``undecided'', respectively.

We write axioms in predicate logic basically expressing the properties of the labelling relative to $R$ making it a legitimate labelling.  These are the following:

Consider the following classical theory $\Delta (R, Q_0, Q_1, Q_?)$.
\begin{enumerate}
\item $\forall x (Q_0(x)\vee Q_1(x)\vee Q_?(x))$
\item $\neg\exists x(Q_i(x)\wedge Q_j(x))\mbox{ for } i\neq j, i, j\in \{0,1,?\}$
\item $\forall y(\forall x(xRy\to Q_0(x))\to Q_1(y)$
\item $\forall y (\exists x (xRy\wedge Q_1(x))\to Q_0(y))$
\item $\forall y(\forall x (xRy\to (Q_0(x)\vee Q_?(x)))\wedge \exists x (xRy\wedge Q_?(x))\to Q_?(y))$
\end{enumerate}

Any model of $\Delta$ with domain $D$ defines an argumentation framework with the set of argument $S = D$, and the attack relation is $R$ and the labelling $\lambda_\Delta$ is what we obtain from the elements satisfying the respective predicates $Q_0, Q_1$ and $Q_?$.  Notice that we are not using ``$=$''.

If we want to characterise any specific argumentation framework $\CA =(S, R)$ we need equality and we need constant names for every element of $S$. We write the following additional axioms $\CA$
\begin{enumerate}
\setcounter{enumi}{5}
\item $\forall x(\bigvee_{a\in S}x=a)$
\item $\bigwedge_{a,b\in S, a\neq b} a\neq b$
\item $\bigwedge_{a,b\in R} aRb$
\end{enumerate}

The use of predicate logic to talk about $(S, R)$ is meta-level. The predicates $Q_1, Q_0$ and $Q_?$ are not logical connectives. We cannot write, for example, the expression $Q_0(Q_0(x))$. In contrast, the logic {\bf CN} is object level. It can express the predicates $Q_1, Q_0$ and $Q_?$ in the object level, as well as the relation $R$, as follows (see, however, Remark~\ref{RMarch22} of Section~\ref{540-S3}):
\[\begin{array}{l}
Q_1(x) =\mbox{def } x\\
Q_0(x)=\mbox{def } Nx\\
Q_?(x)=\mbox{def } \neg x\wedge\neg Nx \\
xRy=\mbox{def } x\to Ny.
\end{array}
\]
Note the difference between object level and meta-level. Suppose we want to instantiate $(S, R)$ with arguments which are formulas of predicate logic, say we have  $xRy$ and we instantiate $x=\alpha$ and  $y=\beta$. In the meta-level language we cannot write $\alpha R \beta$. Even if we allow for the use of names ``$\alpha$" and ``$\beta$" and add appropriate axioms for the correct handling of names, we still do not know what ``$\alpha$" $R$ ``$\beta$" means in terms of semantics. The meta-level translation does not give any meaning to $R$ it is just a translation. On the other hand , we can write in our object level system the formula $ \alpha \to N\beta$, and if we use predicate logic with $N$, we can let the models of this predicate logic with $N$ provide a proposed semantics for the instantiation by using Theorem \ref{540-T4} as a definition.
Furthermore we can write additional axioms for $N$. For example   we can also write $NNx$ and add an axiom (which is valid)
$$
x\leftrightarrow NNx
\leqno {\bf NN}:
$$
On the other hand the meta-level predicate approach can deal with infinite non-finitary networks , while our object level approach requires the network to be finitary. We shall discuss the predicate approach further in the context of comparing with paper \cite{540-35} in Section 9.

\subsection{The logics \BC\BN\ and \BC\BN\BN}

Having explained the object level nature of our translation of argumentation into {\bf CN}, let us now focus on this logic.

So what kind of a logic is {\bf CN}?  Can we say more about it, in addition to what we said in Section 1?  The answer is yes.  We can add a modal point of view to that of Section 1.

Consider a modal logic with two possible worlds \Bt\ and \Bs.  Assume that \Bs\ is accessible to \Bt\ and \Bt\ is accessible to \Bs.  So we have a symmetric non-reflexive relation. See Figure \ref{540-F99}.

\begin{figure}
\centering
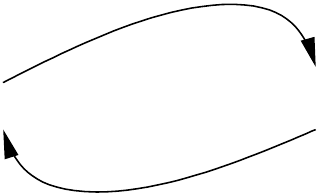
\caption{}\label{540-F99}
\end{figure}

We further require on assignments to atoms that if the atom $q$ is assigned $\top$ at \Bt, then it is assigned $\top$ at \Bs\ (but not necessarily vice versae).

Let \Bt\ be the actual world and let us have the modal connective $N$ with the following truth table
\begin{itemize}
\item $\Bt \vDash NA$ iff $\Bs\vDash \neg A$
\item $\Bs\vDash NA$ iff $ \Bt\vDash \neg A$
\item $\vDash A$ iff $\Bt\vDash A$ (i.e.\ \Bt\ is the actual world)
\end{itemize}
Let us call this modal logic {\bf CNN}.

We therefore have the following true in the actual world \Bt
\begin{itemize}
\item ${\bf CNN} \vDash Nq\to \neg q, q$ atomic
\item ${\bf CNN}\vDash NNA\leftrightarrow A$
\item ${\bf CNN} \vDash N(A\wedge B)\leftrightarrow NA\vee NB$
\item ${\bf CNN}\vDash N(A\vee B)\leftrightarrow NA\wedge NB$
\item ${\bf CNN}\vDash N(N\neg q\to q)$
\item ${\bf CNN}\vDash (A\to B)\to N(NB\to NA)$
\end{itemize}

In this context for atomic arguments $q$ we understand the following in \Bt:
\begin{itemlist}{(*33)}
\item [(*1)] $q$ attacks $p$, reads $q\to Np$
\item [(*2)] $q$ is out, reads $Nq$
\item [(*3)] $q$ is in, reads $q$
\item [(*4)] $q$ is undecided, reads $\neg q \wedge\neg Nq$.
\end{itemlist}

\begin{definition}\label{540-D97}
Let {\bf CNN} be the logic extending {\bf CN} with the following axioms
\[
\begin{array}{ll}
{\bf NN}:& NNx\leftrightarrow x\\
{\bf N}\wedge: & N(x\wedge y)\leftrightarrow Nx\vee Ny\\
{\bf N}\vee:&N(x\vee y)\leftrightarrow Nx\wedge Ny\\
\BN\leftrightarrow:&x\leftrightarrow y \mbox{ implies } Nx \leftrightarrow Ny \mbox{, where $IN$ does not occur in $x$ or $y$\footnotemark}\\
{\bf N}\neg:&N\neg x \leftrightarrow \neg Nx\\

\end{array}
\]
\end{definition}
\footnotetext{Note that we have $\top  \leftrightarrow (Nx\to \neg x)$ but not  $N\top  \leftrightarrow N(Nx\to \neg x)$.}

\begin{example}\label{540-E97}
Consider the network of Figure \ref{540-F98}

\begin{figure}
\centering

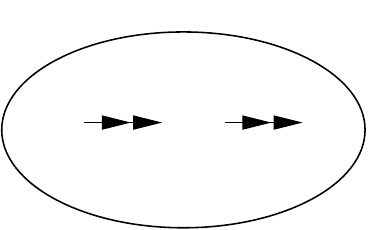
\caption{}\label{540-F98}
\end{figure}

The Figure describes a network containing the points $x,y,z$ where $x$ is the only attacker of $y$ and $y$ is the only attacker of $z$. Given that, we know that no matter what complete extension we are dealing with, the value of $x$ must be equal to the value of $z$.
\begin{itemize}
\item $x=$ in $\To y =$ out $\To z =$ in
\item $x=$ out $\To y =$ in $\To z =$ out
\item $x=$ und $\To y =$ und $\To z =$ und
\end{itemize}

If we look at $\BF_{\rm in}, \BF_{\rm out}$ and $\BF_{\rm und}$ applied to these points, we see that the formulas
\begin{enumerate}
\item $x\rightarrow Ny$
\item $y\rightarrow Nz$
\item $y \leftrightarrow Nx$
\item $z \leftrightarrow Ny$
\item $\neg x\wedge\neg Nx\rightarrow \neg y\wedge\neg Ny$
\item $\neg y\wedge\neg Ny\rightarrow \neg z\wedge\neg Nz$
\end{enumerate}
must hold in the logic. We also have, of course, the axiom schema $Nu\to\neg u$.

These formulas must prove $x\leftrightarrow z$  and $Nx \leftrightarrow Nz$. We now proceed to prove them:
\begin{enumerate}[label=\alph*]
\item
Assume $x$, then from 1. and 4. we get $z$.
\item
Assume $z$, then from 4. we get $Ny$ and from the axiom we get $\neg y$ and so from 3. we get $\neg Nx$. If we had $\neg x$ then from 5. and 6. we would have got $\neg z$. Therefore we must have $x$.
\item
Assume $Nx$, then from 2. and 3. we get $Nz$.
\item
Assume $Nz$, then from axiom we get $\neg z$, and therefore from 4. and 1. we get $\neg x$. If we had $\neg Nx$, we would have got from 5. and 6. $\neg Nz$, and so we must have $Nx$.
\end{enumerate}
We thus proved  $x \leftrightarrow z$ from a. and b. and we proved  $Nx \leftrightarrow Nz$ from c. and d..

%
%
\end{example}

\begin{example}\label{540-E98}
Consider the network $\CA =(S, R)$ of Figure \ref{540-F97}
\begin{figure}
\centering
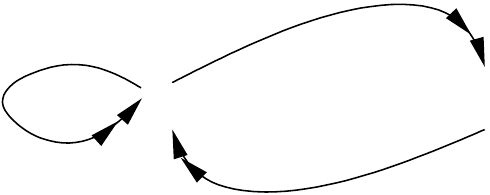
\caption{}\label{540-F97}
\end{figure}

The axioms of $\Delta_{\CA}$ are the following:
\begin{enumerate}
\item $x\leftrightarrow Nx\wedge Ny$
\item $y\leftrightarrow Nx$
\item $x\to Nx$
\item $x\to Ny$
\item\begin{enumerate}
\item $(\neg x \wedge\neg y)\wedge(\neg Nx\vee \neg Ny)\to \neg x\wedge\neg Nx$
\item $\neg x \wedge\neg Nx\to \neg y\wedge\neg Ny$
\end{enumerate}
\end{enumerate}
The network $\CA$ has two possible extensions $y=$ in, $x=$ out, and $x=y=$ und.  So we want to be able to prove (using $\Delta_{\CA}$ and the extra axiom $NNA \leftrightarrow A$)  that either $y=$ in and $x=$ out or $y=x=$ und.  This means we want to prove (6)
\begin{enumerate}
\setcounter{enumi}{5}
\item $(y\wedge Nx) \bigvee (\neg y \wedge\neg Ny \wedge\neg x \wedge\neg Nx)$
\end{enumerate}
Let us show we can do it!

From (1) we get $x \to Nx$ and since we  have the axiom $Nx \to \neg x$ we get
\begin{enumerate}
\setcounter{enumi}{6}
\item $\neg x$
\end{enumerate}
We now prove
\begin{enumerate}
\setcounter{enumi}{7}
\item $Ny \rightarrow \neg Ny$
\end{enumerate}
Assume $Ny$, therefore by axiom we get $\neg y$ and so from (2) we get $\neg Nx$, and so from from (7) and (5b) we get $\neg Ny$. Thus we have proved (9) below:
\begin{enumerate}
\setcounter{enumi}{8}
\item $\neg Ny$.
\end{enumerate}
Therefore, in view of (9) and (7), (6) becomes (6*)
\begin{enumerate}
\item[6*.]
$(y\wedge Nx)\bigvee(\neg y \wedge\neg Nx)$
\end{enumerate}
or, equivalently (6**)
\begin{enumerate}
\item [6**.] $(y\vee Nx)\to y\wedge Nx$.
\end{enumerate}
In view of (2), (6**) becomes
\begin{enumerate}
\item [6***.] $(y\vee y)\to (y\wedge y)$.
\end{enumerate}
which holds.
\end{example}

\section{Conjunctive and disjunctive attacks in {\bf CN}}
In 2009 Gabbay \cite{540-3} introduced the option of conjunctive (joint) attacks and the notion of disjunctive attacks. The basic configuration can be seen in Figures \ref{540-F25a}, \ref{540-F25b},  \ref{540-F25c} and \ref{540-F25d}.  In Figure \ref{540-F25c} the nodes $\{y_1\comma y_m\}$ jointly mount a disjunctive attack on the nodes $\{z_1\comma z_n\}$. The
meaning in Figure \ref{540-F25c} is that if all of $y_i$ are in then at least one of $z_j$ is out.   See Figure \ref{540-F25d}.  Translated into {\bf CN} this means:
\[\bigwedge^m_{i=1} y_i\to \bigvee ^n_{j=1} Nz_j
\]

\begin{figure}
\centering
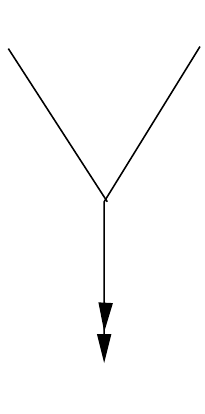
\caption{}\label{540-F25a}
\end{figure}

\begin{figure}
\centering
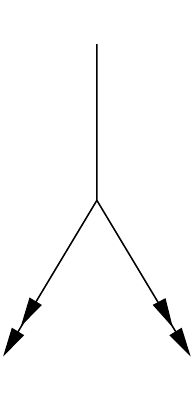
\caption{}\label{540-F25b}
\end{figure}

\begin{figure}
\centering
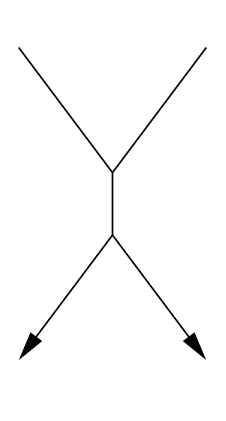
\caption{}\label{540-F25c}
\end{figure}

\begin{figure}
\centering
\[
\bigwedge^m_{i=1} y_i\tO \bigwedge^n_{j=1} z_j
\]
\caption{}\label{540-F25d}
\end{figure}

\begin{definition}[CD network]\label{540-D26}{\ }
\begin{enumerate}
\item Let $S$ be a non-empty set of arguments. Consider the set $\Omega$ of all non-empty subsets of $S$. Use the notation $G, H\subseteq S$ for such subsets. Let $\BBR$ be a binary relation on such  subsets, i.e.
\[
\BBR\subseteq (2^S-\varnothing)^2.
\]
A conjunctive-disjunctive argumentation network (CD-network) has the form $(S, \Omega, \BBR)$.
\item When all attacked sets are singletons we get only joint attacks of the form
\[
G\BBR \{x\}, x\in S
\]
We can use the notation
\[
R_0=\{(G,x)|G\BBR\{x\}\}.
\]
\item When all attacking sets are singletons we get only disjunctive attacks of the form $\{x\}\BBR H$.
\end{enumerate}
\end{definition}
\begin{remark}\label{540-R26}
We need to discuss the nature of the joint-disjunctive attack.  The disjunctive part of the definition is different in nature from the conjunctive part. This may cause difficulties.  Consider the disjunctive attack of Figure \ref{540-F26a}.  $x$ is not attacked and so $x=$ in. Thus we must have either $a=$ out or $b=$ out.  So the extensions are

\begin{figure}
\centering
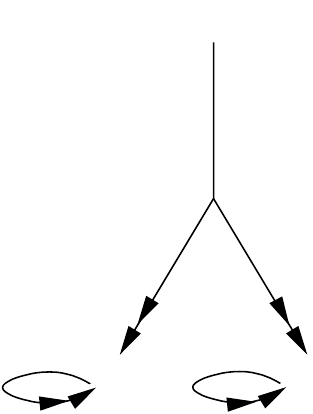
\caption{}\label{540-F26a}
\end{figure}

\begin{center}
$x=$ in, $b=$ out. $a=$ und\\
$x=$ in, $a=$ out, $b=$ und.
\end{center}

In {\bf CN} we can write this as
\[
x\to Na\vee Nb.
\]
The reader should be strongly aware that the translation is from $\CA=(S, R)$ into {\bf CN} and {\em not} the other direction.

So if we have a joint attack of $\{x, a\}\tO b$, then the translation would be $x\wedge a\to Nb$.

To make this point clearer note that the condition $x \to Na\vee Nb$ is equivalent to the following two conditions
\begin{enumerate}
\renewcommand{\labelenumi}{\roman{enumi}.}
\item $x=$  in and $a\neq $ out $\to b=$ out\\
$x\wedge\neg Na\to Nb$
\item $x=$ in and $b\neq $ out $\to a=$ out\\
$x\wedge\neg Nb\to Na$.
\end{enumerate}
Consider now (i) as a new kind of joint attack as in Figure \ref{540-F26b}.

\begin{figure}
\centering
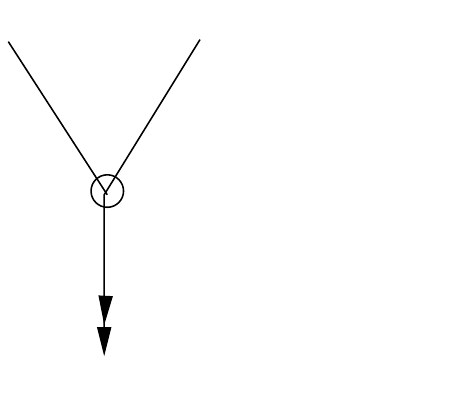
\caption{}\label{540-F26b}
\end{figure}

\newcommand{\diamondto}{\mbox{---}\lozenge\!\!\!\longrightarrow}
In this new kind of joint attack if one or more of the attackers is undecided, then the attack still goes through. This also means that we get a new type of single attack which  we denote by
\[
x \mbox{---}\lozenge\!\!\!\longrightarrow y
\]
where $y$ is out even when $x=$ undecided.

In fact, we can define $x \diamondto y$ as the truth $\top$ disjunctively attacking $x\wedge y$. This view also implies that $x\diamondto y$ iff $y\diamondto x$.

This is not a good understanding of disjunctive attacks. Our intuitive understanding of being undecided is that $x$ is undecided because $x$ can go either way.  Either to $x=$ in or to $x=$ out.   So we might want to say that in Figure \ref{540-F25b}, if all $z_i$ are either in or undecided, with one $z_i$ undecided, then $\bigwedge_i z_i$ is undecided and  hence $z$ is undecided.

It is not our purpose to study in detail disjunctive attacks. We just want to show that the logics {\bf CN} and {\bf CNN} are very good in expressing attacks. The reader can see for a fact that in the {\bf CN} language we can say whatever we want. The attack $x\diamondto y$ can be written as $\neg Nx\to Ny$ or equivalently $\neg Ny \to Nx$.

There is no agreement in the literature on the meaning of the argument $z$ attacking the set of arguments $H=\{z_1\comma z_n\}$.  In \cite{540-3,540-16} the definition is $z\to \bigvee_j Nz_j$, while in \cite{540-4} the definition is $\bigvee_j (z\to Nz_j)$. See Section 8 below for further discussion and also see \cite{540-23}.

These two options are the same in two valued logic but not if we have the third undecided value.

Given the uncertainty of how to define complete extensions for networks with disjunctive attacks, we can simply write the {\bf CN}-formulas $\BF_{\rm in}, \BF_{\rm out}, \BF_{\rm und}$ that we want and {\em stipulate} that the complete extensions are to be all {\bf CN} models of these formulas. For further discussion see Section 8 below.
\end{remark}
Having concluded our discussion in the previous Remark~\ref{540-R26}, we still  do need to prove a correspondence theorem for the case of joint attacks only, since there is agreement about them in the literature.

\begin{definition}\label{540-D201}{\ }
\begin{enumerate}
\item An argumentation network with joint attacks has the form $(S, R_0)$, where $S$ is a non-empty set of arguments and $R_0\subseteq (2^S-\varnothing)\times S$ is the joint attack relation.
\item A legitimate Caminada--Gabbay labelling from $S$ into \{in, out, und\} is a function $\lambda$ satisfying the following conditions:
\begin{enumerate}
\item [(CG1)]  $\lambda(x)=$ in iff either $x$ is not attacked by any $G\subseteq S$ or for every $G\subseteq S$ such that $GR_0x$, there exists a $y\in G$ such that $\lambda(y)=$ out.
\item [(CG2)] $\lambda(x)=$ out iff for some $G\subseteq S$, $GR_0 x$ and for all $y \in G, \lambda (y)=$ in.
\item [(CG3)] $\lambda (x)=$ und iff for all $G\subseteq S$, such that $GR_0 x$ there exists a $y\in G$ such that $\lambda (y)\neq$ in,  and for some $G'\subseteq S$, $GR_0 x$ we have that for all $y\in G', \lambda (y)=$ either in or und.
\end{enumerate}
\item We identify the complete extensions of $(S, R_0)$ with the legitimate Caminada--Gabbay labellings of $S$.
\end{enumerate}
\end{definition}

\begin{theorem}\label{540-T202}
Let $\CA=(S, R_0)$ be an argumentation network with joint attacks.  Let $\Delta_{\CA}$ be its associated {\bf CN} theory with $\BF^{\rm Joint}_{\rm in}, \BF^{\rm Joint}_{\rm out}$ and $\BF^{\rm Joint}_{\rm und}$ as defined in CC2 of Section 4 and listed below.

Then the models of $\Delta_{\CA}$ correspond exactly to the Caminada--Gabbay labellings of $\CA$ according to Definition \ref{540-D201}.
$$
\begin{array}{c}
(\bigvee_{GR_0x}\bigwedge_{z\in G} z)\to Nx\\
\mbox{for all } x\in S
\end{array}
\leqno
\BF^{\rm Joint}_{\rm out}:
$$

$$
\begin{array}{c}
x\leftrightarrow (\bigwedge_{GR_0x} \bigvee_{z\in G} Nz)
\\\mbox{for all } x\in S
\end{array}
\leqno
\BF^{\rm Joint}_{\rm in}:
$$

$$
\bigwedge_{GR_0 x}(\bigvee_{z\in G}\neg z)\wedge(\bigvee_{GR_0 x} \bigwedge_{z\in G} (z\vee \neg Nz))\to \neg x \wedge\neg Nx
\leqno
\BF^{\rm Joint}_{\rm und}:
$$
\end{theorem}

\begin{proof}
\paragraph{Part A.}  Assume we have a legitimate $\lambda$. We define a model $h=h_\lambda$ of $\Delta_{\CA}$.

Let
\begin{itemize}
\item $h_\lambda (x) =1$ if $\lambda (x)=$ in
\item $h_\lambda(Nx) =1$ if $\lambda(x)=$ out
\end{itemize}
We show that all axioms of {\bf CN} hold.

\begin{enumerate}
\item Show $Nx \to \neg x$. Otherwise $h_\lambda (Nx) =h_\lambda (x) =1$, but this means $\lambda (x) =$ in = out which is impossible.
\item We check $\BF^{\rm Join}_{\rm out}$.  Assume $h_\lambda (\bigvee_{GR_0 x} \bigwedge_{z\in G} z) =1$.  So for some $GR_0 x$ we have $h_\lambda (z) =1$ for all $z\in G$. This means for this $G$ that $\lambda (z) =$ in for all $z\in G$. Hence $\lambda (x)=$ out. Hence $h_\lambda (Nx) =1$.
\item We check $\BF^{\rm Joint}_{\rm in}$.
\begin{enumerate}
\item Assume $h_\lambda (x)=1$.  This means $\lambda (x) =$ in.  Hence for any $GR_0 x$ we have that for some $z\in G, \lambda (z)=$ out. Thus for all $GR_0x$ there is a $z\in G$ s.t. $h(Nz) =1$, i.e.\ $h(\bigwedge_{GR_0 x}\bigvee_{z \in G} Nz) =1$.
\item Suppose $h_\lambda (\bigwedge_{GR_0 x}\bigvee_{z\in G} Nz) =1$.  This means for every $GR_0 x$ there is a $z$ such that $h_\lambda (Nz) =1$. But this means for every $GR_0x$ there is a $z\in G$ s.t. $\lambda (z)=$ out. Thus $\lambda (x)=$ in so $h_\lambda(x)=1$.
\end{enumerate}
\item We check $\BF^{\rm Joint}_{\rm und}$:  Assume that
\begin{enumerate}
\item $h_\lambda (\bigwedge_{GR_0 x}\bigvee_{z\in G}\neg z) =1$
\item $h(\bigvee_{GR_0 x} \bigwedge_{z\in G}(z\vee \neg Nz))=1$.
\end{enumerate}
We show that $h(\neg x\wedge\neg Nx) =1$.

From (a) we get that every $GR_0 x$ has a $z$ such that  $h_\lambda (z) =0$, i.e.\ $\lambda (z)\neq$ in.

From (b) we get that there is a $GR_0 x$ such that for all $z\in G$, we have that  $\neg z\to \neg Nz$ holds.

Thus for one $GR_0 x$ there is a $z$ such that $h(\neg z\wedge \neg Nz) =1$.  This means $\lambda (z) \neq$ in and $\lambda (z) \neq$ out, i.e.\ $\lambda (z)=$ undecided.

So we got that for every $GR_0 x$, not all $\lambda(z)=$ in, for $z\in G$ and for one such $G'$ all $z$ are either in or undecided with one $\lambda (z)=$ und.

This implies $\lambda (x)=$ und.  So $h_\lambda (\neg x) =h_\lambda (\neg Nx) =1$.
\end{enumerate}

\paragraph{Part B.}  Assume $h$ is a model of $\Delta_{\CA}$. Define $\lambda_h$
\begin{itemize}
\item $\lambda_h (x)=$ in if $h(x)=1$
\item $\lambda_h(x)=$ out if $h(Nx) =1$
\item $h_\lambda(x)=$ und if $h(\neg x\wedge\neg Nx) =1$.
\end{itemize}
We show $\lambda_h$ is a legitimate labelling.
\begin{enumerate}
\item $\lambda_h$ is well defined because $h$ satisfies
\[
h\vDash Nx\to \neg x.
\]
\item We show (CG1)
\begin{enumerate}
\item If $x$ is not attacked then $x\in \Delta_{\CA}$ and so $\lambda_h(x)=$ in.
\item Assume that for every $GR_0 x$ there is $y\in G$ with $\lambda_h(y)=$ out.
\end{enumerate}
This means for every $GR_0 x$ there is $y\in G$ with $h(Ny)=1$.  So $h(\bigwedge_{GR_0 x}\bigvee_{y\in G} Ny)=1$.

So by $\BF^{\rm Joint}_{\rm in}$ we get that $h(x)=1$ i.e. $\lambda (x)=$ in.
\item We show (CG2).  Assume for some $GR_0 x$ and all $y\in G$ that $\lambda_h(y)=$ in. This means
\[
h(\bigvee_{GR_0x} \bigwedge_{y\in G}y)=1.
\]
Hence by $\BF^{\rm Joint}_{\rm out}, h(Nx)=1$, i.e.\ $\lambda (x)=$ out.
\item We show (CG3).  Suppose for every $GR_0 x$ there is a $y$ s.t.\ $\lambda (y) \neq $ in and in some $G'R_0 x$ for $y\in G', \lambda (y)=$ in or undecided and that there is a $y_0\in G'$ such that $\lambda (y_0)=$ und.

Then $h(\bigwedge_{GR_0 x}\bigvee_{y\in G}\neg y) =1$ and $h(\bigvee_{GR_0 x}\bigwedge_{y\in G}(y\vee \neg Ny)) =1$.

This means that the antecedent of $\BF^{\rm Joint}_{\rm und}$ holds and so $h(\neg x\wedge\neg Nx) =1$ and so $\lambda (x)=$ und.

\end{enumerate}
\end{proof}

\begin{theorem}\label{540-T203}
Let $\CA=(S, R_0)$ be a argumentation network with joint attacks. Then there exists a Caminada--Gabbay legitimate labelling $\lambda$ for it, (i.e.\ it does have complete extensions).
\end{theorem}

\begin{proof}
We reproduce a construction from Gabbay's 2009 paper \cite{540-3} which faithfully reduces joint attacks to single attacks.  This means that $(S, R_0)$ can be faithfully embedded into a traditional network $(S', R')$ with $S\subseteq S'$.  Since $(S', R')$ has complete extensions, we will get that $(S, R_0)$ has complete extensions.

We now proceed to generate the new points to add to $S$ to obtain $S'$, and we define $R'$.

Let $x\in S$. Let $G_1\comma G_n$ be all the attacking sets of $x$. Let $G_i =\{z_{i,1}\comma z_{i,r(i)}\}, i=1\comma n$. Introduce new points as follows:
\[
S_x= \{x(G_i), \Be(x,G_i,z_{i,1})\comma \Be(x,G_i, z_{i,r(i)}|i=1\comma n\}.
\]

Let $S'=S\cup\bigcup_{x\in S} S_x$.

Note that for any configuration $(G_1\comma G_n, x)$, where $G_i$ are all the sets attacking $x$, the sets $S_x$ are all sets of distinct disjoint points.  Note also the correspondence between each $z_{i,j} \in G_i$ and the point $\Be(x, G_i, z_{i,j})$.  So, for example, if $\{a,b\}R_0 x$ and $\{a,b\}R_0 y$ hold and $\{a,b\}$ is the only attacker of $x$ and the only attacker of $y$, the new points will be
\[
\begin{array}{lcl}
S_x&=& \{x(\{a,b\}),\Be(x,\{a,b,\},a),\Be(x,\{a,b\},b)\}\\
S_y&=& \{y(\{a,b\}),\Be(y,\{a,b\},a),\Be(y,\{a,b\},b)\}.
\end{array}
\]

Define $R'$ on $S'$ as follows.

For any $G_1\comma G_n$ attacking any $x$, let
\[
\begin{array}{l}
z_{i,j}R'\Be(x,G_i.z_{i,j})\\
\Be(x,G_i,z_{i,j})R'x(G_i)\\
x(G_i) R'x.
\end{array}
\]

Figures \ref{540-F204} and \ref{540-F205} show what we have down for the case of $x\in S$ and $G_i$ are the only joint attacker of $x$ with $G_i=\{z_{i,1}\comma z_{i, r(i)}\}, i=1\comma n$.

\begin{figure}
\centering
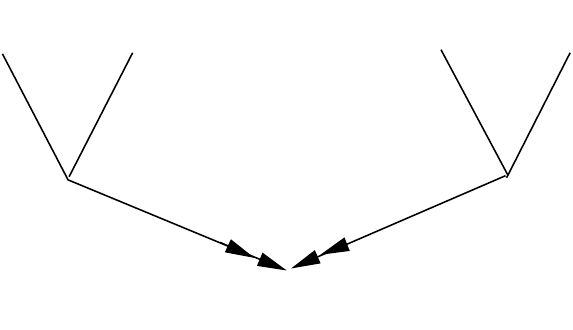
\caption{Joint attacks of $G_1\comma G_n$}\label{540-F204}
\end{figure}

\begin{figure}
\centering
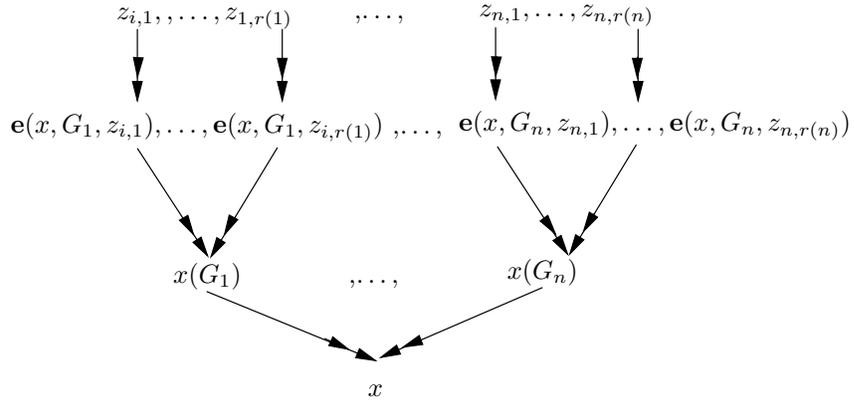
\caption{Reduction of the joint attacks }\label{540-F205}
\end{figure}

We want to prove the following claims.

\paragraph{Claim 1.}  Let $\lambda$ be a legitimate labelling of $(S', R')$. Then $\mu =\lambda\upharpoonright S$ is a legitimate labelling of $(S, R_0)$.

\paragraph{Claim 2.}  Let $\mu$ be a legitimate labelling of $(S, R_0)$. Then $\mu$ can be uniquely extended to a labelling $\lambda$ of $(S', R')$ and this $\lambda$ is legitimate for $(S', R')$.

To prove the above we need some shorthand definitions.

Let $G$ be a set of nodes. We say that $G$ is in if every $x\in G$ is in. We say that $G$ is out if some $x\in G$ is out and we say $G$ is undecided if every $x$ in $G$ is not in and some $y\in G$ is undecided.

\paragraph{Proof of Claim 1 and Claim 2.}  This is proved in \cite{540-3}.  We can see the idea of the proof by comparing Figures \ref{540-F204} and \ref{540-F205}.

\begin{enumerate}
\item Suppose $x$ is in, in Figure \ref{540-F205} then we must have that all of $x(G_i)$ are out. For $x(G_i)$ to be out, one of $\{\Be(x,G_i,z_{i,j})|j=1,2,\ldots\}$ must be in. So one of $\{z_{i,j}|j=1,\ldots\}$ must be out and so the set $G_i$ must be out.

This means that $x=$ in implies that all $G_i$ are out. In fact we have that $x$ is in if and only if all $G_i$ are out. But this is exactly the condition for $x=$ in, in $(S, R_0)$.
\item Suppose $x$ is out in $(S',R')$. Then some $x(G_i)$ is in. Therefore all of $\{\Be(x,G_i, z_{i,j} |j=1,\ldots\}$ for this $G_i$ are out.  Hence all of $\{z_{i,j}\}$ for this $G_i$ are in, and hence $G_i$ is in. Again this goes in both (iff) direction and is the out condition in $(S, R_0)$.
\item Suppose $x$ is undecided in $(S', R')$. Then all of $x(G_i)$ are either out or undecided, with one $x(G_i)$ at least being undecided. If $x(G_i)$ is out then as we have seen before, $G_i$ must be in. If $x(G_i)$ is undecided then all elements of $G_i$ must be either in or if not, in then undecided, with at least one element of $G_i$ being undecided. Again this argument is an if and only if argument and is exactly the condition for $x=$ und in $(S, R_0)$.
\item The considerations in (1), (2), (3) also show that Claim 2 is true because the values of $G_i$ and $x$ determine the values of $\Be(x, G_i, z_{i,j})$ and $x(G_i)$ uniquely.
\end{enumerate}
The theorem follows from the claims, since $(S', R')$ does have legitimate labellings.

\begin{remark}\label{540-R206}
The considerations in the proof of Theorem \ref{540-T203} show that basically the configuration in Figure~\ref{540-F205} ``implies uniquely'' the configuration of Figure \ref{540-F204}.  This means that if we write in the logic {\bf CNN} the conditions $\BF'_{\rm in}, \BF'_{\rm out}, \BF'_{\rm und}$ for $x$  and for all the additional new points  in $(S', R')$, it should logically imply in {\bf CNN}, the conditions $F^{\rm Joint}_{\rm in}, \BF^{\rm Joint}_{\rm out}$ and $\BF^{\rm Joint}_{\rm und}$ for $x$ in $(S, R_0)$. This means the new points play no real logical role.

Let us prove this, also as an exercise in the deduction capabilities of {\bf CNN}.

Let us write all the $\BF'$ of $x$ for $(S', R')$ and prove all the $\BF^{\rm Joint}$ of $x$ for $(S, R_0)$.

We are given the following data.
\begin{enumerate}
\item $x\leftrightarrow \bigwedge_i Nx(G_i)$
\item $\bigvee_i x (G_i)\to Nx$
\item $x(G_i) \leftrightarrow \bigwedge_j N\Be(x, G_i, z_{i,j})$
\item $\bigvee_j\Be(x,G_i ,z_{i,j})\to Nx(G_i)$
\item
\begin{enumerate}
\item $z_{i,j} \to N\Be(x.G_i,z_{i,j})$
\item $\neg z_{i,j}\wedge\neg Nz_{i,j}\to \neg\Be(x,G_i,z_{i,j})\wedge\neg N\Be(x,G_i,z_{i,j})$
\end{enumerate}
\item $\Be(x,G_i,z_{i,j})\leftrightarrow Nz_{i,j}$
\item
\begin{enumerate}
\item $\bigwedge_i\neg x(G_i)\wedge\bigvee_i\neg Nx(G_i)\to \neg x\wedge\neg Nx$
\item $\bigwedge_i\neg \Be(x,G_i,z_{i,j})\wedge\bigvee_j \neg N\Be(x,G_i,z_{ij})\to \neg x (G_i)\wedge\neg Nx(G_i)$
\end{enumerate}
\end{enumerate}
\noindent
Let us see what we can prove in {\bf CNN} from (1)--(7).
%
%
%
%
We start by proving $\BF^{\rm Joint}_{\rm in}$:
$$
x\leftrightarrow \bigwedge_{GR_0x}\bigvee_{z\in G}Nz
$$
\begin{enumerate}\setcounter{enumi}{7}
\item
\begin{enumerate}
\item
Assume $x$
\item
from (1) we get for all $i$, $Nx(G_i)$.
\item
From (8b) we get $\neg x(G_i)$.
\item
Fix an arbitrary $i$. We show that for at least one $i$ we have $\Be(x,G_i, z_{i,j})$.
\item
we assume in order to reach a contradiction that for all $j$ we have $\neg\Be(x,G_i, z_{i,j})$. If by any chance we have that for at least one $j$ $\neg N\Be(x,G_i, z_{i,j})$, we would get by (7b) $\neg Nx(G_i)$ also, contradicting (8b). Thus we must have $ N\Be(x,G_i, z_{i,j})$ for all $j$. But then by (3) we get $x(G_i)$, contradicting (8c). Therefore (10d) must hold.
\item
So from (6) and (10d) we get that for at least one $j$ we have $Nz_{i,j}$.
\item
summarising, having assumed $x$ in (8a) we get from (8d) and (8e) that for all $i$, there exists a $j$ such that $Nz_{i,j}$ holds. This is the direction $x\rightarrow \bigwedge_{GR_0x}\bigvee_{z\in G}Nz$.
\end{enumerate}
\item
\begin{enumerate}
\item
Assume for each $i$ there exists a $j$ such that $Nz_{i,j}$ holds.
\item
Therefore from (6) we get that for all $i$ there exists a $j$ such that $\Be(x,G_i, z_{i,j})$ holds.
\item
Therefore from (4) we have for all $i$, $Nx(G_i)$ holds.
\item
Therefore from (1) we have that $x$ holds.
\item
Summarising, we have proved that $\bigwedge_{GR_0x}\bigvee_{z\in G}Nz \rightarrow x$.
\end{enumerate}
\end{enumerate}
We continue to derive $\BF^{\rm Joint}_{\rm out}$.
\begin{enumerate}\setcounter{enumi}{9}
\item From (5) we get $\bigwedge_j z_{i,j}\to \bigwedge_j \Be(x,G_i, z_{i,j})$.
\item Using (6) and (3) we get  $\bigwedge_j z_{i,j}\to x(G_i)$.
\item From (2) and (11) we get $\bigvee_i\bigwedge_j z_{i,j}\to Nx$.  This is $\BF^{\rm Joint}_{\rm out}$.

We now show $\BF^{\rm Joint}_{\rm und}$.

We need first to prove (14) below.

From (3) we get
\item $\bigwedge_i \neg x (G_i)\leftrightarrow \bigwedge_i \bigvee_j \neg N\Be(x, G_i, z_{i,j})$.

From (6) substituted in the right hand side of (13) we get
\item $\bigwedge_i \neg x(G_i)\leftrightarrow \bigwedge_i\bigvee_j \neg NNz_{i,j}$

hence $\bigwedge_i\neg x(G_i)\leftrightarrow \bigwedge_i\bigvee_j\neg z_{i,j}$.
\end{enumerate}
\noindent We now carry on to prove $\BF^{\rm Joint}_{\rm und}$.
\begin{enumerate}\setcounter{enumi}{14}
\item
Assume the antecedents of $\BF^{\rm Joint}_{\rm und}$, namely
\begin{enumerate}
\item $\bigvee_i\bigwedge_j (z_{i,j}\vee\neg Nz_{ij})$.
\item $\bigwedge_i\bigvee_j \neg z_{i,j}$.

We want to show the consequent
\item $\neg x\wedge\neg Nx$.
\end{enumerate}
From (a) there is an $i^*$ satisfying
\item $\bigwedge_j (z_{i^*, j}\vee\neg z_{i^*, j})$.

For this $i^*$ we use (b) an get $\bigvee_j\neg z_{i^*, j}$.  Let $j^*$ be a choice for this disjunction.  We thus get that
\item $\neg z_{i^*, j^*}\wedge\neg Nz_{i^*, j^*}$ holds.

Therefore by axiom (5b) we get that
\item $\neg \Be(x. G_{i^*}, z_{i^*, j^*})\wedge\neg N\Be(x, G_{i^*}, z_{i^*, j^*})$ holds.

Look again at $G_{i^*} =\{z_{i^*, 1}\comma z_{i^*, r(i^*)}) $.  We have for each $(i^*, j)$ that either $z_{i^*,j}$ holds in which case by (5a) $N\Be(x,G_{i^*}, z_{i^*, j})$ holds or $\neg z_{i^*,j}$ holds in which case by (15) we get that $\neg z_{i^*, j}\wedge\neg Nz_{i^*, j}$ holds and therefore by (5b)
\item $\neg \Be(x, G_{i^*}, z_{i^*, j})\wedge\neg N\Be(x.G_{i^*}, z_{i^, j})$ holds.

Therefore (7b) holds for $G_{i^*}$ and hence
\item $\neg x (G_{i^*})\wedge\neg N\Be (G_{i^*})$ holds.

Let us see what we got so far.  From (14) we have
\[
\bigwedge_i \neg x (G_i)
\]
From (19) we get
\[
\bigvee _i \neg Nx(G_i)
\]
holds.

Therefore the antecedent of (7a) holds and we get
\begin{enumerate}
\setcounter{enumii}{2}
\item $\neg x\wedge\neg Nx$
\end{enumerate}

\end{enumerate}

\end{remark}
\end{proof}

\section{Higher level attacks}

Higher level attacks are attacks on attacks. This concept was first introduced in 2005 in~\cite{540-8} (see also\cite{540-22} for an expanded version). The idea is very simple. Suppose we are given an abstract system $(S,R)$ where $S$ is a non-empty set of objects and $R$ is a binary relation on $S$, ( i.e. $R$ is a set of pairs $(x,y)$ where $x$ and $y$ are from $S$). In argumentation theory $S$ is a set of arguments and $R$ is the attack relation, (denoted by $x\twoheadrightarrow y$, for $(x,y)$ in $R$), but in general there are other interpretations for $S$ and $R$. In the abstract,  we can  regard the elements of $R$ also as objects and therefore we can expand the relation $R$ into a wider relation $R_1$ also containing elements from R itself. This means that in argumentation we allow also to have attacks on attacks.
So we can also write possibly
\begin{enumerate}[label=\roman*]
\item
$z\twoheadrightarrow (x\twoheadrightarrow y)$
\item
$(u\twoheadrightarrow v)\twoheadrightarrow (x\twoheadrightarrow y)$
\item
$(u\twoheadrightarrow v)\twoheadrightarrow z$
\end{enumerate}
The papers \cite{540-2,540-5,540-6,540-7,540-9} deal only with possibility (i) above, but there is no mathematical reason not to generalise to all three possibilities and even  iterate the process.

The problem is to extend whatever semantics we have for the case of $(S,R)$ to the general iterated case where we turned the elements of R into objects .

In the case of higher level attacks in argumentation, we need to give semantics for the generalised higher level networks.
This is not difficult, thanks to the nature of the Caminada labelling. Since we regard higher level attacks as legitimate objects to be attacked (following the pattern in (i)), we can give such objects labels from the set $\{\text{in},\text{out},\text{undecided}\}$ and require the Caminada conditions (conditions C1,C2, C3 of the beginning of Section~\ref{540-S3}) to hold for the higher level attacks.

Our task in this paper is to express higher level attacks in the system \BC\BN.

To show how to express higher level attacks in {\bf CN} we shall reduce higher level attacks to joint attacks by adding points to the network. First, let us define the networks we are dealing with.

\begin{definition}\label{540-6D1}{\ }
\begin{enumerate}
\item Let $S$ be a non-empty set of arguments. A higher level network based on $S$ has the form $\CA =(S, \rho_1\comma \rho_n)$ where
\[
\begin{array}{ll}
\rho_1\subseteq S\times S,&\mbox{first level attacks}\\
\vdots\\
\rho_{i+1} \subseteq S\times \rho_i,&\mbox{$i+1$ level attacks}\\
\vdots\\
\rho_n\subseteq S\times\rho_{n-1},& \mbox{nth level attacks}
\end{array}
\]
Note that $S, \rho_1\comma \rho_n$ are all pairwise disjoint
\item Let $\CA =(S, \rho_1\comma\rho_n)$ be a higher level network. We define the associated joint attacks network $\CA_{\tau}=(S_\tau, R_\tau)$ as follows:
\begin{enumerate}
\item $S_\tau=S\cup \bigcup_i \rho_i$
\item Let $x\in S_\tau$.  Then there is a unique $i$ such that $x\in \rho_i$. Let $(z_1,x)\comma (z_k, x)$ be all the attackers of $x$ in $\rho_{i+1}$.  Then let the following joint attacks be in $R_\tau$
\[
\begin{array}{l}
\{z_1,(z_1,x)\}R_\tau x\\
\vdots\\
\{z_k,(z_k,x)\}R_\tau x.
\end{array}
\]
\end{enumerate}
\item We stipulate/define that the complete extensions of $\CA$ to be the restrictions of the complete extensions $\lambda$ of $\CA_\tau$ to $S$
\end{enumerate}
\end{definition}

\begin{example}\label{54-6E2}
Consider Figure \ref{540-6F3}

\begin{figure}
\centering
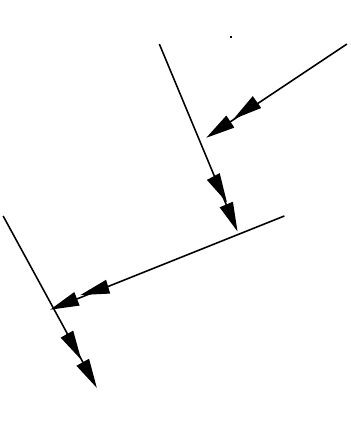
\caption{}\label{540-6F3}
\end{figure}

In Figure \ref{540-6F3} we gave names to the attacks.
\[
\begin{array}{l}
\alpha = z\tO x\\
\beta = y\tO \alpha\\
\gamma = u\tO \beta\\
\delta = w\tO\gamma
\end{array}
\]
It is natural to give the double arrows names, because in higher level attacks we treat them as objects to be attacked.

Figure \ref{540-6F3} becomes Figure \ref{540-6F4}.

\begin{figure}
\centering
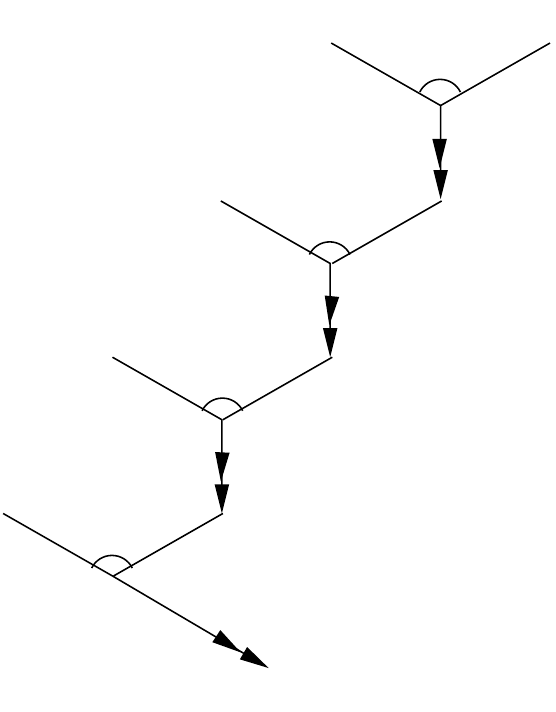
\caption{}\label{540-6F4}
\end{figure}

This illustrates how we reduce higher level attacks to joint attacks, simply by giving names to the attacking arcs.
\end{example}

\section{Disjunctive attacks}
Our starting point is the representation of the disjunctive attack of Figure \ref{540-F25b}.  We noted in Remark \ref{540-R26} that this attack is represented in {\bf CN} by the formula
\[
z\to \bigvee^n_{i=1} Nz_i.
\]
In classical logic this formula is equivalent to the formula
\[
\bigvee^n_{i=1} (z\to Nz_i).
\]
On the face of it, it looks like a disjunctive attack on a set is reduced to attacking one of the elements of the set. The next example will explain what is really going on.

\begin{example}\label{540-ES1}
Consider the networks in Figures \ref{540-FS2}, \ref{540-FS3} and \ref{540-FS4}.

\begin{figure}
\centering
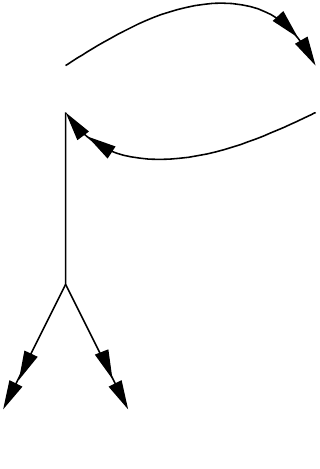
\caption{}\label{540-FS2}
\end{figure}

\begin{figure}
\centering
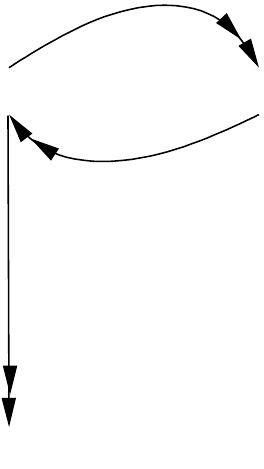
\caption{}\label{540-FS3}
\end{figure}

\begin{figure}
\centering
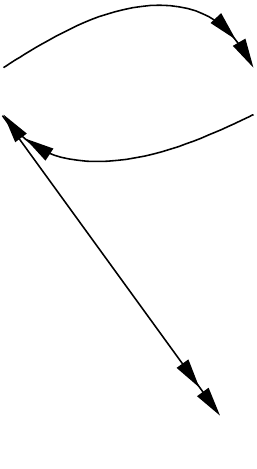
\caption{}\label{540-FS4}
\end{figure}

The representation of the attacks of Figure \ref{540-FS2} in {\bf CN} is through (1)--(3):
\begin{enumerate}
\item $a\to Nb$
\item $b\to Na$
\item $a\to Nx\vee Ny$.
\end{enumerate}
Item (3) is equivalent in classical logic to item (4):
\begin{enumerate}
\setcounter{enumi}{3}
\item $(a\to Nx)\vee (a\to Ny)$.
\end{enumerate}
The representation of the attacks of Figure \ref{540-FS3} is by items (1) and (2) as well as item (4.1)
\begin{enumerate}
\item [4.1.] $a\to Nx$
\end{enumerate}
while the representation of the attacks in Figure \ref{540-FS4} is by items (1), (2) and (4.2).
\begin{enumerate}
\item [4.2.] $a\to Ny$.
\end{enumerate}

We recall that we said that in {\bf CN} the attack relations of the given source network are brought forward (using $N$) from the meta-level in the object level. So the meaning of the equivalence in classical logic of
\[
{\rm (Equiv):}   ~~~ (a\to Nx\vee Ny)\leftrightarrow ((a\to Nx)\vee (a\to Ny))
\]
is the meta-level statement (MD):
\begin{enumerate}
\item [(MD):]  $E$ is a complete extension of the network of Figure \ref{540-FS2} iff $E$ is a complete extension of at least one of the networks of Figures \ref{540-FS3} or \ref{540-FS4}.
\end{enumerate}

We make two critical comments.
\begin{itemlist}{CCCCC}
\item [(CC1)]  (MD) is a meta-level property. We can ask the following question (Q) for example:

(Q) Given traditional networks $(S, R_j), i=1\comma k$, then under what conditions do there exist a disjunctive network $(S, \BBR)$ such that (MD) holds?

(CC2) If our logic {\bf CN} were not classical logic, but say intuitionistic logic (as will be investigated in Part 2 of this paper \cite{540-24}), then the equivalence (Equiv) will not hold and the results would be different.
\end{itemlist}
\end{example}

We are now ready for a formal definition of disjunctive attacks. We use Caminada--Gabbay labelling.
\begin{definition}\label{543-D12}
{\ }
\begin{enumerate}
\item A disjunctive argumentation network has the form $(S, \rho)$, where $S$ is a non-empty set and $\rho \subseteq S\times (2^S-\varnothing)$.    See Figure \ref{543-FDS1}.

\begin{figure}
\centering
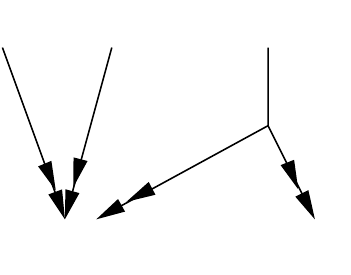
\caption{}\label{543-FDS1}
\end{figure}
\item Let $x\in S$. Let $y_1\comma y_k$ be all the elements of $S$ which attack $x$ directly.  Let $z_1\comma z_m$ all be elements of $S$ which disjunctively attack sets which contain $x$, namely $z_i\rho\{x, u^i_1\comma u^i_{r(i)}\}$.  We say $z_i$ indirectly attacks $x$.

\item Let $\lambda: S\mapsto\{\mbox{in, out, und}\}$.  We say that $\lambda$ is a legitimate Caminada--Gabbay labelling for $(S, \rho)$ iff the following conditions hold:

\item [(D1)] $\lambda (x)=$ in if $\lambda(y_i)=$ out, $i=1\comma k$ and for all $1\leq j \leq m$, if $\lambda (z_j)=$ in then for some $k, \lambda (u^j_k)=$ out.
\item [(D2)] $\lambda (x)=$ out, if for some $1\leq i \leq k, \lambda (y_i)=$ in, or for some $j, (\lambda (z_j)=$ in and for all $1 \leq k \leq r(j)$ we have $\lambda (u^j_k)\neq$ out.
 \item [(D3)] $\lambda (x)=$ und iff all attacks on $x$ direct or indirect are either out or undecided, and where at least one attack is undecided, where the following define the meaning of the terms we just  used.
 \begin{itemize}
 \item A direct attack on $x$ by $y$ is out if $\lambda (y)=$ out. It is undecided if $\lambda (y)=$ und.
 \item An indirect attack of $z$ on $x$ using the disjunctive attacked set $\{x, u_1\comma u_r\}$ is out if either $\lambda(z)$ is out or for some $j, \lambda (u_j)$ is out.
The attack is undecided if $ \lambda(z)$
 and $\lambda (u_j)$ are all different from out, with at least one of them is und.
 \end{itemize}
\end{enumerate}
\end{definition}

\begin{remark}\label{543-RRH1}
Using higher level attacks and $\top$, we can implement disjunctive attacks using conjunctive attacks. See Figure \ref{543-FRR2}.  This figure implements the disjunctive attack $a \rho \{b,c\}$ .   We have that $\{b,c\}\tO \top)$ and $\top\tO(\{b,c\}\tO \top)$ and $a\tO(\top\to (\{b,c\}\tO\top))$

\begin{figure}
\centering
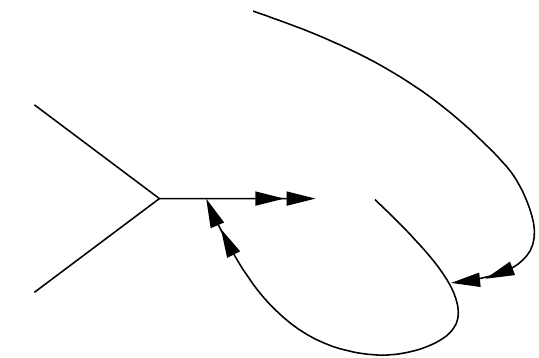
\caption{}\label{543-FRR2}
\end{figure}

\end{remark}

\section{Comparison with literature}\label{comparison}
\subsection{Comparison with the paper of  S.Villata, G. Boella and L .van der Torre: Attack Semantics for Abstract Argumentation}
We compare with paper \cite{540-25}.  The paper is different from ours, but since it has a similar name, we should address it.

Our paper translates argumentation into classical logic with a symbol ``$Na$'', meaning ``$a$ is under attack''.  Paper \cite{540-25} regards the attack arrows as objects and discusses semantics for them. (In our paper ``$a\tO b$'' is a wff ``$a\to Nb$'').  Our set up is completely different. Furthermore, the idea of regarding the attack arrow as an object and labelling it already appears in \cite[Figure 13]{540-2} and in \cite{540-8}.  However, we do  want to make a point in this paper using the machinery of \cite{540-25} and so let us quote some passages from them and then make our point.

\begin{quote}
{\bf Begin quote 1}:

{\bf Example 2} Consider $AF = \langle A,\to\rangle$ with $A =\{a; b; c; d; e\}$
and $\{a\to b, b\to a, b\to c,c\to d, d\to e, e\to c\}$ visualized
in Figure 2. The complete extensions are $\CE_{\rm compl} (AF;A) =
\{\emptyset, \{a\}, \{b;d\}\}$. $\emptyset$  is the unique grounded extension, $\{a\}$ and
$\{b; d\}$ are the preferred extensions, and $\{b; d\}$ is the stable
extension. In the complete extension $\emptyset$ the arguments are rejected,
because they are not defended. In the extension $\{b; d\}$
the other arguments are rejected, because they are attacked
by an accepted counterargument.

\begin{centering}
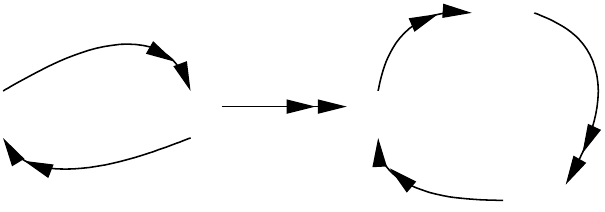

{\bf Figure 2}: The argumentation framework of Example 2.
\end{centering}

\medskip

{\bf Begin quote 2}:

\newcommand{\hookr}{\hookrightarrow}
{\bf Example 3 (Continued from Example 2 in quote 1 above)}:
The grounded extension $\emptyset$ has only successful attacks, the preferred extension $\{a\}$ has $\{a \hookrightarrow b, c\hookr d, d\hookr e, e\hookr c\}$ as successful attacks, and the stable extension $\{b,d\}$ has $\{b\hookr a, b\hookr c, d\hookr e\}$.

\medskip
{\bf Begin quote 3}:

We define attack labeling analogous to argument labeling.
An attack is {\em in} when its attacking argument is {\em in}, an attack is
{\em undecided} when its attacking argument is {\em undecided}, and an
argument is {\em out} when its attacking argument is {\em out}. An attack
is successful when it is {\em in} or {\em undecided}, whereas an argument
is accepted when it is {\em in}. For example, if an argument is rejected,
but at least one of its attacks is successful, then the
argument is undecided.

{\bf Example 5 (Continued from Example 3 in quote 2 above).} The grounded
extension $\emptyset$ has only undecided attacks, the preferred extension
$\{a\}$ has in attack $a\hookr b$  and undecided attacks
$c\hookr d, d\hookr e, e\hookr c$,  and the stable extension $\{b,d\}$has in
attacks $b\hookr a, b\hookr c, d\hookr e$  and no undecided attacks.

{\bf End quotes.}
\end{quote}

Let us implement the attack network of Example 2 in quote 1 above and of its Figure 2 in our {\bf CN} logic using the idea of \cite[Figure 13]{540-2}.  Note that this idea is NOT what \cite{540-25} does but we start with our own approach for comparison.  We add for each arrow $x\tO y$ a new node $z_{x,y}$ and translate the attack $x\tO y$ into the conjunctive attack $(x, z_{x,y})\tO y$. We thus get Figure \ref{540-FQ6}.  This figure has conjunctive attacks and can be dealt with within {\bf CN}.

\begin{figure}
\centering
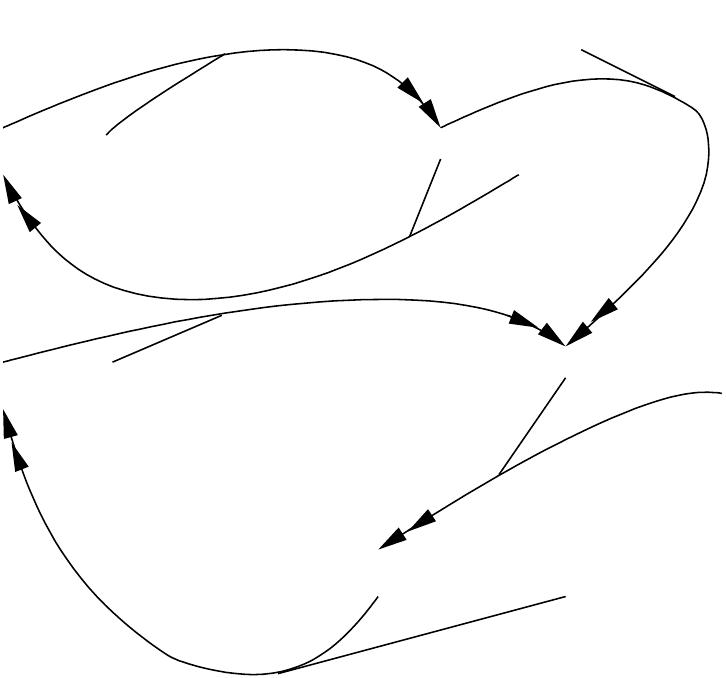
\caption{}\label{540-FQ6}
\end{figure}

Note that all the attack arrows, $z_{a,b}, z_{b,a}, z_{b,c}, z_{c,d}, z_{d,e}$ and $z_{e,c}$ are all in, because they are not attacked.  This is as it should be according to the Dung traditional approach.

However, as we said, the above is not the attack semantics of \cite{540-25}. According to \cite{540-25} an attack $x\tO y$ is out if $x$ is out. So the attacks in Figure \ref{540-FQ6} are not always in, but depend on the extensions chosen, as discussed in \cite{540-25} and quoted above in Example 5 of quote 3.  So we need a new figure to implement the attack semantics of \cite{540-25}.

To achieve that we need to add that $\neg x$ attacks $z_{x,y}$.  This is easy to write in {\bf CN}:  $\neg x \to Nz_{x,y}$.  No new figure is really needed. However we can, if we insist, use Figure \ref{540-FQ7}.  It is a quite complicated higher level figure.
\begin{figure}
\centering
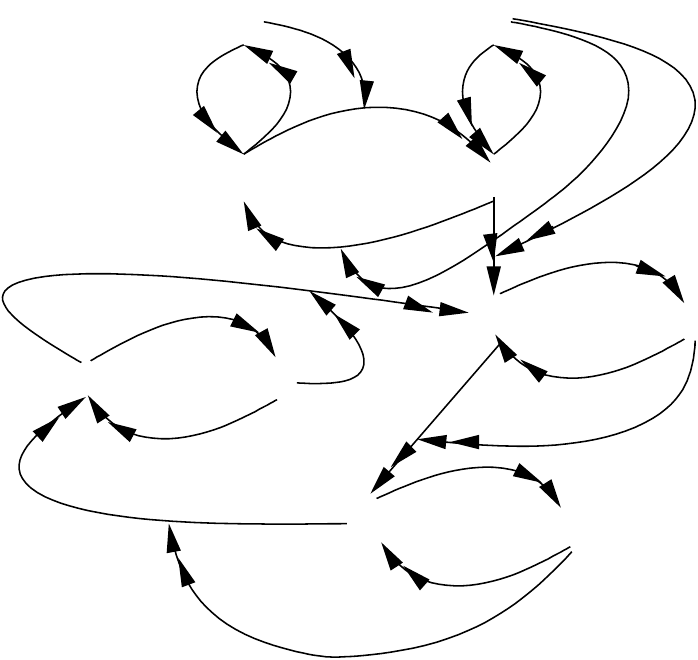
\caption{}\label{540-FQ7}
\end{figure}

If we do not want to use higher level attacks, we can use Figure \ref{540-FQ8}.  The figures tend to be complicated because we need to add the nodes $\neg x$ on $x\tO y$ (or on $z_{x,y}$).  In {\bf CN} we just write $\neg x \to Nz_{x,y}$, which is much simpler and {\em does not add any nodes}!

\begin{figure}
\centering
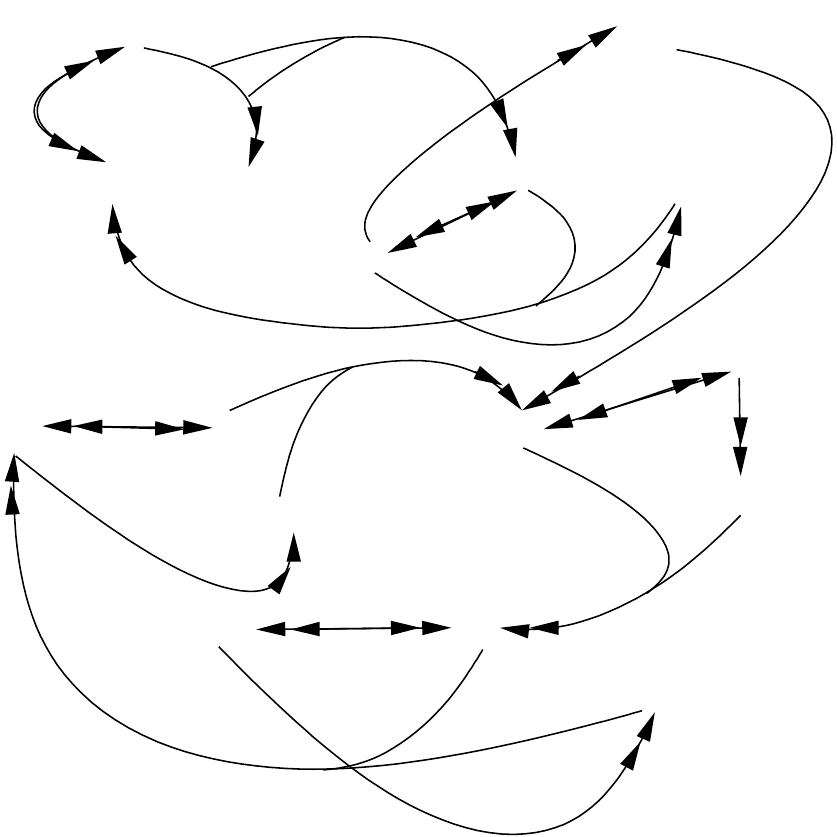
\caption{}\label{540-FQ8}
\end{figure}

Independent of the problem of representing the attack semantics of \cite{540-25} and the idea that {\em  if $x$ is out then $x\tO y$ is also out}, is compatible with another idea of ours, from \cite[Section 4.3]{540-3}.  In \cite{540-3}, we regard $x\tO y$ as a channel, carrying an attack signal, sent by $x$ towards $y$.  So if $x$ is out, no signal is sent but the channel is alive. We can say that according to \cite{540-25}, if $x$ is out, the channel itself is out.

Let us now say more about our approach of \cite[Section 4.3]{540-3}.  We want to look at the attack $a\tO b$ as a signal from $a$ to $b$.  As such it can be considered as an object (a wave front). It therefore can split or it can interfere/join with another signal. This idea was presented in \cite[Section 4.3]{540-3} and will be fully developed in \cite{540-23}.  Figure \ref{540-FS5} illustrates this idea.  Compare with Figure \ref{540-FS2}.

In Figure \ref{540-FS5} the attack emanates from $a$ and disjunctively splits into two attacks one that attacks both $x_1$ and $x_2$ and one that attacks both $y_1$ and $y_2$.  Success, should $a$ be ``in'', is that either both $x_1$ and $x_2$ are out or  both $y_1$ and $y_2$ are out.

\begin{figure}
\centering
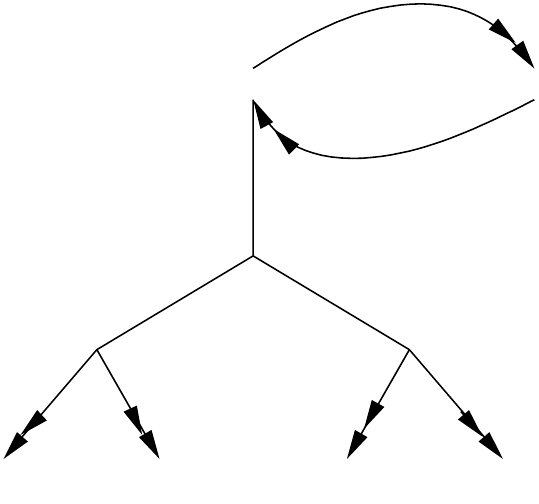
\caption{}\label{540-FS5}
\end{figure}

The attacks of Figure \ref{540-FS5} can be described in {\bf CN} by the clauses:
\begin{enumerate}
\item $a\to Nb$
\item $b\to Na$
\item $a\to N(x_1\vee x_2)\bigvee N(y_1\vee y_2)$

\end{enumerate}

\subsection{Comparison with Gabbay's paper: Modal  Provability  Foundations for Argumentation Networks}
In 2009, Gabbay \cite{540-31} asked a simple question:
\begin{itemlist}{(2009Q)}
\item [(Q2009)]  Given an argumentation network $\CA =(S, R)$, what is its logical content $\Gamma_{\CA}$ in terms of some reasonable logic \BL?
\end{itemlist}
Gabbay answered this question using a version of a modal logic of provability called {\bf LN1}. The similarity with our paper is clear.  We associate a modal logic {\bf CNN} and a theory $\Gamma_{\CA}$ with $\CA$. The view of \cite{540-31} is that $\Gamma_{\CA}$ is the theory satisfying more or less (for accuracy see \cite{540-31}) the equivalence below:
\[
\Gamma_{\CA}\leftrightarrow \square (\bigwedge_x (x\leftrightarrow \bigwedge_i \lozenge (\Gamma_{\CA} \wedge\neg y_i))
\]
where $y_i$ are all the attackers of $x$ in $\CA$.

In words: the logical content $\Gamma$ says that  $x$ is in iff it is possible, from the point of view of  $\Gamma$,  to have all the attackers of $x$ out/false. This is a fixed point equation  for $\Gamma$.

The above is the idea of \cite{540-31}. Once \cite{540-31} develops the semantics for the logic, we end up with a three point linear modal logic $(t_1, t_2, t_3)$, with $t_1$ the actual world and $t_j$ accessible to $t_{j-1}, j =3, 2$.  In other words we have
$t_1 < t_2 < t_3$, where ``$<$" is the modal accessibility relation.

The labelling obtained from such a model is the following, where the triple $(v_1 , v_2 , v_3)$ indicates the respective values in worlds $(t_1, t_2, t_3)$:
\[\begin{array}{l}
\mbox{in }= (\top, \top, \top)\\
\mbox{out } = (\bot, \bot, \bot)\\
\mbox{und }= (\top, \bot, \top)\\
\mbox{und* } =(\top, \bot, \bot)
\end{array}
\]

As you can see, \cite{540-31} has a completely different point of view, and furthermore, the point of view of \cite{540-31} can get two types of undecided!  Comparing with our current paper we have a classical point of view, and have a two worlds linear modal logic, where the assignment is restricted like in intuitionistic logic, giving rise to the following labelling with only one type of undecided:
\[\begin{array}{l}
\mbox{in } =(\top, \top)\\
\mbox{out } =(\bot, \bot)\\
\mbox{und }= (\bot, \top)
\end{array}\]

To obtain a better comparison, let us organise the worlds $t_1, t_2, t_3$ in different accessibility ordering, say $<^*$:
\[
t_2 <^* t_3 <^* t_1.
\]
The four possible values can now be seen as in Figure \ref{540-FMay17-1}

\begin{figure}
\centering
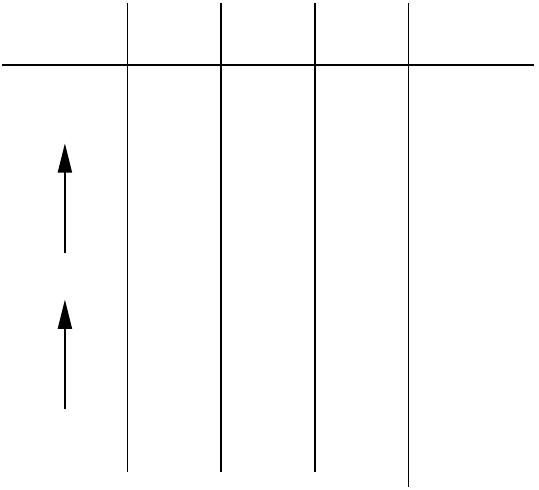
\caption{}\label{540-FMay17-1}
\end{figure}

For the ordering $t_2 <^* t_3 <^* t_1$, the assignments are clearly intuitionistic. True continues to be true in this ordering.

Consider now the loop of Figure \ref{540-FMay17-2}.

\begin{figure}
\centering
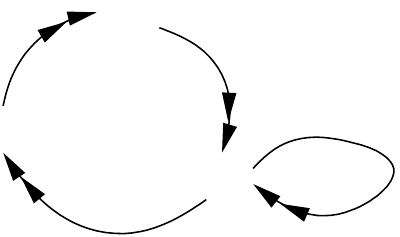
\caption{}\label{540-FMay17-2}
\end{figure}

In this loop $b$ attacks itself , so its degree of undecided  must be  possibly different from  the degrees of the other points in the loop and must  be in the middle,  equally away from the value ``in" and from the value ``out". Indeed paper \cite{540-31} can distinguish between the node $b$ being undecided and the other nodes  . There is a problem, however, in this case.  There are only two values of undecided , und and und*, and so there is no middle, so what value do we give $b$?  Paper \cite{540-31} calculates the values according to its fixed point semantics rationale and  gives $b$  the value und, while the value of $a$ is und* and the value of $c$ is also und.  Thus according to \cite{540-31}, we get:
\[
a =(\bot,\bot,\top), b=c=(\bot,\top,\top).
\]

We, however, would like to give $b$ the other value of undecided, namely $b = (\bot, \bot, \top)$, and so we have to give  $c= (\bot, \top, \top)$ and have $a=b$.We can offer the following rationale for these values:
\begin{enumerate}
\item $x$ being undecided means that $x$ can go either way, to be in or be out, but is neither, i.e.\ $x\neq (\top,\top,\top), x\neq (\bot,\bot,\bot)$.
\item Different vectors for undecided relate to how likely $x$ is to be in or out.  So when comparing the sequences $(\bot, \top,\top)$ with $(\bot,\bot,\top)$, the less ``$\top$'' we have in the sequence, the nearer  ``out'' $=(\bot, \bot, \bot)$ is the value of the sequence.
\item In Figure \ref{540-FMay17-2}, we need to choose the value for $b$. Paper \cite{540-31} adopted a generous credulous  view for letting an argument $x$ to be in, namely that $x$ is in if it is possible for all its attackers $y$  to be out. Our view in this paper is more skeptical and strict, it  is that $x$ is in if all attackers y are strongly out ($Ny$ holds).  So if  paper \cite{540-31} calculated   that $b= (\bot, \top ,\top)$, and our view is different , then we might consider the other undecided value, namely  $b= (\bot,\bot,\top)$. This option  agrees  with the equational approach, (see Remark \ref{540-RMay17-3} below)  and  so we choose (contrary to \cite{540-31}), to take the value of $b$ to  be $(\bot,\bot,\top)$.
\item $c$ is attacked by $b$, which is  nearer $b=$ out, so $c$ must be still undecided, but nearer the value in.  So we expect $c$ to have the value $(\bot,\top,\top)$
\item From (4) we expect $a$ to be still undecided but nearer the value out and so our only option is to have $a=b=(\bot,\bot,\top)$.
\end{enumerate}

The comparison and discussion of the rationale above suggests that we can look at logics like {\bf CNN}$k$, where $k2$ defined by chains of worlds of the form
\[
w_1 <^* w_2 <^* \comma <^* w_k.
\]
The undecided values can be of the form
\[
\mbox{und}^*_j =(\bot\comma \bot, \top\comma \top)
\]
where $1 < j \leq k$ and the value is $\bot$ until position $k-1$ and $\top$ afterwards.

We reserve the study of {\bf CNN}$k$ for Part 2 of this paper.

\begin{remark}\label{540-RMay17-3}
Note that the equational approach of \cite{540-39} also supports our analysis and choice of value for $b$.  In this approach we solve the following equations of the nework of Figure \ref{540-FMay17-2}.
\begin{enumerate}
\item $b =(1-b)(1-a)$
\item $c=(1-b)$
\item $a=(1-c)$
\end{enumerate}
From (2) and (3) we get
\begin{enumerate}
\setcounter{enumi}{3}
\item $a=b$
\end{enumerate}
Substituting in (1) we get
\begin{enumerate}
\setcounter{enumi}{4}
\item $a=(1-a)^2$
\end{enumerate}
which solves to $a=b=0.38$ and $c=0.62$.

This is in agreement with our analysis provided we understand  undecided as numerical value strictly between 0 and 1 and nearer out  to mean a numerical value nearer to 0.

We can also view the equational approach as a {\bf CNN}$\infty$ approach, where the possible worlds intuitionistic model is taken to be $[0,1]$, with 0 the actual world and the value $x=k$ for $x\in S$ and $k\in [0,1]$ means
\begin{itemize}
\item $y\vDash x$ iff $x \leq k$
\end{itemize}

Again, this connection will be investigated in Part 3 of this paper.
\end{remark}

\subsection{Comparison with David Pearce logic}
We compare our logic {\bf CNN} with Pearce's logic of \cite{540-25}.  Pearce introduced his logic in 1995 in the context of studying Answer Set Programming (ASP). Pearce 's logic is very similar to our logic and so we need to offer a comparison. Furthermore, Pearce did apply his logic to argumentation.\footnote{
The following is a private communication from David Pearce:\\
``Equilibrium logic was designed to capture stable reasoning and, not surprisingly, you can apply it to capture stable extensions in argumentation theory. Together with other colleagues we were able to show how to capture stable extensions in assumption-based argumentation. However, including strong negation would be a new topic."}
  In Pearce's own words \cite{540-30}, equilibrium logic $N_5$ provides
\begin{enumerate}
\item a general methodology for building nonmonotonic logics;
\item a logical and mathematical foundation for ASP-type systems, enabling one to prove useful metatheroetic properties;
\item further means of comparing ASP with other approaches to nonmonotonic reasoning.
\end{enumerate}

The semantics for Pearce's logic uses the same two possible worlds linear model which we use, namely an actual world, $t$, a possible world $s$, with the restriction on atoms $q$ that if $t\vDash q$ then $s\vDash q$. The difference between us and Pearce is in what we actually do with this semantics.

In order to be able to compare the two logics, we need to agree on notation.  Our approach and notation so far can be summarised as follows:

\paragraph{Comments Group 1. }  The model has two worlds, $t$ and $s$.  $t$ is the acutral world and $s$ is possible relative to $t$. For atomic $q$ we have $t\vDash q$ implies $s\vDash q$.

\begin{enumerate}
\item Our view is that this model is a modal semantics model with two possible worlds linearly ordered whose atomic assignments satisfy a restriction.
\item Pearce's view is that this model is an intuitionistic Kripke model with two worlds (G\"{o}del's logic with two worlds).  This semantics is well known and can be axiomatised. See \cite{540-30}.
\end{enumerate}

\paragraph{Comments Group 2.} We use the connectives $\neg$ (classical negation), $\to$ (classical implication) and modality $N$.  We also use $\wedge$ and $\vee$.

The semantical conditions are
\begin{enumerate}
\item $t\vDash \neg A$ iff $t\not\vDash A$
\item $t\vDash A\to B$ iff ($t\vDash A$ implies $t\vDash B$)
\item Same as above for $s$
\item $t\vDash NA$ iff $s\vDash \neg A$
\item $s\vDash NA$ iff $t\vDash \neg A$
\item $x\vDash A\vee B$ iff $x\vDash A$ or $x\vDash B, x\in \{t,s\}$
\item $x\vDash A\wedge B$ iff $x\vDash A$ and $x\vDash B, x\in \{t, s\}$
\end{enumerate}

\paragraph{Comments Group 3.} Pearce uses intiutionistic connectives and Nelson strong negation. He calls his system $N_5$, the logic of here and there.  To describe Pearce's system $N_5$, let us use the notation below to distinguish the intuitionistic connectives from the classical ones and to distinguish $N_5$ from {\bf CNN}.  Let us use the following:
\begin{enumerate}
\item ``$\Nneg$'' for intuitionistic negation.  Its semantic satisfaction condition is
\begin{itemize}
\item $t\vDash \Nneg A$ iff $t\not\vDash A$ and $s\not\vDash A$
\item $s\vDash \Nneg A$ iff $s\not\vDash A$
\end{itemize}
(Note that for atoms we have $t\vDash A$ implies $s\vDash A$!)
\item ``$\To$'' for intuitionistic implication with the semantic conditions
\begin{itemize}
\item $t\vDash A\To B$ iff ($t\vDash A$ implies $s\vDash B$) and ($s\vDash A$ implies $s\vDash B$)
\item $s\vDash A\To B$ iff ($s\vDash A$ implies $s\vDash B$).
\end{itemize}
\item Use $\wedge$ and $\vee$ with the same semantic conditions as Comment 2.6 and 2.7 respectively.
\item Note that for wffs with $\wedge, \vee, \Nneg$ and $\To$ we have
\begin{itemize}
\item $\vDash A$ implies $s\vDash A$
\end{itemize}
\item Pearce adds strong negation ``$\sim$'' to the system $N_5$ in a similar way to the way we added \BN\ in Section 1.  With each atom $q$, Pearce adds another atom $\sim q$ and adds the axiom $\sim q \To \Nneg q$.

Thus $N_5$ is the intuitionistic theory $\{\sim q \to \Nneg q\}$ based on the atoms $\{q, \sim q|q \mbox{ atomic}\}$.  Note that for us in {\bf CN}, we added $Nq$ with the axiom $Nq\to \neg q$ and we have $t\vDash Nq$ iff $s\vDash \neg q$. This implies that if $\neg q \wedge Nq$ holds at $t$ then $\Nneg q$ holds at $t$.

\end{enumerate}
Just like we did for $N$ in allowing $N$ to apply to any wff and pushing $N$ through to the atoms, Pearce does a similar thing for $\sim$.  $\sim$ satisfies the following rules both at $t$ and at $s$:
\begin{itemize}
\item $\sim (A\wedge B)$ iff $\sim A\vee\sim B$
\item $\sim (A\vee B)$ iff $\sim A\wedge\sim B$
\item $\sim A\To B)$ iff $A\wedge\sim B$
\item $\sim (\Nneg A)$ iff $A$
\item $\sim\sim A$ iff $A$.
\end{itemize}
These rules  for ``$\sim$'' allow us to push ``$\sim$'' right in front of the atoms.

\paragraph{Comments Group 4.} Comparison so far is as follows:
\begin{enumerate}
\item Both ``$\sim$'' and ``$N$'' associate with any atom $q$ the atom $\sim q$ or $Nq$ respectively.  Both $\sim $ and $N$ can be pushed down to the atoms syntactically.  $N$, however, directly changes the world it is evaluated in. It is a true explicit modality. In comparison, $\sim$ remains in the same world. It basically adds a strong atom $\sim q$ to any atom $q$, but $\sim q$ is evaluated semantically in the same world.   It can push the
evaluation to another world only through the axioms  $\sim \Nneg q \To q$ and $\sim q\To \Nneg q$ and this is a roundabout way of doing it.

Our logic {\bf CNN} also added with each atom $q$ another atom $Nq$, but $Nq$ connects with an accessible world where $q$ must be false.
\item Pearce moves from world $t$ to world $s$ using $\Nneg$ and $\To$.  We move from $t$ to $s$ using $N$.  We can also move from $s$ to $t$ using $N$.  Pearce has no direct connective which can move the evaluation from world $s$ to world $t$.
\end{enumerate}

\paragraph{Comments Group 5.}  To further compare {\bf CNN} and $N_5$, let us define the intuitionisitc connectives in our logic {\bf CNN}.
\begin{enumerate}
\item We let
\[\begin{array}{l}
\Nneg X =\mbox{def.} \neg X\wedge NX\\
 (X\To Y)=\mbox{def } (X\to Y)\wedge N(X\wedge\neg Y \wedge Ny).
 \end{array}
 \]

\item Let $A$ and $B$ be two wffs built up from the atoms and the above connectives $\Nneg$ and $\To$ and $\wedge$ and $\vee$, as defined in Comment 5.1 above.  In other words, we have a sublanguage {\bf INN} of {\bf CNN} with wffs defined as follows:
\begin{itemize}
\item atoms $q$ are in {\bf INN}
\item If $A$ and $B$ are in {\bf INN} then so are $A\wedge B$ and $A\vee B$
\item if $A$ and $B$ are in {\bf INN}, so are \\
$(\neg A\wedge NA)$, this is $\Nneg A$\\ and\\ $(A\to B)\wedge N(A\wedge\neg B\wedge NB)$, this is $A\To B$.
\end{itemize}
\item We now check and see that the truth conditions of intuitionistic semantics hold for these translations (i.e.\ the condition in Comment 3.4 holds).
\item [(3.1)] For atomic $q$ the condition $t\vDash q$ implies $s\vDash q$ holds.
\item [(3.2)] Assume by induction that the condition in Comment 3.4  ($t\vDash X$ implies $s\vDash X$) holds for $X=A$ and $X=B$. We show that it holds for $A\wedge B$ and $A\vee B$. This is immediate.
\item [(3.3)] Assume that the condition in Comment 3.4 holds for $A$ and we show that the condition in Comment 3.4 holds for $\Nneg A$.  We have
\begin{itemize}
\item  $t\vDash \Nneg A$ iff (by definition) $t\vDash \neg A\wedge NA$ iff $t\vDash \neg A$ and $s\vDash \neg A$ iff $t\not\vDash A$ and $s\not\vDash A$. \\ We also have\\ $s\not\vDash A$ iff $s\not\vDash A$ and $
t\not\vDash A$\\ (because our inductive hypothesis was that Comment 3.4 holds, i.e.\ $t\vDash A$ implies $s\vDash A$)\\ iff ${s}\vDash \neg A$ and $s\vDash NA$ iff $s\vDash \neg A\wedge NA$ iff (by definition) $s\vDash \Nneg A$.
\item We show now that if $t\vDash \Nneg A$ then $s\vDash \Nneg A$.  We have: $t\vDash \Nneg A$ iff $t\vDash \neg A\wedge NA$. This implies $s\not\vDash A$ and we have seen before that this is equivalent to $s\Vdash \Nneg A$.
\end{itemize}
\item [(3.4)] We now check the case of $A\To B$. Assume by induction that $t \vDash X$ implies $s\vDash X$ holds for $X=A$ and for $X=B$.  We show that the semantic condition of Comment 3.2 holds for $A\To B$ and that also if $t\vDash A\To B$ holds then $s\vDash A\To B$ also holds.  We have
\begin{itemize}
\item $t \vDash A\To B$ iff (by definition) $t\vDash  (A\to B)\wedge N(A\wedge \neg B\wedge NB)$ iff $t\vDash A\to B$ and $s\vDash (\neg (A\wedge\neg B)\vee\neg NB$ iff $(t\vDash A\to B$ and $s\vDash A\to B$) or $(t\vDash A\to B\mbox{ and } s\vDash \neg NB)$.  \\
But $s\vDash \neg NB$ iff $t\vDash B$. So we have that:   $t\vDash A\to B$ iff $t\vDash (A\to B)\vee B$.

Therefore we can continue

$t\vDash A\To B$ iff $(t\vDash A\to B$ and $s\vDash A\to B) \vee (t\vDash A\to B$ and $t\vDash B)$. But $t\vDash B$ implies $s\vDash B$ and we continue

iff $t\vDash (A\to B)\wedge s\vDash A\to B)$.  This is the correct condition for the case of $t$.

We check $s$.  $s\vDash A\To B$ iff by definition $s\vDash (A\to B)$ and $s\vDash N(A\wedge\neg B)$ iff $s\vDash (A\to B)$ and $t\vDash A\to B$.
\end{itemize}
Given the above, we now show that
\begin{itemize}
\item $s\vDash A\To B$ iff $s\vDash A$ then $s\vDash B$
\end{itemize}
If $s\vDash A\To B$ then by the above $s\vDash A\to B$ and so if $s\vDash A$ then $s\vDash B$.

If $s\vDash A$ implies $s\vDash B$ then $s\vDash A\to B$. So the only option for $s\not\vDash A\To B$ is that $s\not\vDash N(A\wedge\neg B\wedge NB)$ and so $t\vDash A\wedge\neg B\wedge NB$. But then this implies by the condition of Comment 3.4 for $A$ that $s\vDash A\wedge\neg B$, contradicting what we have just assumed that $s\vDash A\to B$.
\end{enumerate}
There remain to show $A\To B$ satisfies the condition of Comment 3.4.  This is obvious, however, from the semantic condition of satisfaction for ``$\To$''.
\begin{enumerate}
\item [(3.5)] We summarise the results of the current  item (Comment 5.3).
We have just shown that {\bf CNN}  contains the intuitionistic part of $N_5$ as the fragment {\bf INN} (we defined the intuitionistic negation and intuitionistic implication in our system and shown they behave correctly). The reader can verify that the law of excluded middle for this embedded intuitionistic fragment does not hold, so the fragment is not trivial, ( i.e.\ $q\vee  (\neg q \wedge  Nq)$ is not a theorem of {\bf CNN}, take $q= (\bot, \top)$).

We now ask, can we also define ``$\sim$" as well using $N$?

Let us check.
We need a formula of {\bf CNN}, say $\alpha(q)$, which will act as the connective $\sim$, namely
\[
\sim q = \mbox{ (by definition) }\alpha(q)
\]
This formula must satisfy
\begin{enumerate}
\item  For atomic $q,  \alpha(q) \To  \Nneg q$
\item  For any formula $A$ of {\bf INN}, $ \alpha( \Nneg A) \Leftrightarrow A$
\end{enumerate}
If the value of $q$ is taken to be  $( \bot, \top)$ or $(\top, \top) $, we get because of (a). and the fact that in this case  the value of $\Nneg q$  is $(\bot, \bot)$, that the value of $\alpha(q)$ must be $(\bot, \bot)$.
To summarise we must have  (c) and (d).  to hold
\begin{enumerate}
\setcounter{enumii}{2}
\item  $\alpha(\bot, \top) = ( \bot, \bot)$
\item $\alpha(\top , \top) = ( \bot, \bot)$
\end{enumerate}
but if we have this , how can (b) hold?
For both cases  where value of $q$ is taken to be  $( \bot, \top)$ or $(\top, \top) $, the value of  $\Nneg q$ is the same and so how can we have
\[
\alpha( \Nneg q) \Leftrightarrow q  ?
\]
The only way out is to say that rule (b) is a syntactical reduction, making $\sim$ not a real connective. O.K. this is acceptable but different from our logic where $N$ is a connective.
So (b) holds syntactically, this means that (a) has to hold by a restriction on the assignment of the formal syntactical atom ``$\sim q$".
The only way to it is to take ``$\sim q$" to be the same as  ``$\Nneg q$".
This is fine except that when we iterate to have
\[
\Nneg  \Nneg q  \To q.
\]
The first ``$\Nneg $" is not a connective but a syntactical rewrite operator and only the second $\Nneg$ is a connective  and the implication is valid. This is not a problem  because we will write it syntactically as
$\sim \Nneg q \To q$
and the ``$\sim$" will rewrite the following ``$\Nneg$" out , leaving only $q\To q$.

\end{enumerate}

\paragraph{Comments Group 6.} We now start with $N_5$ as our given intuitionistic system with $\{\Nneg, \To, \wedge, \vee, \sim\}$ and the semantics as described in \cite{540-30} and  in (Comment Group 3) above and check whether we can define {\bf CN} using the connectives of $N_5$.  We need to define $N$ and the classical connectives $\{\neg, \to, \wedge, \vee\}$.  We show that this is not possible. Consider the formula of {\bf CN}.
\begin{itemize}
\item $\neg (q\wedge Nq), q$ atomic.
\end{itemize}
This formula is always true at world $t$ but can be false at world $s$ (for $q=\top$ at $s$ and at $t$).

No intuitionisitc wff, even with ``$\sim$'' can be true at $t$ and false at $s$. Therefore $N$ and $\neg$ are not definable in Pearce's logic.

It may still be possible to translate argumentation networks into $N_5$ using its strong negation. We leave this study to part 2 of our paper.

\subsection{Comparison with four meta-level approach papers \cite{540-32,540-35,540-36,540-37} and \cite{540-45}}

To compare with papers of Dvorak {\em et al.}, Doutre {\em et al.}, and Grossi, we need some preliminary methodological discussion about interpretations and translations.  All the above papers translate argumentation networks into other systems and so to explain what they are doing and compare with our paper we need a framework  of reference, a matrix schema for comparison.

Consider the schematic Figure \ref{540-FMay19-1}.  This figure is a schema for possible embeddings, containing four systems.  We name the system and give examples to help the reader visualise the schema. We then discuss options for interpretations in general and then explain how the papers we are discussing fall into the schema of Figure \ref{540-FMay19-1}.  Once we understand how these papers fit into the schema, our comparison will be concluded.

\begin{figure}
\centering
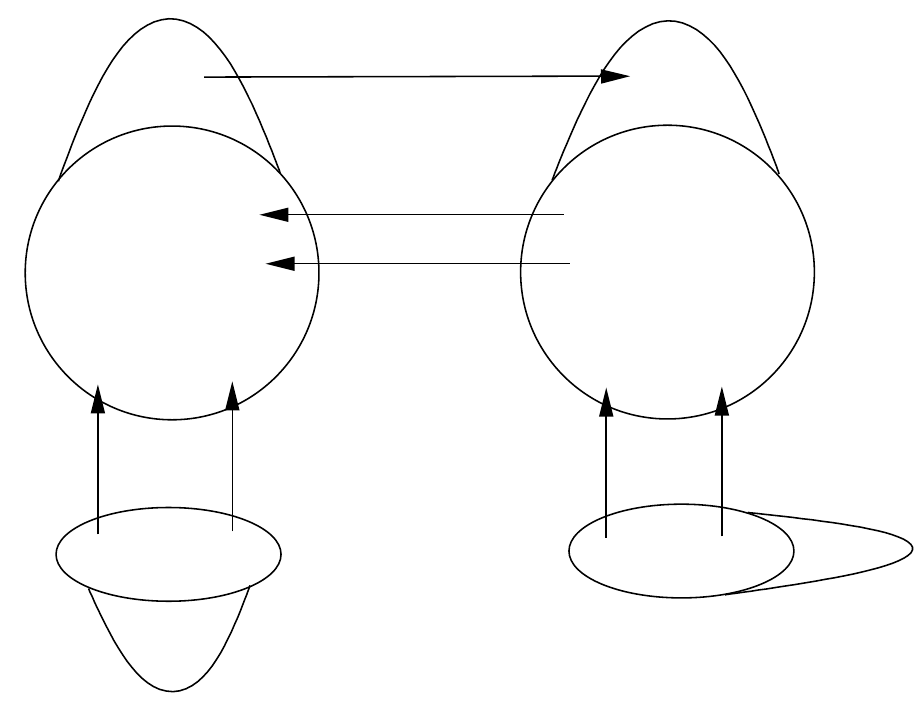
\caption{}\label{540-FMay19-1}
\end{figure}

In Figure \ref{540-FMay19-1}, the top left circle is a system \BT. Think of \BT\ as classical logic and think of $x, e, y'$ as formulas of classical logic. Think of $\BT^+$ as an extension of  \BT, say if \BT\ is classical logic then $\BT^+$ extends  \BT\ with monadic second order quantification.

The system $\BS_1$ can be interpreted into \BT.  Say $x$ is mapped to $x'$ (the arrows indicate the mapping) and $y$ is mapped to $y'$.  Think of $\BS^+$ as an extension of \BS.  Say $\BS =$ intuitionistic propositional logic and $\BS^+$ is some extension of it.  The element $e$ in \BT\ is not a result of the mapping of \BS\ into \BT.  It has no source in \BS.  Similarly $\BS_2$ is mapped into \BT\ and $\BS^+_2$ extends $\BS_2$.  $\BT^+$ can be mapped into an extension $\BS^+_2$ of $\BS_2$.  If not all of $\BT^+$ is mapped, at least some of it, say $\alpha$, is mapped onto $\alpha'$.

Think of $\BS_2$ as modal logic, mapped into classical logic.  Think of $\alpha$ as a quantifier of second order, which is realised/mapped as extending modal logic (as a new modality $\square_\alpha$?, for example).

The systems \BE\ and $\BE^+$ are mapped into  modal logic. Think, for example, of \BE\ as a substructural logic or a default system or some normative system or as an argumentation system being interpreted in modal logic.

We now discuss possible properties of such embeddings.

\paragraph{Type 1: encapsulation of \BS\ in \BT.}
This type gives names in \BT\ to the elements of \BS\ and mirrors in \BT\ in all the movements of \BS.  It is like embedding \BS\ in \BT\ using G\"{o}del numbers.  Such encapsulation is used in translating systems into classical logic in order to use Theorem Proving machines to talk about/manipulate \BS.  None of the papers involved above are pure encapsulation.  Encapsulation is very general, \BS\ can be a recipe for making pizzas and \BT\ can just talk about it.  \BT's role in encapsulation is that of a Turing machine.

\paragraph{Type 2: Meta-level embedding of \BS\ into \BT.}
This type is most common. The language \BT\ describes in a meaningful way the system \BS.  \BT\ acts as a meta-level language describing \BS\ by internally mirroring \BS.  To understand how this works, think of a mathematical theory modelling some physical phenomenon, say in an equational model governing the motion of a planet around the sun.  The meta-theory \BT\ describing modelling \BS\ (e.g.\ argumentation) in this way is not supposed to provide a meaning (semantics) for what it describes. It is just supposed to provide a formal language of accurate description of \BS, allow one to describe variation of \BS\ and enable the use of  properties and tools available in \BT\ to investigate \BS\ (very common is to use \BT\ for complexity issues for \BS).  Of course, one tends to use \BT\ which is well known, general enough and well investigated and well endowed with tools.

\paragraph{To summarise Type 2.}
\begin{itemlist}{(PPP)}
\item [(P1)] $\BT$ describes \BS\ from the meta-level point of view but does not give \BS\ any meaning.
\item [(P2)] \BT\ can describe at the same time several different $\BS_1, \BS_2$ and so can serve as an environment for suggesting combinations of $\BS_1$ and $\BS_2$.  For example, if \BT\ is classical logic and $\BS_1$ is modal logic and $\BS_2$ is intuitionistic logic, then embedding $\BS_1$ and $\BS_2$ into \BT, may suggest how to formulate modal intuitionistic logics.
\item [(P3)] Any tool available for \BT\ can be used for \BS.  So for example, andy complexity studies for classical logic can be used for modal logic through the embedding of modal logic inside classical logic.
\item [(P4)] New research ideas can be imported from \BT\ to \BS.  Some may be interesting for \BS\ and some may be just mathematics hacks of no interest.
\end{itemlist}

The paper of Dvorak {\em et al.}, paper \cite{540-35} is an example of this. The formal mathematical language is second order monadic first order logic.  This can serve as  the modelling language for the majority, if not all,  of the varieties of argumentation networks. It is intended by the authors to be to argumentation like ALGOL is to algorithms. It is an exact mathematical logic language strong enough to express whatever you want to say about argumentation networks. To quote the authors of \cite{540-35} own words:

\begin{quote}
Begin quote.\\
We propose the formalism of monadic second order logic (MSO) as a unifying framework for representing and reasoning with various semantics of abstract argumentation. We express a wide range of semantics within the proposed framework, including the standard semantics due to Dung, semi-stable, stage, cf2, and resolution-based semantics.   We provide building blocks which make it easy and straightforward to express further semantics and reasoning tasks.  Our results show that MSO can serve as a {\em lingua franca} for abstract argumentation that directly yields to complexity results. In particular, we obtain that for argumentation frameworks with certain structural properties the main computational problems with respect to MSO-expressible semantics can all be solved in linear time. Furthermore, we provide a novel characterisation of resolution-based grounded semantics.

\begin{center}
.\\.\\.
\end{center}

Starting with the seminal work by Dung, the area of argumentation has evolved to one of the most active research branches within Artificial Intelligence.  Dung's abstract argumentation frameworks, where arguments are seen as abstract entities which are just investigated with respect to how they relate to each other, in terms of ``attacks'', are nowadays well understood and deferent semantics (i.e., the selection of sets of arguments which are jointly acceptable) have been proposed.  In fact, there seems to be no single ``one suits all'' semantics, but it turned out that studying a particular setting within various semantics and to compare the results is a central research issue within the field. Different semantics give rise to different computational problems, such as deciding whether an argument is acceptable with respect to the semantics under consideration, that require different approaches for solving these problems.

This broad range of semantics for abstract argumentation demands for a {\em unifying framework} for representing and reasoning with the various semantics. Such a unifying framework would allow us to see what the various semantics have in common, in what they differ, and ideally, it would offer generic methods for solving the computational problems that arise within the various semantics. Such a unifying framework should be general enough to accommodate most of the significant semantics, but simple enough to be decidable and computationally feasible.\\
End quote
\end{quote}

Properties (P1)--(P4) mentioned above hold for this case. See also the discussion in Section 5.1 above.  Our own interpretation is object level and gives the argumentation network a meaning in terms of strong negation $N$.  As we mention in Section 5.1, the meta-level interpretation has to translate ``$x$ attacks $y$'' as is, by providing a predicate letter ``$R$'' for attack and write $xRy$. The meta-level interpretation does not give a meaning to $R$. Our interpretation writes $x\to Ny$ for the attack of $x$ on $y$. So when $x$ and $y$ are instantiated by, e.g., wffs $x=\alpha$ and $y=\beta$, we get a meaning for the attack of $\alpha$ on $\beta$, namely the meaning our logic gives to $\alpha\to N\beta$.  The paper of Dvorak cannot write ``$\alpha R\beta$'', but even if it could do so, it would have to wait for us to say  to Dvorak {\em et al.} what we mean by ``$\alpha$ attacks $\beta$'' and then Dvorak {\em et al.} would try to say it formally in  monadic second order classical logic.

The other paper of Dvorak, paper \cite{540-36}, also provides a more propositional meta-language to describe more tightly higher level
argumentation network. Properties (P1)--(P4) above still apply.  The emphasis here is on (P3).

To quote the authors' own words:

\begin{quote}
Begin quote:\\
This paper reconsiders Modgil's Extended Argumentation Frameworks (EAFs) that extend Dung's abstract argumentation frameworks by attacks on attacks.  This allows to encode preferences directly in the framework and thus also to reason about the preferences themselves. As a first step to reduction-based approaches to implement EAFs, we give an alternative (but equivalent) characterisation of acceptance in EAFs. Then we use this characterisation to provide EAF encodings for answer set programming and propositional logic. Moreover, we address an open complexity question and the expressiveness of EAFs.\\
End quote.
\end{quote}

The authors of \cite{540-36} use propositional encoding of the argumentation network. This is still meta-level. To explain briefly what is the difference:  while predicate logic would formalise ``$x$ attacks $y$'' as ``$xRy$'', using a predicate $R$, the propositional encoding uses a propositional atom $r_{x,y}$.  You need a new atom for each pair $x,y$. So $r_{x,y}=\top$ means that $x$ does attack $y$ and $r_{x,y} =\bot$ means that  $x$ does not attack $y$.

So for example to formalise:
\begin{itemize}
\item (For all $x$), if ($x$ is ``in'' iff all its attackers are ``out'')
\end{itemize}
We can write in predicate logic:
\begin{itemize}
\item  $\forall x [Q_{\rm in} (x) \leftrightarrow \forall y (yRx \to Q_{\rm out} (y))]$
\end{itemize}
or in propositional logic
\begin{itemize}
\item  $\bigwedge_{x\in S} [x\leftrightarrow \bigwedge_{y \in S} ({r_{x,y}}\to \neg y)]$.
\end{itemize}
The advantage of the use of $r_{x,y}$ is that we can say directly that $x$ does not attack $y$, i.e.\ we can write  $\neg r_{x,y}$. Our logic \BC\BN\BN\ cannot do that. Although we can represent  in our logic  ``$x$ attacks $y$" as   ``$x \to Ny$" and use it to calculate extensions, we cannot represent ``$x$ does not attack $y$" as  ``$\neg (x\to  Ny)$", because $x$ may be out,   i.e.\ $\neg x$ is true, and so $x\to  Ny$ is true, even though $x$ does not attack $y$. We shall remedy this in Part 2 of this paper. Note also that using direct naming of attacks via $xRy$ or  $r_{x,y}$ we can also formalise higher level attacks of the form $z$ attacks $r_{x,y}$.

So to summarise, paper \cite{540-36} is also a meta-level interpretation geared towards supplying algorithms.

The paper of Doutre {\em et al.} \cite{540-32} is also meta-level, using propositional logic.  It is also geared towards algorithms.  To quote the authors:
\begin{quote}
Begin quote:\\
We provide a logical analysis of abstract argumentation frameworks and their dynamics. Following previous work, we express attack  relation and argument status by means of propositional variables and define acceptability criteria by formulas of propositional logic. We here study the dynamics of argumentation frameworks in terms of basic operations on these propositional variables, viz.\ change of their truth values. We describe these operations in a uniform way with in well known variant of propositional Dynamic Logic {\sf PDF}: the Dynamic Logic of Propositional Assignments, {\sf DL-PA}.  The atomic programs of {\sf DL-PA} are assignments  of propositional variables to truth values, and complex programs can be built by means of the connectives of sequential and nondeterministic composition and test. We start by showing that in {\sf DL-PA}, the construction of extensions can be performed by a {\sf DL-PS} program that is parameterized by the definition of acceptance.  We then mainly focus on how the acceptance of one or more arguments can be enforced and show that this can be achieved by changing the truth values of the propositional variables describing the attack relation in a minimal way.\\
End quote.
\end{quote}

The paper of Grossi \cite{540-37} is also meta-level, but it uses modal logic as the meta-level language.  The schematic situation is described in Figure \ref{540-FMay19-1}  by the embedding schema of [\BE\ and $\BE^+$] into [$\BS_2$ and $\BS^+_2$] which in turn are embedded into [\BT\ and $\BT^+$].  Here \BE\ is argumentation networks, $\BS_2$ is modal logic, $\BS^+_2$ is an expansion which contains extra connectives of modal logic to compensate for the  lack of quantifiers.  \BT\ is classical predicate logic. $\alpha$ is the extra ``quantifier'' connectives imported by Grossi into modal logic to  enable Grossi to interpret argumentation. The basic modal accessibility is taken by Grossi to be the inverse of the attack relation:
\begin{itemize}
\item $x$ attacks $y$ means $x$ is an accessible world to $y$.
\end{itemize}
The situation is best described by Grossi's own words below.  You can immediately see that Grossi focusses on the (P4)  aspects of the translation:
\begin{quote}
Begin quote:\\
The paper presents a study of abstract argumentation theory from the point of view of modal logic. The key thesis upon which the paper builds is that argumentation frameworks can be studied as Kripke frames. This simple observation allows us to import a number of techniques and results from modal logic to argumentation theory, and opens up new interesting avenues for further research. The paper gives a glimpse of the sort of techniques that can be imported, discussing complete calculi for argumentation, adequate model-checking and bisimulation games and sketches an agent for future research at the interface of modal logic and argumentation theory.\\
End quote.
\end{quote}

The paper of Besnard, Doutre and Herzig, {\em Encoding argument graphs in logic} is a meta-level paper discussing properties (specification) which any meta-level interpretation should satisfy. For example, the interpretations of papers \cite{540-32,540-35,540-36,540-37} should be checked to see if they satisfy the principles outlined in paper \cite{540-45}.  Paper \cite{540-45} is a meta-meta-level paper. We quote the authors' own description of their paper:

\begin{quote}
Begin quote.\\
Argument graphs are a common way to model argumentative reasoning. For reasoning or computational purposes, such graphs may have to be encoded in a given logic. This paper aims at providing a systematic approach for this encoding.  This approach relies upon a general, principle-based characterisation of argumentation semantics.

In order to provide a method to reason about argument graphs \cite{540-47}, Besnard and Doutre first proposed encodings of such graphs and semantics in propositional logic \cite{540-46}.  Further work by different authors following the same idea was published later \ldots  However, all these approaches wee devoted to specific cases in the sense that for each semantics, a dedicated encoding was proposed from scratch.  We aim here at a generalisation, by defining a {\em systematic} approach to encoding argument graphs (which are digraphs) and their semantics in a logic $\vdash$.  Said differently, our objective is to capture the extensions under a given semantics of an argument graph in a given logic (be it propositional logic or any other logic). We hence generalise the approach originally introduced in \cite{540-46} by parametrizing the encoding in various ways, including principles defining a given semantics.

We consider abstract arguments first, and then provide guidelines to extend the approach to structured arguments (made up of a support that infers a conclusion).\\
End quote.
\end{quote}

\subsection{Comparison with the paper \cite{540-34} of Arieli and Caminada}
To explain what Arieli and Caminada are doing and to evaluate and compare it with our paper, let us start with the meta-level point of view of Dvorak \cite{540-35} and the discussion in Section 5.1 above. We know that the Caminada labelling has three values \{in, out, und\}.  The meta-level approach which interprets argumentation in monadic second order classical predicate logic, would use variables $x, y, z\ldots$ for arguments and the binary predicate $R$ for the attack relation and three unary predicates $Q_1, Q_0$ and $Q_?$ for the values in, out and undecided respectively. See section 5.1.  The axioms of $\Delta (R, Q_0, Q_1, Q_?)$ say in predicate logic, among
 other things, that each $x$ gets exactly one value. We are now ready to lead, step by step, from the above to the Arieli and Caminada paper.

 Let us start from a given $(S, R)$.  For each $a\in S$, consider $Q_1(a), Q_0(a)$ and $Q_?(a)$ as atomic propositions of the classical propositional calculus.  Since these are all connected with the letter ``$a$'' we can change notation and write
 \begin{itemize}
\item $ a^+$ for $Q_1(a)$
 \item $a^-$ for $Q_0(a)$
 \item $a^?$ for $Q_?(a)$
 \end{itemize}
 Since we know that exactly one of the above can be true, we can forget about $a^?$ and use only the pair $(a^+,a^-)$.  We have:
 \begin{itemize}
 \item If $a$ is in, then we have $a^+=\top, a^-=\bot$.
 \item If $a$ is out we have $a^+=\bot, a^-=\top$.
\item  If $a=$ undecided, we have $a^+=\bot, a^-=\bot$.
 \end{itemize}
 We must add the restriction that we can never have both $a^+=\top$ and $a^-=\top$.

 We thus have turned the logic with $Q_1(x), Q_0(x), Q_?(x)$ into a three-valued logic with values $\Bt =(1,0), \Bf=(0,1)$ and $?=(0,0)$, where these values mean as follows:
 \begin{itemize}
 \item  $x=\Bt$ means $x$ is ``in'', or equivalently $x^+=\top$ and $x^-=\bot$ or equivalently $(x^+,x^-)=(\top,\bot)$ or equivalently $Q_1(x)\wedge\neg Q_0(x)=\top$.
 \item $x=\Bf$ means $x$ is ``out'', or equivalently $x^+=\bot$ and $x^-=\top$ or equivalently $(x^+,x^-)=(\bot,\top)$ or equivalently $\neg Q_1(x) \wedge Q_0 (x)=\top$.
 \item $x=?$ means $x$ is undecided or equivalently $x^+=x^-=\bot$ or equivalently $(x^+,x^-)=(\bot,\bot)$ or equivalently $\neg Q_1 (x)\wedge\neg Q_0(x)=\top$.
 \end{itemize}

 The value $x^+=x^-=\top$ is forbidden.

 We now have a propositional calculus with three values $\{\Bt,\Bf, ?\}$ and a reduction of each proposition $x$ into two propositions $x^+,x^-$, such that $x=(x^+,x^-)$ and the restriction on the assignments on $\{x^+,x^-|x\in S\}$ which says that $x^+=x^-=\top$ is forbidden.

 If we add to this version of classical propositional calculus quantifiers over propositions $\forall x, \exists x, x\in S$, we get the system of \cite{540-34}.

 We can also turn the pairs $(x^+,x^-)$ into a lattice calculus. Note that we have\\
  $x=(x^+,x^-)= \mbox{ (say) } (x_1, x_2)$\\
  $\neg x =(x^-,x^+)=(x_2, x_1)$.\\
  But also we can write similarly for $y$ and get:\\
  $x=Q_1(x) \wedge\neg Q_0 (x) = (x_1, x_2)$\\
  $y=Q_1(y)\wedge\neg Q_0(y) =(y_1,y_2)$.

  Therefore
  \[\begin{array}{rcl}
  x\wedge y &=& Q_1(x)\wedge Q_1(y)\wedge\neg Q_0(x)\wedge\neg Q_0(y)\\
  &=& Q_1(x) \wedge Q_1(y)\wedge\neg (Q_0(x)\vee Q_0(y))
  \end{array}
  \]
  We therefore get as a rule of the calculus
  \[
  (x_1,x_2)\wedge (y_1,y_2) =(x_1\wedge y_1, x_2\vee y_2).
  \]
  We get other rules in a similar fashion.\footnote{In Part 2  we shall compare with Gabbay's 1985 (published in 2012 as \cite{540-40}). }

  Now that we have reduced the idea of \cite{540-34} to the meta-level system of \cite{540-35}, we can read the quotation word by word from \cite{540-34} and understand it in a new light.  We have changed the notation slightly.    \cite{540-34} uses ``$\bot$'' for undecided, we use ``?'' so as not to confuse our readers.

  \begin{quote}
  Begin quote.\\
  \paragraph{Three-valued Semantics and Signed Formulas.}
  As indicated previously, our purpose in this paper is to provide a this, logic-based, perspective on argumentation frameworks, and to relate it to the two other points of view presented in the two previous subsections. [Gabbay note: these are 2.1: Extension based semantics and 2.2: labelling based semantics.]  In this section we define the framework for doing so, using {\em signed theories}.  Following \cite{540-41}, we introduce these theories in the context of three-valued semantics (see leo \cite{540-43}).

  Consider the truth values $t$ (`true'), $f$ (`false') and ? (`neither true nor false'). A natural ordering, reflecting differences in the `measure of truth' of else elements, is $f < ? < t$. The meet (minimum) $\wedge$, join (maximum) $\vee$ and the order reversing involution $\neg$, define by $\neg t =f, \neg f=t$ and $\neg ?=?$, are taken to be the basic operators on $\leq$ for defining eh conduction, disjunction, and the negation connectives (respectively) of Kleene's well-known three-valued logic (see \cite{540-44}).  Another operator which will be useful in the sequel is defined as follows: $a\supset b =t$ if $\in \{f,?\}$ and $a\supset b=b$ otherwise (see \cite{540-42} for some explanations why this operator is useful for defining an implication connective). The truth tales of these basic connectives are given below.

  \[
  \begin{array}{cccc}
  \begin{array}{c|ccc}
  \vee&t&f&?\\
  \hline\\
  t&t&t&t\\
  f&t&f&?\\
  ?&t&?&?
  \end{array}
  & \begin{array}{c|ccc}
  \wedge&t&f&?\\
  \hline\\
  t&t&f&?\\
  f&f&f&f\\?&?&f&?
  \end{array}
  &
  \begin{array}{c|ccc}
  \supset&t&f&?\\
  \hline\\
  t&f&f&?\\
  f&t&t&t\\
  ?&t&t&t
  \end{array}
  &
  \begin{array}{c|c}
  \neg &\\
  \hline\\
  t&f\\
  f&t\\
  ?&?
  \end{array}
  \end{array}
  \]

  The truth values may also be represented by pairs of two-valued components of the lattice $(\{0,1\}, 0 < 1)$ as follows: $t=(1,0), f=(0,1), ?=(0,0)$.   This representation may be intuitively understood as follows: If a formula $\psi$ is assigned the value $(x, y)$, then $x$ indicates whether $\psi$ should be accepted and $y$ indicates whether $\psi$ should be rejected. As shown in the next lemma, the basic operators considered above may also be expressed in terms of this representation by pairs.

  \paragraph{Lemma 11.}  Let $x_1, x_2, y_1, y_2 \in \{0,1\}$. Then:
  \[\begin{array}{l}
  (x_1, y_1) \vee (x_2, y_2) =(x_1\vee x_2, y_1\wedge y_2), \\
  (x_1, y_1)\wedge (x_2, y_2) =(x_1\wedge x_2, y_1\vee y_2),\\
 (x_1,y_1)\supset (x_2, y_2) =(\neg x_1\vee x_2, x_1\wedge y_2), \\
 \neg (x,y) =(y,x).
  \end{array}
  \]

  \newcommand{\st}{{\sf t}}
  In our context, the three values above are used for evaluating formulas in a propositional language $\CL$, consisting of a set of atomic formulas {\sf Atoms}$(\CL)$, the propositional constants \st\ and {\sf f}, and logical symbols $\neg, \wedge, \vee, \supset$.  We denote the atomic formulas of $\CL$ by $p, q, r$, formulas by $\psi, \phi$ and sets of formulas (theories) by $\CT, \CS$. the set of all atoms occurring in a formula $\psi$ is denoted by {\sf Atoms}$(\psi)$ and {\sf Atoms}$(\CT) =\{\mbox{{\sf Atoms}}(\psi) |\psi \in \CT$ is the set of all the atoms occurring in the theory $\CT$.  Now, a {\em valuation} $v$ i s a function that assigns to each atomic formula a truth value from $\{t, f, ?\}$, and $v(\st) = t, v ({\sf f}) =f$.   Any valuation is extended to complex formulas in the obvious way. In particular, $v(\psi\circ \phi) =v(\psi)\circ v(\phi)$ for every $\circ \in \{\neg, \wedge, \vee, \supset\}$. A valuation $v$ {\em satisfies} $\psi$ off $v(\psi) =t$. A valuation that satisfies every formula in $\CT$ is a {\em model} of $\CT$. The set of models of $\CT$ is denoted by {\em mod}$(\CT)$.

  \paragraph{Definition 12.}  Let $\CL$ be a propositional language with  set of atoms {\sf Atoms}$(\CL)$. A {\em signed alphabet} {\sf Atoms}$\pm (\CL)$ is a set that consists of two symbols $p^+, p^-$ for each atom $p \in {\sf Atoms}(\C)$. The language over {\sf Atoms}$^\pm (\CL)$ is denoted by $\CL^\pm$. A value $v$ for $\CL^\pm$ is called {\em coherent}, if there is no $p \in {\sf Atoms}(\CL)$ such that both $v(p^+)=1$ and $v(p^-)=1$.\\
  End quote.
  \end{quote}

 We now compare with our system {\bf CN}. We have a two world modal intuitionistic logic and our values are
 \begin{itemize}
 \item $x=$ in means $x\wedge \neg Nx$
 \item $x=$ out means $\neg x \wedge Nx$
 \item $x=$ undecided means $\neg x \wedge\neg Nx$
 \end{itemize}
 If we let $x^+=x$ and $x^-=Nx$, we get the system of Arieli and Caminada
 \begin{itemize}
 \item $x=$ in iff $(x^+,x^-)=(\top, \bot)$ iff $x\wedge\neg Nx =\top$
 \item $x=$ out iff $(x^+,x^-)=(\bot, \top)$ iff $\neg x\wedge Nx =\top$.
 \item $x=$ und iff $(x^+,x^-) =(\bot, \bot)$ iff $\neg x \wedge\neg Nx=\top$
 \item $(x^+,x^-)=(\top,\top)$ is not allowed.  This means that $x\wedge Nx$ is not allowed. This is our intuitionistic restriction $\vdash Nx\to \neg x$.
\end{itemize}

 To summarise the comparison, Arieli and Caminada turn argumentation into a Kleene 3 valued logic and manipulate  it in the meta-level using Arieli's \cite{540-41,540-42,540-43} lattice theoretic methods.  Our paper manipulates it in the object level in a modal intuitionistic logic.

 The perceptive reader might ask whether the paper of Arieli and Caminada carries any message or point of view, beyond the mathematical manipulation of truth values?  In comparison, the message of our paper is that it is object level, translating the attack relation $xRy$ into $x\to Ny$, where $N$ is strong negation embedded in a modal intuitionistic logic. What is the message of \cite{540-34}?  What can it be used for?  Our answer to this question is that the Arieli--Caminada paper is very important in the context where a calculus of degrees of undecided is needed. If we need to build an argumentation method for an application area with many degrees of undecided, then we need a calculus for undecided values and we can follow the lead of \cite{540-34}.

 \begin{figure}
 \centering
 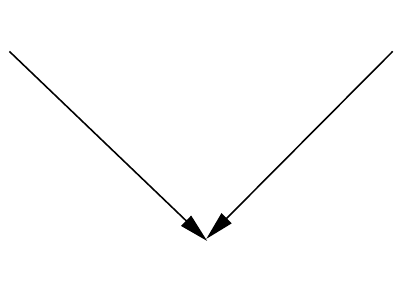
 \caption{}\label{540-FCM-1}
 \end{figure}

Consider Figure \ref{540-FCM-1} . We need to be able to assign a value to $x$ in Figure \ref{540-FCM-1}, given the variety of undecided values of the attackers of $x$, namely $ y_1\comma y_m$.

Our system {\bf CNN} is not general enough for this purpose.  We are committed to the intuitionistic restriction, so we can have undecided values like $(\bot, \top, \top, \top), (\bot, \bot, \top, \top) (\bot, \bot, \bot, \top)$ and no more, for a 4-world modal model.  The Arieli--Caminada approach can handle more, many more values, for example they can have $(\bot, \top, \bot, \top)$, etc.

Imagine a highly divergent group of witnesses of a tragic event. They are all undecided and contradicting and unreliable for a multitude of reasons.  The traditional extension based semantics, as well as the Caminada labelling based semantics will give them all undecided. This is {\em not} what we need. We need a calculus of undecided that will tell us something better. The equational approach \cite{540-39} will give numbers, the modal provability \cite{540-31} might give us several options for undecided, but these papers put forward fixed approaches derived from a fixed points of view serving some other methodological aims. The fact that they also yield several undecided values is only a side effect. Same applies to our logic {\bf CNNk}. These system are not intended as calculi for undecided values and they offer no flexibility. We need an approach which can generate alternative calculi of undecided which can be tailored for different applications. The Arieli--Caminada paper is a start. We leave this research to a future paper.

 It is regrettable that Arieli and Caminada took their paper in the (traditional) direction of generating extensions. It is time for all of us to get out of this way  of thinking, which is nothing but a barrier to our imagination.

\subsection{Comparing with  paper \cite{540-27} ---  Methods for solving reasoning problems in abstract argumentation. A survey}

\cite{540-27} is an excellent survey of methods for computing extensions. It divides the approaches into two:
\begin{enumerate}
\item
Direct methods
\item  Reduction  methods, where the argumentation systems are translated into other systems which have good computational facilities which can be used to calculate extensions.
\end{enumerate}

So the paper \cite{540-27} considers translations only as means of calculating extensions. In comparison, we consider translations as means of giving new and different meaning to the attack relation.

In the authors own words, for example, they say
\begin{quote}
Begin quote.\\
The underlying idea of the reduction approach is to exploit existing efficient software which has originally been developed for other purposes. To this end, one has to formalize the reasoning problems within other formalisms\dots In this approach, the resulting argumentation systems directly benefit from the high level of sophistication today's systems reached\dots
In the remainder of Section~3 we shall present the concepts behind other reduction-based approaches, for instance, the equational approach as introduced by Gabbay in \cite{540-39} and the reductions to monadic second order logic  as proposed in \cite{540-35}.

We will not consider the vast collection of extensions to Dung's frameworks, like,\footnote{Editorial comment:  We omit the references in this list }  value-based, bipolar, extended, constrained, temporal, practical, and fibring argumentation frameworks, as well as argumentation frameworks with recursive attacks, argumentation context systems, and abstract dialectical frameworks. We also exclude abstract argumentation with uncertainty or weights here.\\
End quote.
\end{quote}


So to continue our comparison with paper \cite{540-27}, we highlight the fact that the interest of the authors of paper \cite{540-27} in the
current paper (as well as in papers \cite{540-35} and \cite{540-39}, for example) is limited solely to the question of how these papers provide means to compute extensions.

\section{Conclusion and future research}
In this paper we introduced a modal logic {\bf CNN} with strong negation $N$ playing the role of modality. The semantics has two possible worlds, here (actual world $t$) and there (possible world $s$) with
\[\begin{array}{l}
t\vDash NA \mbox{ iff } s\vDash \neg A\\
s\vDash NA\mbox{ iff } t\vDash \neg A
\end{array}
\]
and with the restriction that for atoms $q$ we have
\[
t\vDash q \mbox{ implies } s\vDash q.
\]

We showed how to translate any abstract argumentation frame $\CA$ into {\bf CNN} via a theory $\Delta_{\CA}$ of {\bf CNN} in such a way that all the models of $\Delta_{\CA}$ give exactly the complete extensions of $\CA$.

Many properties and operations on $\CA$ (such as joint attacks, higher level attacks, bipolar argumentation, and more) become very simple when done in {\bf CNN}.

Further research postponed to Part 2 of this paper includes the following topics:
\begin{enumerate}
\item Translate Assumption Based Argumentation (ABA) into {\bf CNN}.
\item Use intuitionistic logic as a basis for {\bf CNN} (call it {\bf INN}) and see how to obtain intuitionisitc based argumentation. In fact, given any logic \BL\ defined by rules of the type used in ASPIC, we can form {\bf LNN} and translate ASPIC (for \BL) into {\bf LNN}.

This will give us a better understanding of ASPIC and in fact also allow us to define a generalisation of Assumption Based Argumentation.
\item Fully investigate ADF (Abstract Dialectical Frameworks) in {\bf CNN}.
\item It is easy to define predicate argumentation using predicate {\bf CNN}. Simply use predicate {\bf CNN} and allow instantiations of any $(S, R)$ with classical predicate wffs. Translate $\varphi \tO \psi$ as $\varphi \to N\psi$. The semantics of predicate {\bf CNN} will give us the complete extensions for the network.
\item Investigate bipolar networks and {\bf CNN}, as discussed in {\bf CC7}.
\item Investigate in detail iterated disjunctive attacks in {\bf CNN} of the kind of Figure  \ref{540-FS5}.

\item The logic {\bf CNN} is based on two possible worlds $t$ (actual world) and $s$ ($s$ is accessible to $t$). We can investigate the possibility of adding more worlds to the chain, say use 3 worlds $t$ (actual), $s$ (accessible to $t$) and $s'$ (accessible to $s$).

Define
\[\begin{array}{l}
t\vDash NA \mbox{ iff } s\vDash \neg A\\
s\vDash NA \mbox{ iff } s'\vDash \neg A\\
s'\vDash NA \mbox{ iff } t\vDash \neg A.
\end{array}
\]
This new system (call it {\bf CNN3}) might allow us to do the following:
\begin{enumerate}
\item have several grades of undecided, see the comparison in Section 9.2
\item allow for attacks of a node $x$ on another network $\CA_x$.
\item  allow us to use truth in the middle node as a meta-level vehicle to express properties of extensions, (since {\bf NN} refers to the actual world from the middle node). We can force the models to yield, for example, only preferred extensions.
\item Investigate the connection with Logic Programming.
\item allow us to formalise in the object level the notion of ``$x$ does not attack $y$".
\end{enumerate}
\item
Following  the connection with Logic Programming and David Pearce's Equilibrium Logic \cite{540-30}, we investigate whether Pearce's logic can implement argumentation network similarly to the way {\bf CNN}  does.
\item Extend the negation as failure approach of \cite{540-40} to argumentation and compare with this paper and with \cite{540-34}.
\end{enumerate}

We shall investigate these possibilities in Part 2.

\section*{Acknowledgements}
We are grateful to the referees and to
Philippe Besnard ,P. M. Dung , Davide Grossi,   Beishui Liao, David Pearce, Odinaldo Rodrigues  Matthias Thimms and  Leon van der Torre,  for their valuable comments.

\section*{Appendix: Formal definition of the modal logic $\SN$}\label{mike.appendix.2}
In this section , we change the syntax of {\bf CNN} slightly. Besides the classical connectives and the connective $N$, we add a constant t for the actual world. This helps with completeness proofs. We should call, if we want to be strict ,  the modified system with the slightly different name, say,  {\bf CNN*}, but we are not going to bother.

\begin{definition}\label{mike.semantics}
A {\em model} for $\SN$ is a pair $\pa{\Num,v}$ where  $v$ assigns subsets of $\Num$ to atomic propositions such that if $1\in v(p)$ then $2\in v(p)$
\end{definition}

\begin{definition}
If $M= \pa{\Num,v}$ is a model of $\SN$ then define truth in $M$ as follows, (where $m=1,2$ and $t$ is the propositional constant for world 1).

\begin{itemize}
\item
$\istrue{m}{p}$ iff $m\in v(p)$
\item
$\istrue{m}{\top}$
\item
$\isntrue{m}{\bot}$
\item
$\istrue{m}{\one}$ iff $m=1$
\item
$\istrue{m}{\neg A}$ iff $\isntrue{m}{A}$
\item
$\istrue{m}{A\wedge B}$ iff $\istrue{m}{A}$ and $\istrue{m}{B}$
\item
$\istrue{m}{A\vee B}$ iff $\istrue{m}{A}$ or $\istrue{m}{B}$
\item
$\istrue{m}{A\to B}$ iff $\isntrue{m}{A}$ or $\istrue{m}{B}$
\item
$\istrue{m}{NA}$ iff either $m=1$ and $\isntrue{2}{A}$, or $m=2$ and $\isntrue{1}{A}$
\end{itemize}
Write $\istrue{}{A}$ when  $\istrue{m}{A}$ for all $m$ (i.e. $m=1$ or $m=2$); write $\istrue[]{m}{A}$ when $\istrue{m}{A}$ for all $M$; and write $\istrue[]{}{A}$ when $\istrue{m}{A}$ for every $m$ and every $M$.
\end{definition}

Intuitively, a model is a finite linear ordering of worlds -- in the sense of familiar possible world semantics for modal logic -- where logical connectives are defined as usual and an atomic proposition $p$ is true at a world, then it is true at all `later' worlds. The new connective $N$ has the semantics that $NA$ is true at a world if and only if $A$ is false as the `next' world where, in the case of interpreting $N$,  `next' cycles back to the first world at the last world (so $NA$ is true at the last world when $A$ is false at the first).

The condition on atomic propositions ensures that $N$ is a kind of negation, at least for atomic propositions interpreted at the first world. If $\istrue{1}{Np}$ then $\isntrue{2}{p}$ and so $\isntrue{1}{p}$. Moreover the condition on interpreting $N$ ensures that $\istrue{1}{NNp\rightarrow p}$ as $\istrue{1}{NNp}$ implies $\isntrue{1}{Np}$ implies $\istrue{1}{p}$.

\begin{definition}
The theorems of $\SN$ are defined as follows:
$$
\begin{gathered}
\begin{array}{cl}
& \text{tautologies} \\
(\K)& N (A\wedge B) \leftrightarrow N A\vee N B \\
(\F) & \neg N A \leftrightarrow N \neg A \\
(\C) & A\to NN A \\
\end{array}
\end{gathered}
\qquad
\begin{gathered}
\begin{array}{cl}
(\A)&\one \rightarrow (p\rightarrow N\neg p)) \\
(\T)&\one \leftrightarrow N\one
\end{array}
\end{gathered}
$$
$$
\infer[(\MP)]{B}{A&A\to B} \\
\qquad
\infer[(\N)]{N\neg A}{A}
$$
For a set $\Gamma$ we write $\Gamma\vdash A$ when there are $A_1\dots A_n\in\Gamma$ such that $\vdash (A_1\wedge\dots\wedge A_n)\to A$.
\end{definition}

\begin{lemma}\label{mike.subs.lemma}
$\vdash A\to B$ implies $\vdash N B\to N A$, and $\vdash A\leftrightarrow B$ implies $\vdash N A\leftrightarrow N B$,
\end{lemma}
\begin{proof}
The first part follows by the following derivation:
$$
\begin{array}{ll}
\vdash A\to B & \text{assumption}\\
\vdash \neg(A\wedge \neg B) & \text{tautologies and $(\MP)$}\\
\vdash N\neg\neg(A\wedge \neg B) & (\N)\\
\vdash  N(A\wedge \neg B) & \text{tautologies, $(\MP)$ and $(\F)$} \\
\vdash NA \vee N \neg B & (\K)\\
\vdash NA\vee\neg N B & \text{tautologies, $(\MP)$ and $(\F)$} \\
\vdash N B\to N A & \text{tautologies and $(\MP)$}
\end{array}
$$
It is worth expanding one of the steps further: $\vdash N\neg\neg C$ implies by $(\F)$ that $\vdash \neg N\neg C$; also, $\vdash \neg NC\to N\neg C$ is a consequence of $(\F)$. So, using tautologies and $(\MP)$, $\vdash  \neg N\neg C\to \neg\neg N C$ thence $\vdash  \neg N\neg C\to N C$ . Therefore $\vdash N\neg\neg(A\wedge \neg B)$ implies $\vdash  N(A\wedge \neg B)$.

The rest of the lemma follows from the fact that $A\leftrightarrow B$ is short for $(A\to B)\wedge (B\to A)$.
\end{proof}

\begin{theorem}\label{mike.subs.thrm}
$\vdash B\leftrightarrow C$ implies $\vdash A[p/B]\leftrightarrow A[p/C]$
\end{theorem}
\begin{proof}
By induction on $A$ using~\ref{mike.subs.lemma}.
\end{proof}

\begin{theorem}\label{mike.ncn.complete}
$\istrue[]{}{A}$ iff $\SN\vdash A$ 
\end{theorem}
\begin{proof}
For the right-left direction it is a simple matter to verify, by induction on derivations, that $\SN\vdash A$ entails that for any $M$ and any $m$, $\istrue{m}{A}$.

For the right left direction we argue that if $\Sigma$ is consistent then there is an $M$ such that, for every $A\in\Sigma$, $\istrue{m}{A}$ for $m=1$ and $m=2$. First, since $\SN\nvdash A$ then using familiar methods, we can find a set $\Gamma$  that is consistent with respect to $\SN$ and maximal in the sense that $B\in\Gamma$ or $B\nin\Gamma$ for every $B$. Moreover, we have that $\Gamma\vdash A$ implies $A\in\Gamma$

Now, consider the set $\Delta=\upa{B\mid N\neg B\in \Gamma}$. If $B,\neg B\in \Delta$ then both $N\neg B\in\Gamma$ and $N \neg\neg B\in\Gamma$, and so by~\ref{mike.subs.thrm}, $(\F)$ and tautological reasoning, $\neg N B, N B\in\Gamma$ which is impossible as $\Gamma$ is consistent. Also, since $\neg NB\in\Gamma$ or $N B\in\Gamma$, it follows that $N\neg B\in\Gamma$ or $N B\in\Gamma$, and so $B\in\Delta$ or $\neg B\in\Delta$. We conclude that $\Delta$, like $\Gamma$, is maximal and consistent.

Now, by the definition of $\Delta$ and the fact it is maximal consistent, $NB\in\Gamma$ implies $B\nin\Delta$; and if $NB\nin\Gamma$ then $N\neg B\in\Gamma$ and so $B\in\Delta$. Also, if $B\in\Gamma$ then, by $(\C)$, $NNB \in \Gamma$ and so $N \neg \neg B\in \Gamma$ which implies that $\neg N B\in \Delta$ thence $NB\nin\Delta$; and conversely, if $B\nin\Gamma$ then $\neg B\in\Gamma$ and so $NN\neg B\in \Gamma$ which implies $N\neg N B\in\Gamma$ thence $NB\in\Delta$. We conclude that
$$
B\in\Gamma\text{ iff } NB\nin\Delta\text{\quad and\quad }B\in\Delta\text{ iff }NB\in\Gamma \leqno{(\dagger)}
$$

We now describe a model $M$ using $\Gamma$ and $\Delta$. Suppose that $\one\in\Gamma$, and set $v(p)=\upa{1,2}$ iff $p\in\Gamma$ and $v(p)=\upa{2}$ iff $p\in\Delta$. Then, as $\one\in\Gamma$, if $p\in \Gamma$ then by $(\A)$, $N\neg p\in \Gamma$ and so $p\in \Delta$. Thus $M$ is indeed a model as defined in~\ref{mike.semantics}.

We now verify, by induction on $A$, that $\istrue{1}{A}$ iff $A\in\Gamma$ and $\istrue{2}{A}$ iff $A\in\Delta$
\begin{itemize}
\item
If $A$ is atomic the result follows by the definition of $M$.
\item
If $A$ is $\one$ then, by assumption, $\one\in \Gamma$. Moreover, making use of $(\T)$, $\one\in \Gamma$ iff $N\one\in\Gamma$ iff $\neg N\one\nin\Gamma$ iff $N\neg\one\nin\Gamma$. If follows that $\one\in\Gamma$ iff $\one\nin\Delta$.
\item
If $A$ is a truth functional connective (i.e. is not $N$), then the result follows easily by the induction hypothesis and the maximal consistency of $\Gamma$ and $\Delta$.
\item
Suppose $A$ is $NB$. Then using~$(\dagger)$ above we have:
$$
\begin{array}{r@{\text{ iff }}ll}
\istrue{1}{A} & \istrue{1}{NB} \\
&  \isntrue{2}{B} \\
& B\nin\Delta & \text{ind. hyp.}\\
& NB \in\Gamma & (\dagger)\\
& A\in\Gamma
\end{array}
\quad
\begin{array}{r@{\text{ iff }}ll}
\istrue{2}{A} & \istrue{2}{NB} \\
&  \isntrue{1}{B} \\
& B\nin\Gamma & \text{ind. hyp.} \\
& NB \in\Delta& (\dagger) \\
& A\in\Delta
\end{array}
$$
\end{itemize}
This completes the induction.

The argument is similar if we suppose $\one\in\Delta$. Therefore, if $\Sigma$ is consistent, then there is a model $M$ such that $\istrue{}{A}$ for all $A\in\Sigma$.
\end{proof}

\begin{theorem}
If $A$ does not contain $\one$, nor any subformula of the form $NB$ unless $B$ is atomic, then $\DSN\vdash A$ iff $\SN\vdash\one\to A$.
\end{theorem}
\begin{proof}
Given \ref{mike.ncn.complete}, $\SN\vdash\one\to A$ iff $\istrue[]{1}{A}$.

For any $M$, the function $h$ over the atomic formulae of $\DSN$ such that:
$$
h(p)=1 \text{ iff } \istrue{1}{p}\quad\text{and}\quad h(Np)=1\text{ iff }\istrue{1}{Np}
$$
is clearly a truth assignment satisfying item 4 of~\ref{540-D1}, and moreover by an easy induction on $A$, $\istrue{1}{A}$ implies $h\vDash A$ (as $h\vDash\dots$ is characterised in~\ref{540-D1}).

On the other hand, if $h$ is a truth assignment according to~\ref{540-D1}, define $v$ -- and so $M$ -- as follows:
$$
v(p) =\left\{\begin{array}{ll} \upa{1,2} &\text{if } h(p)=1 \\ \varnothing & \text{if } h(Np)=1 \\ \upa{2} & \text{otherwise} \end{array}\right.
$$
(note that the same numbers have different functions here: the $1$ of the assertion $h(p)=1$ acts as a propositional truth value of~\ref{540-D1}; the $1$ of the assertion $v(p) = \upa{1,2}$ acts as a possible world).  An easy induction on $A$ verifies that $h\vDash A$ implies $\istrue{1}{A}$.

Therefore, there is an $h$ such that $h\vDash A$ iff there is an $M$ such that $\istrue{1}{A}$. It follows from this that $\DSN\vdash A$ iff $\SN\vdash\one\to A$.
\end{proof}


\begin{thebibliography}{99}
\bibitem{540-1}
 M. Caminada and D. M. Gabbay. A logical account of formal argumentation. {\em Studia Logica},
93(2-3):109--145, 2009.

\bibitem{540-2}
D. M. Gabbay. Semantics for higher level attacks in extended argumentation frames. Part 1: Overview. {\em Studia Logica}, 93(2--3):355--379, 2009.

\bibitem{540-3}
D. M. Gabbay. Fibring argumentation frames. {\em Studia Logica}, 93(2-3):231--295, 2009

 \bibitem{540-4}
 S. H. Nielsen and S. Parsons. A generalisation of DungÕs abstract framework for argumentation:
arguing with sets of attacking arguments. {\em LNCS Vol. 4766}, pp. 54--73. Springer,

 \bibitem{540-5}
 Do Duc Hanh, Phan Minh Dung and Phan Minh Thang. Inductive Defense for ModgilÕs Extended Argumentation Framework.
{\em J Logic Computation} doi: 10.1093/logcom/exq018. First published online: June 16, 2010

 \bibitem{540-6}
 S. Modgil. Reasoning about preferences in argumentation frameworks. {\em Artificial Intelligence},
173(9--10): 901--93, 2009.

 \bibitem{540-7}
 Pietro Baroni, Federico Cerutti, Massimiliano Giacomin, and Giovanni Guida. Encompassing attacks to attacks in abstract argumentation frameworks. In {\em ECSQARU}, pages 83--94, 2009.

 \bibitem{540-8}
Howard Barringer, Dov M. Gabbay, and John Woods. Temporal dynamics of support and attack networks: From argumentation to zoology. In {\em Mechanizing Mathematical Reasoning}, pages 59--98, 2005.

 \bibitem{540-9}
Pietro Baroni, Federico Cerutti,  Massimiliano Giacomin, and Giovanni Guida.
AFRA: Argumentation framework with recursive attacks.  {\em International Journal of Approximate Reasoning}
Volume 52, Issue 1, January 2011, Pages 19--37.


\bibitem{540-10}
 G. Brewka. {\em  Dialectical Frameworks}. Bologna, Dec. 2013 2

 \bibitem{540-11}
 G. Brewka and S. Woltran. Abstract dialectical frameworks. In {\em Proc. of the 20th International
Conference on the Principles of Knowledge Representation and Reasoning (KR
2010)}, pages 102--111, 2010.

\bibitem{540-12}
G. Brewka, P. Dunne, and S. Woltran. Relating the Semantics of Abstract Dialectical
Frameworks and Standard AFs. In {\em Proc. IJCAI-11}, pp. 780--785, 2011.

 \bibitem{540-13}
 Allen Hazen.  Comments on the Logic of Constructible Fasity (Strong Negation).
{\em Bulletin of the Section of Logic},
Volume 9/1 (1980), pp. 10--13.
Re-edition 2010 [original edition, pp. 10--15]


\bibitem{540-14}
P. Baroni, M. Giacomin, and G. Guida. SCC-recursiveness: a general schema for argumentation
semantics.  {\em Artificial Intelligence}, 168 (1--2):162--210, 2005.

\bibitem{540-15}
 D. M. Gabbay. The equational approach to CF2 semantics, short version. In {\em Proceedings
COMMA 2012, ComputationalModels of Argument}, Edited by Bart Verheij, Stefan Szeider,
and Stefan Woltran, IOS press, 2012, pp 141--153.

\bibitem{540-16}
D. M. Gabbay.     Theory of Semi-Instantiation in Abstract Argumentation. To appear {\em Logica Universalis}.

\bibitem{540-17}
D. M.  Gabbay. {\em Meta-Logical Investigations in Argumentation Networks }.
Research Monograph  College Publications,  2013, 770 pp

\bibitem{540-18}
D. M. Gabbay and O. Rodrigues.
Probabilistic argumentation, an equational approach. To appear in {\em Logica Universalis}, \href{http://arxiv.org/abs/1503.05501}{arxiv.org/abs/1503.05501}.

\bibitem{540-19}
A. Bochman. Propositional argumentation and causal reasoning.  In {\em IJCAI-05, Proceedings of the Nineteenth International Joint Conference on Artificial Intelligence}, Edinburgh, Scotland, 2005.

\bibitem{540-20}
P.M. Dung, R.A. Kowalski, and F. Toni. Assumption-based argumentation. In {\it Argumentation in Artificial Intelligence}, Edited by Guillermo
Simari and Iyad Rahwan, pp 199–-218, Springer US, 2009.

\bibitem{540-21}
 H. Prakken. An abstract framework for argumentation with structured arguments, {\it Argument and Computation}, 1(2):93--124, 2011.

\bibitem{540-22}
H. Barringer,  D. Gabbay and J. Woods, Temporal, Numerical and Meta-level Dynamics in Argument Networks. In {\it Argumentation and Computation}, 3(2--3), pp 143--202, 2012.

\bibitem{540-23}
 D. Gabbay. A theory of Disjunctive attacks in argumentation networks, {\it In Preparation}.

\bibitem{540-24}
D. Gabbay and M. Gabbay, The Attack as Strong Negation, Part II. {\it In Preparation}.

\bibitem{540-25}
 S.Villata, G. Boella and L .van der Torre. Attack Semantics for Abstract Argumentation. In {\em 22nd International Joint Conference on Artificial Intelligence (IJCAI 2011)}, IJCAI/AAAI, p. 406-413, 2011.

\bibitem{540-26} Serena Villata, Guido Boella, Dov M. Gabbay, Leendert van der Torre.
 Modelling defeasible and prioritized support in bipolar argumentation.
{\em Annals of Mathematics and Artificial Intelligence},
December 2012, Volume 66, Issue 1-4, pp 163-197
Date: 19 Sep 2012

\bibitem{540-27}
G\"{u}nther Charwat, Wolfgang Dvo\v{r}\'{a}, Sarah A. Gaggl, Johannes P.Wallner, and Stefan Woltran.
 Methods for solving reasoning problems in abstract argumentation Ð A survey.
 http://dx.doi.org/10.1016/j.artint.2014.11.008

 \bibitem{540-28}
  Cayrol, C., Lagasquie-Schiex, M.C.: On the acceptability of arguments in bipolar argumentation
frameworks. In: Godo, L. (ed.) {\em ECSQARU}. Lecture Notes in Computer Science, vol. 3571,
pp. 37--389. Springer (2005)
  \bibitem{540-29}
   Cayrol, C., Lagasquie-Schiex, M.C.: Coalitions of arguments: a tool for handling bipolar argumentation
frameworks. {\em Int. J. Intell. Syst.} 25(1), 83--109 (2010)

\bibitem{540-30}
David Pearce.  Equilibrium logic, {\em Annals of Mathematics and Artificial Intelligence}
June 2006, Volume 47, Issue 1-2, pp 3-41
Date: 14 Sep 2006

\bibitem{540-31}
D. Gabbay.   Modal  Provability  Foundations for Argumentation Networks ,
{\em Studia Logica}, 93(2-3): 181--198, 2009,

\bibitem{540-32}
S.  Doutre, A. Herzig and L. Perrussel.
 A Dynamic Logic Framework for Abstract Argumentation. (2014) In: {\em International Conference on Principles of Knowledge Representation and Reasoning - KR 2014}, 20--24 July 2014 (Vienna, Austria).

\bibitem{540-33}
 P.  Besnard, \'{E}.  GrŽgoire and B. Raddaoui.
A Conditional Logic-Based Argumentation Framework,
in {\em 7th International Conference on Scalable Uncertainty Management (SUM'13)}, Springer, Lecture Notes in Computer Science (LNCS)  8078, pp. 44-56, 2013

\bibitem{540-34}
 O. Arieli, and M.W. A. Caminada.
A QBF-based formalization of abstract argumentation semantics. {\em J. Applied Logic }11(2): 229-252 (2013)

\bibitem{540-35}
 W. Dvoÿr«ak, S. Szeider, and S. Woltran.  Abstract Argumentation via Monadic Second Order Logic, in E. H¬ullermeier et al. (Eds.): {\em SUM 2012}, LNAI 7520, pp. 85--98, 2012.
 Springer-Verlag Berlin Heidelberg 2012

\bibitem{540-36}
 W. Dvoÿr«ak, S. A. Gaggl,  T. Linsbichler, and J. P. Wallner
 Reduction-Based Approaches to Implement ModgilÕs Extended Argumentation Frameworks, in
T. Eiter et al. (Eds.): {\em Brewka Festschrift}, LNAI 9060, pp. 249--264, 2015.
Springer International Publishing Switzerland 2015

\bibitem{540-37} D. Grossi.  Argumentation in the View of Modal Logic, in
P. McBurney, I. Rahwan, and S. Parsons (Eds.): {\em ArgMAS 2010}, LNAI 6614, pp. 190--208, 2011.
Springer-Verlag Berlin Heidelberg 2011

\bibitem{540-38}    A. Kakas, F. Toni, and P. Mancarella.  Argumentation Logic , in  {\em  Proceeding COMMA 2014}, Editors Simon Parsons, Nir Oren, Chris Reed, Federico Cerutti,  Pages 345 - 356.   DOI10.3233/978-1-61499-436-7-345

\bibitem{540-39}
D. M. Gabbay. An equational approach to argumentation networks. {\em Argument and Computation},
special issue on the equational approach to argumentation, 3(2-3), pp. 87--142, 2012.


\bibitem{540-40} D. M. Gabbay.  What is negation as failure,  paper written in 1985, revised version  in  A. Artikis et al. (Eds.): {\em Sergot Festschrift}, LNAI 7360, pp. 52--78. Springer, Heidelberg (2012)

\bibitem{540-41}
O. Arieli. Paraconsistent reasoning and preferential entailments by signed quantified Boolean
formula. {\em ACM Transactions on Computational Logic}, 8(3), 2007. Article 18.

\bibitem{540-42}
O. Arieli and A. Avron. Reasoning with logical bilattices. {\em Journal of Logic, Language, and
Information}, 5(1):25Ð63, 1996.

\bibitem{540-43}
O. Arieli and M. Denecker. Reducing preferential paraconsistent reasoning to classical entailment.
{\em Journal of Logic and Computation}, 13(4):557--580, 2003.

\bibitem{540-44}
S. C. Kleene. {\em Introduction to metamathematics}. Van Nostrand, 1950.

\bibitem{540-45}
P. Besnard, S. Doutre and A. Herzig.
Encoding Argument Graphs in Logic.  In {\em Proceedings of the 15th International Conference on Information                Processing and Management of Uncertainty in Knowledge-Based Systems (IPMU'2014)}, A. Laurent, O. Strauss, B. Bouchon-Meunier and R. R. Yager, eds., pp. 345--354. Springer, 2014.

\bibitem{540-46}
P. Besnard and S. Doutre. Checking the acceptability of a set of arguments. In:
Delgrande, J.P., Schaub, T. (eds.) {\em Proc. NMR 2004}, pp. 59--64 (2004)

\bibitem{540-47}
P. M.  Dung. On the acceptability of arguments and its fundamental role in nonmonotonic reasoning, logic programming and n-person games.  {\em Artificial Intelligence},
77(2), 321--357 (1995)
\end{thebibliography}
\end{document}